\documentclass[a4paper,12pt]{amsart}

\usepackage{a4wide,graphicx,amssymb,amscd,amsmath,latexsym,amsbsy,amsthm,color,multicol}

\usepackage{marginnote,mathtools,epsfig,enumitem,array,calc,rotating,bm,bbm,mathrsfs} 
\usepackage{subfig}
\usepackage[libertine,cmintegrals,cmbraces,vvarbb]{newtxmath}
\usepackage{booktabs}
\usepackage{multirow}
\usepackage{breqn}
\usepackage{etoolbox}

\usepackage{tikz}

\usetikzlibrary{decorations.markings}
\usetikzlibrary{calc} 
\usetikzlibrary{arrows} 
\usetikzlibrary{arrows.meta} 
\usetikzlibrary{patterns} 
\usetikzlibrary{decorations.pathreplacing} 

\usepackage[
   hmarginratio={1:1},     
   vmarginratio={1:1},     
   textwidth=457pt,        
   textheight=630pt,       
   heightrounded,          
]{geometry}

\makeatletter
\let\old@tocline\@tocline
\let\section@tocline\@tocline

\newcommand{\subsection@dotsep}{4.5}
\newcommand{\subsubsection@dotsep}{4.5}
\patchcmd{\@tocline}
  {\hfil}
  {\nobreak
     \leaders\hbox{$\m@th
        \mkern \subsection@dotsep mu\hbox{.}\mkern \subsection@dotsep mu$}\hfill
     \nobreak}{}{}
\let\subsection@tocline\@tocline
\let\@tocline\old@tocline

\patchcmd{\@tocline}
  {\hfil}
  {\nobreak
     \leaders\hbox{$\m@th
        \mkern \subsubsection@dotsep mu\hbox{.}\mkern \subsubsection@dotsep mu$}\hfill
     \nobreak}{}{}
\let\subsubsection@tocline\@tocline
\let\@tocline\old@tocline

\let\old@l@subsection\l@subsection
\let\old@l@subsubsection\l@subsubsection

\def\@tocwriteb#1#2#3{%
  \begingroup
    \@xp\def\csname #2@tocline\endcsname##1##2##3##4##5##6{%
      \ifnum##1>\c@tocdepth
      \else \sbox\z@{##5\let\indentlabel\@tochangmeasure##6}\fi}%
    \csname l@#2\endcsname{#1{\csname#2name\endcsname}{\@secnumber}{}}%
  \endgroup
  \addcontentsline{toc}{#2}%
    {\protect#1{\csname#2name\endcsname}{\@secnumber}{#3}}}%

\newlength{\@tocsectionindent}
\newlength{\@tocsubsectionindent}
\newlength{\@tocsubsubsectionindent}
\newlength{\@tocsectionnumwidth}
\newlength{\@tocsubsectionnumwidth}
\newlength{\@tocsubsubsectionnumwidth}
\newcommand{\settocsectionnumwidth}[1]{\setlength{\@tocsectionnumwidth}{#1}}
\newcommand{\settocsubsectionnumwidth}[1]{\setlength{\@tocsubsectionnumwidth}{#1}}
\newcommand{\settocsubsubsectionnumwidth}[1]{\setlength{\@tocsubsubsectionnumwidth}{#1}}
\newcommand{\settocsectionindent}[1]{\setlength{\@tocsectionindent}{#1}}
\newcommand{\settocsubsectionindent}[1]{\setlength{\@tocsubsectionindent}{#1}}
\newcommand{\settocsubsubsectionindent}[1]{\setlength{\@tocsubsubsectionindent}{#1}}

\renewcommand{\l@section}{\section@tocline{1}{\@tocsectionvskip}{\@tocsectionindent}{\@tocsectionnumwidth}{\@tocsectionformat}}%
\renewcommand{\l@subsection}{\subsection@tocline{1}{\@tocsubsectionvskip}{\@tocsubsectionindent}{\@tocsubsectionnumwidth}{\@tocsubsectionformat}}%
\renewcommand{\l@subsubsection}{\subsubsection@tocline{1}{\@tocsubsubsectionvskip}{\@tocsubsubsectionindent}{\@tocsubsubsectionnumwidth}{\@tocsubsubsectionformat}}%
\newcommand{\@tocsectionformat}{}
\newcommand{\@tocsubsectionformat}{}
\newcommand{\@tocsubsubsectionformat}{}
\expandafter\def\csname toc@1format\endcsname{\@tocsectionformat}
\expandafter\def\csname toc@2format\endcsname{\@tocsubsectionformat}
\expandafter\def\csname toc@3format\endcsname{\@tocsubsubsectionformat}
\newcommand{\settocsectionformat}[1]{\renewcommand{\@tocsectionformat}{#1}}
\newcommand{\settocsubsectionformat}[1]{\renewcommand{\@tocsubsectionformat}{#1}}
\newcommand{\settocsubsubsectionformat}[1]{\renewcommand{\@tocsubsubsectionformat}{#1}}
\newlength{\@tocsectionvskip}
\newcommand{\settocsectionvskip}[1]{\setlength{\@tocsectionvskip}{#1}}
\newlength{\@tocsubsectionvskip}
\newcommand{\settocsubsectionvskip}[1]{\setlength{\@tocsubsectionvskip}{#1}}
\newlength{\@tocsubsubsectionvskip}
\newcommand{\settocsubsubsectionvskip}[1]{\setlength{\@tocsubsubsectionvskip}{#1}}

\patchcmd{\tocsection}{\indentlabel}{\makebox[\@tocsectionnumwidth][l]}{}{}
\patchcmd{\tocsubsection}{\indentlabel}{\makebox[\@tocsubsectionnumwidth][l]}{}{}
\patchcmd{\tocsubsubsection}{\indentlabel}{\makebox[\@tocsubsubsectionnumwidth][l]}{}{}

\newcommand{\@sectypepnumformat}{}
\renewcommand{\contentsline}[1]{%
  \expandafter\let\expandafter\@sectypepnumformat\csname @toc#1pnumformat\endcsname%
  \csname l@#1\endcsname}
\newcommand{\@tocsectionpnumformat}{}
\newcommand{\@tocsubsectionpnumformat}{}
\newcommand{\@tocsubsubsectionpnumformat}{}
\newcommand{\setsectionpnumformat}[1]{\renewcommand{\@tocsectionpnumformat}{#1}}
\newcommand{\setsubsectionpnumformat}[1]{\renewcommand{\@tocsubsectionpnumformat}{#1}}
\newcommand{\setsubsubsectionpnumformat}[1]{\renewcommand{\@tocsubsubsectionpnumformat}{#1}}
\renewcommand{\@tocpagenum}[1]{%
  \hfill {\mdseries\@sectypepnumformat #1}}

\let\oldappendix\appendix
\renewcommand{\appendix}{%
  \leavevmode\oldappendix%
  \addtocontents{toc}{%
    \protect\settowidth{\protect\@tocsectionnumwidth}{\protect\@tocsectionformat\sectionname\space}%
    \protect\addtolength{\protect\@tocsectionnumwidth}{2em}}%
}
\makeatother

\makeatletter
\settocsectionnumwidth{2em}
\settocsubsectionnumwidth{2.5em}
\settocsubsubsectionnumwidth{3em}
\settocsectionindent{1pc}%
\settocsubsectionindent{\dimexpr\@tocsectionindent+\@tocsectionnumwidth}%
\settocsubsubsectionindent{\dimexpr\@tocsubsectionindent+\@tocsubsectionnumwidth}%
\makeatother

\settocsectionvskip{0pt}
\settocsubsectionvskip{-5pt}
\settocsubsubsectionvskip{-5pt}

\settocsectionformat{\bfseries}
\settocsubsectionformat{\mdseries}
\settocsubsubsectionformat{\mdseries}
\setsectionpnumformat{\bfseries}
\setsubsectionpnumformat{\mdseries}
\setsubsubsectionpnumformat{\mdseries}

\let\oldtableofcontents\tableofcontents
\renewcommand{\tableofcontents}{%
  \vspace*{-\linespacing}
  \oldtableofcontents}

\numberwithin{equation}{section}
\theoremstyle{plain}
 
\newtheorem{thm}{Theorem}[section]

\newtheorem{lem}[thm]{Lemma}
\newtheorem{cor}[thm]{Corollary}

\theoremstyle{remark}
\newtheorem{rema}[thm]{Remark}

\newcommand{\Z}{\mathbb{Z}}

\newcommand{\C}{\mathbb{C}}
\newcommand{\V}{\mathbb{V}}

\newcommand{\ket}[1]{\left|#1\right\rangle}      
\newcommand{\bra}[1]{\left\langle #1\right|}     

\newcommand{\PROD}[3]{\mathop{\overrightarrow\prod}\limits_{#1 \le #2 \le #3 }}
\newcommand{\iPROD}[3]{\mathop{\overleftarrow\prod}\limits_{#1 \le #2 \le #3 }}
\newcommand{\gen}[1]{\mathrm{#1}}

\newcommand{\hypref}[2]{\ifx\href\asklfhas #2\else\href{#1}{#2}\fi}
\newcommand{\Secref}[1]{Section~\ref{#1}}

\newcommand{\Appref}[1]{Appendix~\ref{#1}}

\newcommand{\Figref}[1]{Figure~\ref{#1}}

\renewcommand{\eqref}[1]{(\ref{#1})}

\def\[{\begin{equation}}
\def\]{\end{equation}}
\def\<{\begin{eqnarray}}
\def\>{\end{eqnarray}}

\title[]{Functional relations in nineteen-vertex models \\ with domain-wall boundaries}
\author{A. Bossart}
\author{W. Galleas}

\address{Institut f\"ur Theoretische Physik, Eidgen\"ossische Technische Hochschule Z\"urich, Wolfgang-Pauli-Strasse 27, 8093 Z\"urich, Switzerland}
\email{albossar@student.ethz.ch}

\address{Institut f\"ur Theoretische Physik, Eidgen\"ossische Technische Hochschule Z\"urich, Wolfgang-Pauli-Strasse 27, 8093 Z\"urich, Switzerland}
\email{galleasw@phys.ethz.ch}

\subjclass[2010]{82B23; 39B32}
\keywords{Nineteen-vertex models, Izergin-Korepin model, Fateev-Zamolodchikov model, domain-wall boundaries}
\thanks{The work of W.G. is partially supported by the Swiss National Science Foundation through the NCCR SwissMAP}

\begin{document}

\begin{abstract}
This work is concerned with functional properties shared by partition functions of nineteen-vertex models with domain-wall boundary conditions. In particular, we describe both Izergin-Korepin and Fateev-Zamolodchikov models with the aforementioned boundary conditions and show their partition functions are governed by a system of functional equations originated from the associated Yang-Baxter algebra.
\end{abstract}

\maketitle

\tableofcontents

\section{Introduction} \label{sec:INTRO}

Vertex models of Statistical Mechanics can be regarded as a generalization of the ice model \cite{Bernal_Fowler_1933, Pauling_1935} proposed in the early 1930s 
aiming to describe the entropy of ice as its temperature goes to zero. In the case of ice we are actually considering 
$\mathrm{H}_2 \mathrm{O}$ molecules arranged in a crystalline structure and it is natural to suppose other molecular systems can also be described in the same way. For instance, this is the case of the $\mathrm{KH}_2 \mathrm{PO}_4$ molecule covered by the \textsc{KDP} model \cite{Slater_1941}. See also \cite{Rys_1963} and \cite{Nagle_1966} for other variants of the ice model. The aforementioned molecular structures are all particular cases of the well known six-vertex model, which turns out to be a two-dimensional abstraction of the former using concepts of \emph{graph theory}.
More accurately, the six-vertex model consists of a collection of \emph{colored graphs} embedded in a two-dimensional lattice where each vertex has degree four or one; and no loops are allowed. As for the edges, each one can then assume two distinct configurations or colors. 
By allowing each edge to assume three distinct colors, we then have the so-called nineteen-vertex model as a possible two-dimensional 
lattice system generalizing the ideas of the six-vertex model.

\subsection{Integrable nineteen-vertex model}
As a matter of fact, the denomination \emph{nineteen-vertex model} is very broad and one still needs to declare the statistical weights for each allowed graph configuration, in addition to the boundary conditions under consideration, in order to have the model fully defined. 
For instance, the previously mentioned six-vertex model is not a generic one as its statistical weights are carefully chosen in such a way that the model's partition function exhibits special properties allowing physical quantities to be computed exactly. In other words, here we are considering vertex models integrable in the sense of Baxter \cite{Baxter_book}; and this requires the model's statistical weights to satisfy the Yang-Baxter equation.

As for the symmetric six-vertex model, there is essentially only one solution of the associated Yang-Baxter equation. However, a similar uniqueness statement does not hold for generic two-dimensional vertex models. For instance, to the best of our knowledge the main representatives of nineteen-vertex models solving the Yang-Baxter equation corresponds to:
\begin{itemize}
\item Izergin-Korepin model \cite{Izergin_Korepin_1981},
\item Fateev-Zamolodchikov model \cite{Fateev_Zamolodchikov_1980};
\end{itemize}
as well as solutions based on the $q$-deformed Lie superalgebras $U_q \left[ \widehat{\mathfrak{osp}}(1|2) \right]$,
$U_q [ \widehat{\mathfrak{sl}}^{(2)}(1|2) ]$ and $U_q [ \widehat{\mathfrak{osp}}^{(2)}(2|2)]$ \cite{Bazhanov_Shadrikov_1987, Galleas_2004, Galleas_Martins_2006, Yang_Zhen_2001}.
In the present paper we shall restrict our attention to the Izergin-Korepin (IK) and Fateev-Zamolodchikov (FZ) models; and it is important to remark those models also exhibit an underlying quantum affine Lie algebra $U_q [\widehat{\mathfrak{G}}]$ \cite{Bazhanov_1984, Jimbo_1986b}. In the case of the IK model
we have $\widehat{\mathfrak{G}} = A_2^{(2)}$ while $\widehat{\mathfrak{G}} = B_1^{(1)}$ for the FZ model.

\subsection{Boundary conditions}

After having the statistical weights of a vertex model fixed, one still needs to define appropriate boundary conditions in order to having the model's partition function completely defined. Interestingly, different choices of boundary conditions not only influences the physical properties of the vertex model in the thermodynamical limit \cite{Korepin_Justin_2000} but also changes drastically the kind of mathematical problem one needs to deal with in order to obtain the sought partition function in closed form.

For instance, by choosing \emph{periodic boundary conditions} one can resort to Kramers and Wannier transfer matrix technique \cite{Kramers_1941a, Kramers_1941b}; and the evaluation of the model's partition function can be translated into the eigenvalue problem for the associated transfer matrix. As for the integrable nineteen-vertex models described above, the transfer matrix eigenvalue problem can be tackled through Tarasov's formulation of the Algebraic Bethe Ansatz \cite{Tarasov_1988}. However, there still exists several other classes of boundary conditions which render vertex models of interest from both Physics and Mathematics perspectives. For instance, among the possible choices of boundaries we have the so-called \emph{domain-wall boundary conditions}; and this is the case we give special emphasis in this work.
As for two-dimensional vertex models, domain-wall boundaries were introduced by Korepin in \cite{Korepin_1982} as a tool for studying scalar products of Bethe vectors. However, it was already realized in \cite{Korepin_1982} that this type of boundary conditions gives rise to genuine vertex models which deserves independent attention. Hence, given the rich physical and mathematical structures associated to the six-vertex model with such type of boundary conditions, it is natural to wonder if we can extend the previous studies to more sophisticated  two-dimensional vertex models. In this way, we find nineteen-vertex models to be natural targets as they, similarly to the six-vertex model, also constitute a pillar supporting hierarchies of integrable systems of Statistical Mechanics. 

\subsection{Previous results}
The literature devoted to nineteen-vertex models with domain-wall boundaries is to date quite modest when compared to the one studying the six-vertex model.
In the case of the six-vertex model, some unusual physical behavior have been noticed; and this was mainly due to Izergin's determinantal representation for the model's partition function \cite{Izergin_1987}. For instance, Izergin's formula has possibilitated the study of the influence of boundary conditions in the thermodynamical limit of the six-vertex model \cite{Korepin_Justin_2000, Bleher_2006, Bleher_2009, Bleher_2010} and the formation of limit shapes associated to spatial separation of phases \cite{Cohn_1996, Colomo_Pronko_2007}. However, it is important to remark that several other determinantal representations are also available \cite{Galleas_2016, Galleas_2016b, Galleas_2018}, as well as multiple contour integrals representations \cite{Galleas_2012, Galleas_2013}. 
In this way, it is compelling to try to extend the results available for the six-vertex model to nineteen-vertex models in order to further our understanding of the role played by boundary conditions in the thermodynamical limit of two-dimensional lattice models.

As for the FZ model with domain-wall boundaries, a determinantal representation has been obtained in \cite{Caradoc_2006} by identifying the FZ model with a spin-$1$ version of the six-vertex model. However, the problem is not that simple in the case of the IK model and a determinantal formula has been obtained in \cite{Garbali_2016} only for a special value of the anisotropy parameter.

\subsection{Our approach}
The determinantal formulae of \cite{Caradoc_2006} and \cite{Garbali_2016}, obtained respectively for the FZ model and for a special case of the IK model, result from recurrence relations satisfied by the models' partition functions. 
This recursive approach is essentially the same method originally put forward by Korepin in \cite{Korepin_1982} for the six-vertex model; ultimately leading to Izergin's representation \cite{Izergin_1987}.
However, one inherent step of this approach is making an \emph{educated guess} for the sought partition function; which can then be shown to correspond to the actual partition function if it satisfies the aforementioned recurrence relations in addition to extra properties. In this way, the construction of such determinantal representations can elude us in more sophisticated models.

An alternative method based on functional equations was put forward in \cite{Galleas_2010} and subsequently refined in the series of works \cite{Galleas_2011, Galleas_2012, Galleas_2013, Galleas_2016a, Galleas_2016b}. We shall refer to this approach as Algebraic-Functional (AF) method and it has been responsible, among other results, to the construction of single determinant representations for the elliptic solid-on-solid model with domain-wall boundaries \cite{Galleas_2016a, Galleas_2016b}; which were previously thought to not admit such type of representations.
Hence, given the above described scenario, the extension of the AF method to nineteen-vertex models with domain-wall boundaries is a sound problem and it is the main goal of the present paper.

\subsection{Outline} 
We have organized this paper as follows. In \Secref{sec:19V} we describe the algebraic formulation of integrable nineteen-vertex models, with special emphasis to the IK and FZ models as they are the specific vertex models we will be considering in the present work. In \Secref{sec:19V} we also precise the boundary conditions 
relevant to our forthcoming analysis and present properties expected from the models' partition functions. \Secref{sec:AFM} is then devoted to the formulation of the AF method to both IK and FZ models with domain-wall boundaries in an unified way. Functional equations governing our models' partition functions are then derived and inspected in \Secref{sec:PROP}. In particular, in \Secref{sec:PROP} we also discuss their \emph{strength} in characterizing the sought partition functions. \Secref{sec:CONCL} is
then left for concluding remarks and technical details and  extra results are presented in the appendices.

\section{Nineteen-vertex models} \label{sec:19V}

This work is concerned with nineteen-vertex models with particular domain-wall boundary conditions; however, our analysis will require a more general formulation allowing also for other choices of boundaries. In this way, we shall start this section with the introduction of conventions and concepts which will assist us throughout the next sections.

Write $[n] \coloneqq \{0, 1, \dots, n, n+1 \}$ and let $\mathrm{Rect}_{K,L} \coloneqq [K] \times [L] \subseteq \Z_{\geq 0}^2$ denote a two-dimensional
lattice with bulk grid formed by the crossing of $K$ rows and $L$ columns. Also, let
\[
\mathcal{G} =   \bigcup_{\substack{i \in [K] \backslash \{0, K+1 \} \\ j \in [L] \backslash \{0, L+1 \} }} g_{i,j}
\]
be a graph built from the juxtaposition of local subgraphs $g_{i,j}$. The latter consists of $g_{i,j} = \left( \mathcal{V}_{i,j}, \mathcal{E}_{i,j}   \right)$
with vertices $\mathcal{V}_{i,j} = \{ v_{i,j} , v_{i,j-1} , v_{i,j+1} , v_{i-1,j} , v_{i+1,j} \}$
and edges 
\[ \label{edges}
\mathcal{E}_{i,j} = \{ d_{v_{i,j}}( v_{i,j-1}) , d_{v_{i,j}}( v_{i,j+1} ) , d_{v_{i,j}}( v_{i-1,j} ) , d_{v_{i,j}}( v_{i+1,j}) \} \; .
\]
In \eqref{edges} we have used  $d_{v_{i,j}}( v_{k,l} ) = d_{v_{k,l}}( v_{i,j} )$ to denote the edge connecting generic vertices $v_{i,j}$ and $v_{k,l}$. We then embed $\mathcal{G}$ on $\mathrm{Rect}_{K,L}$ by identifying $v_{i,j}$ with $(i,j) \in \mathrm{Rect}_{K,L}$.

Next we would like to promote $\mathcal{G}$ to an \emph{edge-colored} graph $\mathcal{G}^{*}$ obtained through the assignment 
$d_{v_{i,j}}( v_{k,l} ) \mapsto d_{v_{i,j}}^{(\alpha)} ( v_{k,l} )$ for all edges in $\mathcal{G}$. The label $\alpha$ is then introduced to characterize
the color or configuration assigned to a given edge. 
Here we are interested in the so-called nineteen-vertex models and, in that case, each edge $d_{v_{i,j}}^{(\alpha)} ( v_{k,l} )$ in $\mathcal{G}^{*}$
can take on three distinct configurations. For instance, we shall write $\alpha = 1, 2, 3$ and use respectively 
$\substack{
\begin{tikzpicture}
\draw [postaction=decorate,decoration={markings, mark=at position 0.6cm with {{\arrow[black]{Triangle}}}}, thick]  (2,0) -- (1,0);
\end{tikzpicture} \\ }
$, 
$\substack{
\begin{tikzpicture}
\draw [dashed, thick]  (1,0) -- (2,0);
\end{tikzpicture} \\ }
$
and $\substack{
\begin{tikzpicture}
\draw [postaction=decorate,decoration={markings, mark=at position 0.6cm with {\arrow[black]{Triangle}}}, thick]  (1,0) -- (2,0);
\end{tikzpicture} \\ }
$
to depict the corresponding horizontal edges. Similarly, we use \hspace{0.1cm} $\substack{
\begin{tikzpicture}
\draw [postaction=decorate,decoration={markings, mark=at position 0.45cm with {\arrow[black]{Triangle}}}, thick]  (0,0) -- (0,0.7);
\end{tikzpicture} \\}
$ \hspace{0.1cm}, \hspace{0.1cm}
$\substack{
\begin{tikzpicture}
\draw [dashed, thick]  (0,0) -- (0,0.7);
\end{tikzpicture} \\}
$ 
\hspace{0.1cm} and \hspace{0.1cm} $\substack{
\begin{tikzpicture}[>=stealth]
\draw [postaction=decorate,decoration={markings, mark=at position 0.45cm with {\arrow[black]{Triangle}}}, thick]  (0,0.7) -- (0,0);
\end{tikzpicture} \\}
$ \hspace{0.1cm} to illustrate vertical edges associated respectively to $\alpha= 1, 2, 3$.
Moreover, in the case of nineteen-vertex models, we restrict the number of possible edge-colored graphs $g_{i,j}$ to nineteen among the 
$3^4 = 81$ possibilities. The allowed graphs $g_{i,j}$ are then depicted in \Figref{fig:19v}.

\begin{rema}[Conservation of arrows] \label{AR}
The diagrammatic representations collected in \Figref{fig:19v} makes manifest an important \emph{conservation law} in nineteen-vertex models. For instance, one can readily see in \Figref{fig:19v} that all graphs $g_{i,j}$ have the same number of arrows pointing inwards and outwards. Here we refer to this rule as \emph{conservation of arrows}.
\end{rema}

\begin{figure} \centering
\scalebox{0.85}{
\begin{tikzpicture}
\path (0,0) node[circle, fill=black, scale=0.3](C) {}
      (-1,0) node[circle, fill=black, scale=0.3](W) {}
      (1,0) node[circle, fill=black, scale=0.3](E) {}
      (0,-1) node[circle, fill=black, scale=0.3](S) {}
      (0,1) node[circle, fill=black, scale=0.3](N) {};
\draw [postaction=decorate,decoration={markings, mark=at position 0.5cm with {\arrow[black]{Triangle}}}, thick]  (C) -- (W);
\draw [postaction=decorate,decoration={markings, mark=at position 0.5cm with {\arrow[black]{Triangle}}}, thick]  (C) -- (E);
\draw [postaction=decorate,decoration={markings, mark=at position 0.5cm with {\arrow[black]{Triangle}}}, thick]  (N) -- (C);
\draw [postaction=decorate,decoration={markings, mark=at position 0.5cm with {\arrow[black]{Triangle}}}, thick]  (S) -- (C);

\begin{scope}[xshift=2.5cm]
\path (0,0) node[circle, fill=black, scale=0.3](C) {}
      (-1,0) node[circle, fill=black, scale=0.3](W) {}
      (1,0) node[circle, fill=black, scale=0.3](E) {}
      (0,-1) node[circle, fill=black, scale=0.3](S) {}
      (0,1) node[circle, fill=black, scale=0.3](N) {};
\draw [postaction=decorate,decoration={markings, mark=at position 0.5cm with {\arrow[black]{Triangle}}}, thick]  (W) -- (C);
\draw [postaction=decorate,decoration={markings, mark=at position 0.5cm with {\arrow[black]{Triangle}}}, thick]  (E) -- (C);
\draw [postaction=decorate,decoration={markings, mark=at position 0.5cm with {\arrow[black]{Triangle}}}, thick]  (C) -- (N);
\draw [postaction=decorate,decoration={markings, mark=at position 0.5cm with {\arrow[black]{Triangle}}}, thick]  (C) -- (S);
\end{scope}

\begin{scope}[xshift=5cm]
\path (0,0) node[circle, fill=black, scale=0.3](C) {}
      (-1,0) node[circle, fill=black, scale=0.3](W) {}
      (1,0) node[circle, fill=black, scale=0.3](E) {}
      (0,-1) node[circle, fill=black, scale=0.3](S) {}
      (0,1) node[circle, fill=black, scale=0.3](N) {};
\draw [postaction=decorate,decoration={markings, mark=at position 0.5cm with {\arrow[black]{Triangle}}}, thick]  (W) -- (C);
\draw [postaction=decorate,decoration={markings, mark=at position 0.5cm with {\arrow[black]{Triangle}}}, thick]  (C) -- (E);
\draw [postaction=decorate,decoration={markings, mark=at position 0.5cm with {\arrow[black]{Triangle}}}, thick]  (N) -- (C);
\draw [postaction=decorate,decoration={markings, mark=at position 0.5cm with {\arrow[black]{Triangle}}}, thick]  (C) -- (S);
\end{scope}

\begin{scope}[xshift=7.5cm]
\path (0,0) node[circle, fill=black, scale=0.3](C) {}
      (-1,0) node[circle, fill=black, scale=0.3](W) {}
      (1,0) node[circle, fill=black, scale=0.3](E) {}
      (0,-1) node[circle, fill=black, scale=0.3](S) {}
      (0,1) node[circle, fill=black, scale=0.3](N) {};
\draw [postaction=decorate,decoration={markings, mark=at position 0.5cm with {\arrow[black]{Triangle}}}, thick]  (C) -- (W);
\draw [postaction=decorate,decoration={markings, mark=at position 0.5cm with {\arrow[black]{Triangle}}}, thick]  (E) -- (C);
\draw [postaction=decorate,decoration={markings, mark=at position 0.5cm with {\arrow[black]{Triangle}}}, thick]  (C) -- (N);
\draw [postaction=decorate,decoration={markings, mark=at position 0.5cm with {\arrow[black]{Triangle}}}, thick]  (S) -- (C);
\end{scope}

\begin{scope}[xshift=10cm]
\path (0,0) node[circle, fill=black, scale=0.3](C) {}
      (-1,0) node[circle, fill=black, scale=0.3](W) {}
      (1,0) node[circle, fill=black, scale=0.3](E) {}
      (0,-1) node[circle, fill=black, scale=0.3](S) {}
      (0,1) node[circle, fill=black, scale=0.3](N) {};
\draw [postaction=decorate,decoration={markings, mark=at position 0.5cm with {\arrow[black]{Triangle}}}, thick]  (C) -- (W);
\draw [postaction=decorate,decoration={markings, mark=at position 0.5cm with {\arrow[black]{Triangle}}}, thick]  (E) -- (C);
\draw [postaction=decorate,decoration={markings, mark=at position 0.5cm with {\arrow[black]{Triangle}}}, thick]  (N) -- (C);
\draw [postaction=decorate,decoration={markings, mark=at position 0.5cm with {\arrow[black]{Triangle}}}, thick]  (C) -- (S);
\end{scope}

\begin{scope}[xshift=12.5cm]
\path (0,0) node[circle, fill=black, scale=0.3](C) {}
      (-1,0) node[circle, fill=black, scale=0.3](W) {}
      (1,0) node[circle, fill=black, scale=0.3](E) {}
      (0,-1) node[circle, fill=black, scale=0.3](S) {}
      (0,1) node[circle, fill=black, scale=0.3](N) {};
\draw [postaction=decorate,decoration={markings, mark=at position 0.5cm with {\arrow[black]{Triangle}}}, thick]  (W) -- (C);
\draw [postaction=decorate,decoration={markings, mark=at position 0.5cm with {\arrow[black]{Triangle}}}, thick]  (C) -- (E);
\draw [postaction=decorate,decoration={markings, mark=at position 0.5cm with {\arrow[black]{Triangle}}}, thick]  (C) -- (N);
\draw [postaction=decorate,decoration={markings, mark=at position 0.5cm with {\arrow[black]{Triangle}}}, thick]  (S) -- (C);
\end{scope}

\begin{scope}[xshift=15cm]
\path (0,0) node[circle, fill=black, scale=0.3](C) {}
      (-1,0) node[circle, fill=black, scale=0.3](W) {}
      (1,0) node[circle, fill=black, scale=0.3](E) {}
      (0,-1) node[circle, fill=black, scale=0.3](S) {}
      (0,1) node[circle, fill=black, scale=0.3](N) {};
\draw [postaction=decorate,decoration={markings, mark=at position 0.5cm with {\arrow[black]{Triangle}}}, thick]  (W) -- (C);
\draw [postaction=decorate,decoration={markings, mark=at position 0.5cm with {\arrow[black]{Triangle}}}, thick]  (C) -- (E);
\draw [dashed, thick]  (N) -- (C);
\draw [dashed, thick]  (C) -- (S);
\end{scope}

\begin{scope}[yshift=-2.5cm]
\path (0,0) node[circle, fill=black, scale=0.3](C) {}
      (-1,0) node[circle, fill=black, scale=0.3](W) {}
      (1,0) node[circle, fill=black, scale=0.3](E) {}
      (0,-1) node[circle, fill=black, scale=0.3](S) {}
      (0,1) node[circle, fill=black, scale=0.3](N) {};
\draw [postaction=decorate,decoration={markings, mark=at position 0.5cm with {\arrow[black]{Triangle}}}, thick]  (C) -- (W);
\draw [postaction=decorate,decoration={markings, mark=at position 0.5cm with {\arrow[black]{Triangle}}}, thick]  (E) -- (C);
\draw [dashed, thick]  (N) -- (C);
\draw [dashed, thick]  (C) -- (S);
\end{scope}

\begin{scope}[xshift=2.5cm, yshift=-2.5cm]
\path (0,0) node[circle, fill=black, scale=0.3](C) {}
      (-1,0) node[circle, fill=black, scale=0.3](W) {}
      (1,0) node[circle, fill=black, scale=0.3](E) {}
      (0,-1) node[circle, fill=black, scale=0.3](S) {}
      (0,1) node[circle, fill=black, scale=0.3](N) {};
\draw [postaction=decorate,decoration={markings, mark=at position 0.5cm with {\arrow[black]{Triangle}}}, thick]  (C) -- (N);
\draw [postaction=decorate,decoration={markings, mark=at position 0.5cm with {\arrow[black]{Triangle}}}, thick]  (S) -- (C);
\draw [dashed, thick]  (W) -- (C);
\draw [dashed, thick]  (C) -- (E);
\end{scope}

\begin{scope}[xshift=5cm, yshift=-2.5cm]
\path (0,0) node[circle, fill=black, scale=0.3](C) {}
      (-1,0) node[circle, fill=black, scale=0.3](W) {}
      (1,0) node[circle, fill=black, scale=0.3](E) {}
      (0,-1) node[circle, fill=black, scale=0.3](S) {}
      (0,1) node[circle, fill=black, scale=0.3](N) {};
\draw [postaction=decorate,decoration={markings, mark=at position 0.5cm with {\arrow[black]{Triangle}}}, thick]  (N) -- (C);
\draw [postaction=decorate,decoration={markings, mark=at position 0.5cm with {\arrow[black]{Triangle}}}, thick]  (C) -- (S);
\draw [dashed, thick]  (W) -- (C);
\draw [dashed, thick]  (C) -- (E);
\end{scope}

\begin{scope}[xshift=7.5cm, yshift=-2.5cm]
\path (0,0) node[circle, fill=black, scale=0.3](C) {}
      (-1,0) node[circle, fill=black, scale=0.3](W) {}
      (1,0) node[circle, fill=black, scale=0.3](E) {}
      (0,-1) node[circle, fill=black, scale=0.3](S) {}
      (0,1) node[circle, fill=black, scale=0.3](N) {};
\draw [postaction=decorate,decoration={markings, mark=at position 0.5cm with {\arrow[black]{Triangle}}}, thick]  (W) -- (C);
\draw [postaction=decorate,decoration={markings, mark=at position 0.5cm with {\arrow[black]{Triangle}}}, thick]  (C) -- (N);
\draw [dashed, thick]  (C) -- (E);
\draw [dashed, thick]  (S) -- (C);
\end{scope}

\begin{scope}[xshift=10cm, yshift=-2.5cm]
\path (0,0) node[circle, fill=black, scale=0.3](C) {}
      (-1,0) node[circle, fill=black, scale=0.3](W) {}
      (1,0) node[circle, fill=black, scale=0.3](E) {}
      (0,-1) node[circle, fill=black, scale=0.3](S) {}
      (0,1) node[circle, fill=black, scale=0.3](N) {};
\draw [postaction=decorate,decoration={markings, mark=at position 0.5cm with {\arrow[black]{Triangle}}}, thick]  (C) -- (W);
\draw [postaction=decorate,decoration={markings, mark=at position 0.5cm with {\arrow[black]{Triangle}}}, thick]  (N) -- (C);
\draw [dashed, thick]  (C) -- (E);
\draw [dashed, thick]  (S) -- (C);
\end{scope}

\begin{scope}[xshift=12.5cm, yshift=-2.5cm]
\path (0,0) node[circle, fill=black, scale=0.3](C) {}
      (-1,0) node[circle, fill=black, scale=0.3](W) {}
      (1,0) node[circle, fill=black, scale=0.3](E) {}
      (0,-1) node[circle, fill=black, scale=0.3](S) {}
      (0,1) node[circle, fill=black, scale=0.3](N) {};
\draw [postaction=decorate,decoration={markings, mark=at position 0.5cm with {\arrow[black]{Triangle}}}, thick]  (C) -- (N);
\draw [postaction=decorate,decoration={markings, mark=at position 0.5cm with {\arrow[black]{Triangle}}}, thick]  (E) -- (C);
\draw [dashed, thick]  (W) -- (C);
\draw [dashed, thick]  (S) -- (C);
\end{scope}

\begin{scope}[xshift=15cm, yshift=-2.5cm]
\path (0,0) node[circle, fill=black, scale=0.3](C) {}
      (-1,0) node[circle, fill=black, scale=0.3](W) {}
      (1,0) node[circle, fill=black, scale=0.3](E) {}
      (0,-1) node[circle, fill=black, scale=0.3](S) {}
      (0,1) node[circle, fill=black, scale=0.3](N) {};
\draw [postaction=decorate,decoration={markings, mark=at position 0.5cm with {\arrow[black]{Triangle}}}, thick]  (N) -- (C);
\draw [postaction=decorate,decoration={markings, mark=at position 0.5cm with {\arrow[black]{Triangle}}}, thick]  (C) -- (E);
\draw [dashed, thick]  (W) -- (C);
\draw [dashed, thick]  (S) -- (C);
\end{scope}

\begin{scope}[xshift=2.5cm, yshift=-5cm]
\path (0,0) node[circle, fill=black, scale=0.3](C) {}
      (-1,0) node[circle, fill=black, scale=0.3](W) {}
      (1,0) node[circle, fill=black, scale=0.3](E) {}
      (0,-1) node[circle, fill=black, scale=0.3](S) {}
      (0,1) node[circle, fill=black, scale=0.3](N) {};
\draw [postaction=decorate,decoration={markings, mark=at position 0.5cm with {\arrow[black]{Triangle}}}, thick]  (W) -- (C);
\draw [postaction=decorate,decoration={markings, mark=at position 0.5cm with {\arrow[black]{Triangle}}}, thick]  (C) -- (S);
\draw [dashed, thick]  (N) -- (C);
\draw [dashed, thick]  (C) -- (E);
\end{scope}

\begin{scope}[xshift=5cm, yshift=-5cm]
\path (0,0) node[circle, fill=black, scale=0.3](C) {}
      (-1,0) node[circle, fill=black, scale=0.3](W) {}
      (1,0) node[circle, fill=black, scale=0.3](E) {}
      (0,-1) node[circle, fill=black, scale=0.3](S) {}
      (0,1) node[circle, fill=black, scale=0.3](N) {};
\draw [postaction=decorate,decoration={markings, mark=at position 0.5cm with {\arrow[black]{Triangle}}}, thick]  (C) -- (W);
\draw [postaction=decorate,decoration={markings, mark=at position 0.5cm with {\arrow[black]{Triangle}}}, thick]  (S) -- (C);
\draw [dashed, thick]  (N) -- (C);
\draw [dashed, thick]  (C) -- (E);
\end{scope}

\begin{scope}[xshift=7.5cm, yshift=-5cm]
\path (0,0) node[circle, fill=black, scale=0.3](C) {}
      (-1,0) node[circle, fill=black, scale=0.3](W) {}
      (1,0) node[circle, fill=black, scale=0.3](E) {}
      (0,-1) node[circle, fill=black, scale=0.3](S) {}
      (0,1) node[circle, fill=black, scale=0.3](N) {};
\draw [postaction=decorate,decoration={markings, mark=at position 0.5cm with {\arrow[black]{Triangle}}}, thick]  (S) -- (C);
\draw [postaction=decorate,decoration={markings, mark=at position 0.5cm with {\arrow[black]{Triangle}}}, thick]  (C) -- (E);
\draw [dashed, thick]  (W) -- (C);
\draw [dashed, thick]  (C) -- (N);
\end{scope}

\begin{scope}[xshift=10cm, yshift=-5cm]
\path (0,0) node[circle, fill=black, scale=0.3](C) {}
      (-1,0) node[circle, fill=black, scale=0.3](W) {}
      (1,0) node[circle, fill=black, scale=0.3](E) {}
      (0,-1) node[circle, fill=black, scale=0.3](S) {}
      (0,1) node[circle, fill=black, scale=0.3](N) {};
\draw [postaction=decorate,decoration={markings, mark=at position 0.5cm with {\arrow[black]{Triangle}}}, thick]  (C) -- (S);
\draw [postaction=decorate,decoration={markings, mark=at position 0.5cm with {\arrow[black]{Triangle}}}, thick]  (E) -- (C);
\draw [dashed, thick]  (W) -- (C);
\draw [dashed, thick]  (C) -- (N);
\end{scope}

\begin{scope}[xshift=12.5cm, yshift=-5cm]
\path (0,0) node[circle, fill=black, scale=0.3](C) {}
      (-1,0) node[circle, fill=black, scale=0.3](W) {}
      (1,0) node[circle, fill=black, scale=0.3](E) {}
      (0,-1) node[circle, fill=black, scale=0.3](S) {}
      (0,1) node[circle, fill=black, scale=0.3](N) {};
\draw [dashed, thick]  (S) -- (C);
\draw [dashed, thick]  (C) -- (E);
\draw [dashed, thick]  (W) -- (C);
\draw [dashed, thick]  (C) -- (N);
\end{scope}
\end{tikzpicture}}
\caption{Graphs $g_{i,j}$ in nineteen-vertex models.}
\label{fig:19v}
\end{figure}
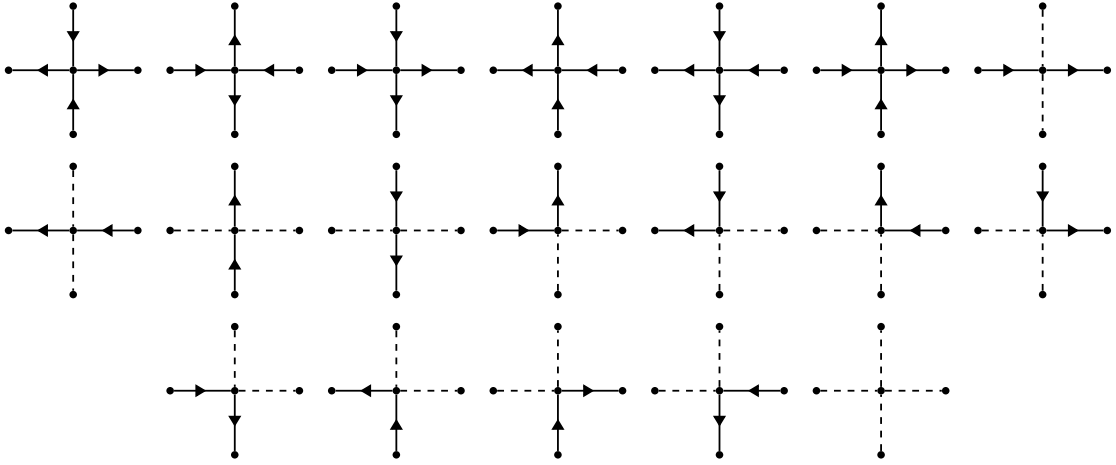

At the end of the day one would like to associate a partition function to the graph $\mathcal{G}^{*}$ embedded on $\mathrm{Rect}_{K,L}$.
That will additionally require the introduction of boundary conditions and statistical weights for local graph configurations $g_{i,j}$. We shall return to this issue in the following subsections.

\subsection{Algebraic formulation} \label{sec:ALG}

Write $\mathcal{R} (\lambda_i , \mu_j)^{\alpha_{i,j+1} , \beta_{i,j}}_{\alpha_{i,j} , \beta_{i+1,j}}$ for the statistical weight associated to the local edge-colored graph $g_{i,j}$ as shown in \Figref{fig:gij}.
\begin{figure} \centering
\scalebox{1.3}{
\begin{tikzpicture}
\path (0,0) node[circle, fill=black, scale=0.25](C) {}
      (-1,0) node[circle, fill=black, scale=0.25](W) {}
      (1,0) node[circle, fill=black, scale=0.25](E) {}
      (0,-1) node[circle, fill=black, scale=0.25](S) {}
      (0,1) node[circle, fill=black, scale=0.25](N) {};
      
 \node[scale=0.7,above,left] at (W) {{$\scriptscriptstyle d_{v_{i,j}}^{(\alpha_{i,j})} (v_{i,j-1})$}};
 \node[scale=0.7,above,right] at (E) {{$\scriptscriptstyle d_{v_{i,j}}^{(\alpha_{i,j+1})} (v_{i,j+1})$}};
 \node[scale=0.7,above] at (N) {{$\scriptscriptstyle d_{v_{i,j}}^{(\beta_{i+1,j})} (v_{i+1,j})$}};
 \node[scale=0.7,below] at (S) {{$\scriptscriptstyle d_{v_{i,j}}^{(\beta_{i,j})} (v_{i-1,j})$}};
\draw [thick]  (C) -- (W);
\draw [thick]  (C) -- (E);
\draw [thick]  (N) -- (C);
\draw [thick]  (S) -- (C);
\end{tikzpicture}}
\caption{Local edge-colored graph $g_{i,j}$.}
\label{fig:gij}
\end{figure}
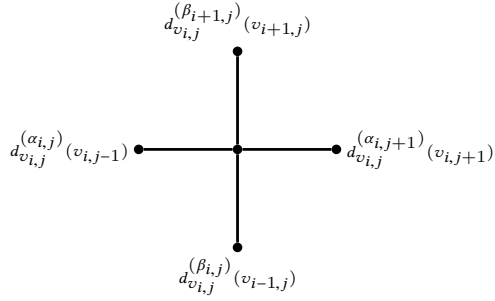
Also, let us introduce vectors
\<
\vec{\alpha}_j &\coloneqq& \left( \alpha_{1,j}, \alpha_{2,j}, \dots , \alpha_{K,j}   \right) \nonumber \\
\vec{\beta}_{i} &\coloneqq& \left( \beta_{i,1}, \beta_{i,2}, \dots,   \beta_{i,L} \right) \; .
\>
In this way, we can define a partition function with fixed boundary conditions for $\mathcal{G}^{*}$ on $\mathrm{Rect}_{K,L}$ as
\< \label{PF}
Z_{\vec{\beta}_0}^{\vec{\beta}_{K+1}} (\vec{\alpha}_0 \mid \vec{\alpha}_{L+1}) \coloneqq \sum_{\alpha_{i,j}, \beta_{i,j} \in \{1,2,3 \}}
\prod_{\substack{i \in [K] \backslash \{0, K+1 \} \\ j \in [L] \backslash \{0, L+1 \} }} \mathcal{R} (\lambda_i , \mu_j)^{\alpha_{i,j+1} , \beta_{i,j}}_{\alpha_{i,j} , \beta_{i+1,j}} \; .
\>
The RHS of \eqref{PF} looks overwhelming at first sight but, fortunately, it can be rewritten in an operatorial manner along the lines of Kramers and Wannier transfer matrix technique. 
In order to present such operatorial formulation, let us introduce vector spaces
$\V = \V_a = \V_i \simeq \C^3$ for $i=1,2, \dots, L$ and let $\{e_1 , e_2 , e_3 \}$ be standard basis vectors of $\C^3$. More precisely, we take
\<
e_1 \coloneqq \begin{pmatrix} 1 \\ 0 \\ 0 \end{pmatrix} \quad , \quad e_2 \coloneqq \begin{pmatrix} 0 \\ 1 \\ 0 \end{pmatrix} \quad \text{and} \quad
e_3 \coloneqq \begin{pmatrix} 0 \\ 0 \\ 1 \end{pmatrix} \; .
\>
Also, write $E_{\alpha, \alpha'} \in \mathrm{End}( \C^3 )$ for unit matrices defined by $E_{\alpha, \alpha'} (e_{\beta}) \coloneqq \delta_{\alpha', \beta} \; e_{\alpha}$ for  $\beta = 1, 2, 3$. 
Next we define the matrix $\mathcal{R} \colon \C \times \C \to \mathrm{End}( \V \otimes \V)$ as
\[ \label{rmat}
\mathcal{R}(\lambda_i, \mu_j) \coloneqq \sum_{\substack{\alpha, \alpha' \in \{1,2,3 \} \\ \beta , \beta' \in \{1,2,3 \} }} \mathcal{R}(\lambda_i, \mu_j)_{\alpha, \beta}^{\alpha' , \beta'}  \; E_{\alpha, \alpha'} \otimes E_{\beta, \beta'} \; .
\] 

\begin{rema} \label{AR1}
The conservation of arrows pointed out in Remark \ref{AR} reflects   in the $\mathcal{R}$-matrix formalism \eqref{rmat} by only allowing non-vanishing statistical weights $\mathcal{R}_{\alpha, \beta}^{\alpha' , \beta'}$ such that $\alpha + \beta = \alpha' + \beta'$.
\end{rema}

Here we intend to express the partition function \eqref{PF} in terms of the $\mathcal{R}$-matrix \eqref{rmat}. With that goal in mind we then introduce the so-called \emph{monodromy matrix} $\mathcal{T} \colon \C \times \C^L \to \mathrm{End} (\V_a \otimes \V_{\mathcal{Q}})$ with 
$\V_{\mathcal{Q}} \coloneqq \bigotimes_{i=1}^L \V_i$. More precisely, we write
\[ \label{mono}
\mathcal{T}(\lambda_i \mid \{ \mu_j \}) \coloneqq \PROD{1}{j}{L} \mathcal{R}_{a j} (\lambda_i, \mu_j) 
\quad \in \; \mathrm{End}(\V_a \otimes \V_1 \otimes \otimes \dots \otimes \V_L)
\]
using tensor leg notation. The monodromy matrix $\mathcal{T}$ can also be regarded as matrix in $\mathrm{End}(\V_a)$ with entries in 
$\mathrm{End}(\V_{\mathcal{Q}})$. In this way, we also have
\< \label{abcd}
\mathcal{T}(\lambda \mid \{ \mu_j \}) \eqqcolon \begin{pmatrix}
\mathcal{A}_1 (\lambda) & \mathcal{B}_1 (\lambda) & \mathcal{B}_2 (\lambda) \\
\mathcal{C}_1 (\lambda) & \mathcal{A}_2 (\lambda) & \mathcal{B}_3 (\lambda) \\
\mathcal{C}_2 (\lambda) & \mathcal{C}_3 (\lambda) & \mathcal{A}_3 (\lambda) 
\end{pmatrix} \; ,
\>
deliberately omitting the dependence on parameters $\mu_j \in \C$ in the RHS. In what follows we shall then use the notation 
$\mathcal{T}_{\alpha}^{\beta}$ to refer to the entry of \eqref{abcd} corresponding to the element $E_{\alpha, \beta} \in \mathrm{End}(\V_a)$.   
Next we define vectors
\[
\ket{\vec{\beta}_{i}} \coloneqq \bigotimes_{j=1}^L e_{\beta_{i,j}} \qquad \in \; \V_{\mathcal{Q}}
\]
completing, in this way, the ingredients required to reformulate \eqref{PF}. Then, using \eqref{rmat}-\eqref{abcd}, we can rewrite our partition function
with fixed boundary conditions in terms of entries of the monodromy matrix $\mathcal{T}$ as
\< \label{PFT}
Z_{\vec{\beta}_0}^{\vec{\beta}_{K+1}} (\vec{\alpha}_0 \mid \vec{\alpha}_{L+1}) = \bra{\vec{\beta}_{K+1}} \iPROD{1}{i}{K} \mathcal{T}(\lambda_i \mid \{ \mu_j \})_{\alpha_{i,0}}^{\alpha_{i,L+1}} \ket{\vec{\beta}_{0}} \; .
\>

The statistical weights associated to configurations of graphs $g_{i,j}$ are encoded in the $\mathcal{R}$-matrix \eqref{rmat}. 
Although they are still generic up to this point, integrability in the sense of Baxter requires the $\mathcal{R}$-matrix \eqref{rmat}
to satisfy the Yang-Baxter equation. 
The following discussion will then be restricted to integrable nineteen-vertex models and, in that case, we can consider
$\mathcal{R}(\lambda, \mu) = \mathcal{R}(\lambda - \mu)$ and use the convention 
\< \label{bw}
\mathcal{R}(\lambda) = \begin{pmatrix}
a(\lambda) & & & & & & & &  \\
 & b(\lambda) & & c(\lambda) & & & & &  \\
 & & d_{1,1}(\lambda) & & d_{1,2}(\lambda) & & d_{1,3}(\lambda) & &  \\
 & \bar{c}(\lambda) & & b(\lambda) & & & & &  \\
 & & d_{2,1}(\lambda) & & d_{2,2}(\lambda) & & d_{2,3}(\lambda) & &  \\
 & & & & & b(\lambda) & & c(\lambda) &  \\
 & & d_{3,1}(\lambda) & & d_{3,2}(\lambda) & & d_{3,3}(\lambda) & &  \\
 & & & & & \bar{c}(\lambda) & & b(\lambda) &  \\
 & & & & & & & & a(\lambda) 
\end{pmatrix} 
\>
in order to ease our presentation. In the next subsections we shall then discuss two distinct sets of statistical weights satisfying the Yang-Baxter equation.

\subsection{The IK and FZ models} \label{sec:IKFZ}

Strictly speaking, integrability in Statistical Mechanics is not a well defined concept as it is in Classical Mechanics 
\cite{Caux_2011}. Nevertheless, Baxter's concept of commuting transfer matrices \cite{Baxter_book} has played a major role in identifying two-dimensional lattice models whose physical properties can be computed exactly. Here we will be considering  nineteen-vertex models integrable in the sense of Baxter; and this requires the $\mathcal{R}$-matrix \eqref{bw} to satisfy the Yang-Baxter equation. More precisely, we will focus on statistical weights $a$, $b$, $c$, $\bar{c}$ and $d_{i,j}$
constrained by
\< \label{ybe}
&&\mathcal{R}_{12} (\lambda_1 - \lambda_2) \mathcal{R}_{13} (\lambda_1 - \lambda_3) \mathcal{R}_{23} (\lambda_2 - \lambda_3) = \nonumber \\
&&\qquad \qquad \qquad \mathcal{R}_{23} (\lambda_2 - \lambda_3) \mathcal{R}_{13} (\lambda_1 - \lambda_3) \mathcal{R}_{12} (\lambda_1 - \lambda_2) 
\>
in $\mathrm{End}(\V_1 \otimes \V_2 \otimes \V_3)$. In contrast to the six-vertex model, there are several solutions of \eqref{ybe} corresponding to nineteen-vertex models. In what follows we shall describe two of them, namely the Izergin-Korepin (IK) and the Fateev-Zamolodchikov (FZ) models.

The IK model originally appeared as the quantization of integrable structures associated to the Shabat-Mikhailov model \cite{Izergin_Korepin_1981}. The latter is a relativistic field theory in $1+1$ dimensions whose integrability, in the classical sense, is ensured by the existence of a Lax pair and a classical $r$-matrix. The $\mathcal{R}$-matrix of the IK model then arises as the \emph{quantization} of the aforementioned classical $r$-matrix.
On the other hand, the FZ model firstly appeared within the context of factorized scattering \cite{Fateev_Zamolodchikov_1980}. More precisely, the $\mathcal{R}$-matrix of the FZ model was originally obtained as the $\mathcal{S}$-matrix of a quantum field theory enjoying $\mathrm{C}$, $\mathrm{P}$, $\mathrm{T}$ and $\mathrm{U}(1)$ symmetries. Those symmetries are able to fix the $\mathcal{S}$-matrix up to a large extent, while the remaining part is then fixed by the Yang-Baxter equation.
 
The quantum group structure underlying the $\mathcal{R}$-matrices of the IK and FZ models was only later on unveiled in 
\cite{Jimbo_1986b}. They correspond respectively to the quantum affine Lie algebras $U_q [\widehat{A}_2^{(2)}]$ and $U_q [\widehat{B}_1^{(1)}]$; and this algebraic structure also allows the associated statistical weights to be presented in an unified manner. In this way, we have
 \begin{align} \label{IKFZa}
a(\lambda) &=  ( e^{2 \lambda} - \zeta) ( e^{2 \lambda} - q^2) & b(\lambda) &=  q ( e^{2 \lambda} - 1) ( e^{2 \lambda} - \zeta) \nonumber \\ 
c(\lambda) &=  ( 1 - q^2) ( e^{2 \lambda} - \zeta)&  \bar{c}(\lambda) &=  e^{2 \lambda} ( 1 - q^2) ( e^{2 \lambda} - \zeta)
\end{align}
and
\<
d_{\alpha, \beta} (\lambda) = \begin{cases}
q ( e^{2 \lambda} - 1) ( e^{2 \lambda} - \zeta) + e^{2 \lambda} ( q^2 - 1)(\zeta -1) \;\;\;\; \quad\qquad \alpha = \beta = \beta' \\
(e^{2 \lambda} - 1) \left[ ( e^{2 \lambda} - \zeta) + e^{2 \lambda} ( q^2 - 1) \right] \;\;\;\; \quad \qquad\qquad\; \alpha = \beta \neq \beta' \\
( q^2 - 1) \left[ \zeta (e^{2 \lambda} - 1)  q^{(\alpha - \beta)/2} - \delta_{\alpha \beta'} ( e^{2 \lambda} - \zeta)  \right] \qquad \;\;\;\;\; \quad \alpha < \beta \\
e^{2 \lambda}  ( q^2 - 1) \left[ (e^{2 \lambda} - 1)  q^{(\alpha - \beta)/2} - \delta_{\alpha \beta'} ( e^{2 \lambda} - \zeta)  \right] \qquad \;\;\; \quad \alpha > \beta
\end{cases} \nonumber \\
\>
with $\alpha' \coloneqq 4 - \alpha$. The parameter $\zeta$ is, in its turn, given by
\[ \label{IKFZb}
\zeta = \begin{cases}
q \qquad \quad \text{for FZ model} \\
- q^3 \qquad \text{for IK model} 
\end{cases} \\ .
\]

\subsection{Domain-wall boundaries} \label{sec:DWBC}

In this subsection we intend to specialize the partition function \eqref{PFT} to cases of interest in this work. For instance, our main goal here is to study the partition function \eqref{PFT} with $K=L$ and the particular boundary conditions characterized by
\< \label{dw}
\vec{\alpha}_0 = \vec{\beta}_{0} = (1 ,1, \dots, 1) \quad \text{and} \quad  \vec{\alpha}_{L+1} = \vec{\beta}_{L+1}   = (3 , 3, \dots, 3) \; .
\>
We shall then simply write $\mathcal{Z}(\lambda_1, \lambda_2 , \dots , \lambda_L)$ for \eqref{PFT} with boundary conditions
\eqref{dw}.
The variables $\lambda_j \in \C$ are usually referred to as \emph{spectral parameters}, but it is important to remark $\mathcal{Z}$ also depends on $L$ variables $\mu_j \in \C$ commonly referred to as \emph{inhomogeneity parameters}. Moreover, due to \eqref{IKFZa}-\eqref{IKFZb}, the partition function $\mathcal{Z}$ also depends on the quantum deformation parameter $q \eqqcolon e^{\gamma} \in \C$ and $\gamma$ will then be referred to as \emph{anisotropy parameter}.
In order to describe the boundary conditions \eqref{dw} in a more intuitive way, we have also presented a possible edge-colored graph $\mathcal{G}^{*}$ admitted by $\mathcal{Z}$ in \Figref{fig:Zfig}.

\begin{figure} \centering
\scalebox{1}{
\begin{tikzpicture}
\foreach \x in {0,...,5}{
    \foreach \y in {0,...,5}{
        \path (\x,\y) node[circle, fill=white, scale=0.3](v\x\y) {};
}}  

\draw [postaction=decorate,decoration={markings, mark=at position 0.6cm with {\arrow[black]{Triangle}}}, thick]  (v11) --(v01);
\draw [postaction=decorate,decoration={markings, mark=at position 0.6cm with {\arrow[black]{Triangle}}}, thick]  (v12) -- (v02);
\draw [postaction=decorate,decoration={markings, mark=at position 0.6cm with {\arrow[black]{Triangle}}}, thick]  (v13) -- (v03);
\draw [postaction=decorate,decoration={markings, mark=at position 0.6cm with {\arrow[black]{Triangle}}}, thick]  (v14) -- (v04);

\draw [postaction=decorate,decoration={markings, mark=at position 0.6cm with {\arrow[black]{Triangle}}}, thick]  (v10) -- (v11);
\draw [postaction=decorate,decoration={markings, mark=at position 0.6cm with {\arrow[black]{Triangle}}}, thick]  (v20) -- (v21);
\draw [postaction=decorate,decoration={markings, mark=at position 0.6cm with {\arrow[black]{Triangle}}}, thick]  (v30) -- (v31);
\draw [postaction=decorate,decoration={markings, mark=at position 0.6cm with {\arrow[black]{Triangle}}}, thick]  (v40) -- (v41);

\draw [postaction=decorate,decoration={markings, mark=at position 0.6cm with {\arrow[black]{Triangle}}}, thick]  (v41) -- (v51);
\draw [postaction=decorate,decoration={markings, mark=at position 0.6cm with {\arrow[black]{Triangle}}}, thick]  (v42) -- (v52);
\draw [postaction=decorate,decoration={markings, mark=at position 0.6cm with {\arrow[black]{Triangle}}}, thick]  (v43) -- (v53);
\draw [postaction=decorate,decoration={markings, mark=at position 0.6cm with {\arrow[black]{Triangle}}}, thick]  (v44) -- (v54);

\draw [postaction=decorate,decoration={markings, mark=at position 0.6cm with {\arrow[black]{Triangle}}}, thick]  (v15) -- (v14);
\draw [postaction=decorate,decoration={markings, mark=at position 0.6cm with {\arrow[black]{Triangle}}}, thick]  (v25) -- (v24);
\draw [postaction=decorate,decoration={markings, mark=at position 0.6cm with {\arrow[black]{Triangle}}}, thick]  (v35) -- (v34);
\draw [postaction=decorate,decoration={markings, mark=at position 0.6cm with {\arrow[black]{Triangle}}}, thick]  (v45) -- (v44);

\draw [postaction=decorate,decoration={markings, mark=at position 0.6cm with {\arrow[black]{Triangle}}}, thick]  (v11) -- (v12);
\draw [postaction=decorate,decoration={markings, mark=at position 0.6cm with {\arrow[black]{Triangle}}}, thick]  (v21) -- (v11);
\draw [postaction=decorate,decoration={markings, mark=at position 0.6cm with {\arrow[black]{Triangle}}}, thick]  (v21) -- (v22);
\draw [postaction=decorate,decoration={markings, mark=at position 0.6cm with {\arrow[black]{Triangle}}}, thick]  (v31) -- (v21);

\draw [postaction=decorate,decoration={markings, mark=at position 0.6cm with {\arrow[black]{Triangle}}}, thick]  (v31) -- (v32);
\draw [postaction=decorate,decoration={markings, mark=at position 0.6cm with {\arrow[black]{Triangle}}}, thick]  (v41) -- (v31);
\draw [postaction=decorate,decoration={markings, mark=at position 0.6cm with {\arrow[black]{Triangle}}}, thick]  (v42) -- (v41);

\draw [postaction=decorate,decoration={markings, mark=at position 0.6cm with {\arrow[black]{Triangle}}}, thick]  (v12) -- (v13);
\draw [postaction=decorate,decoration={markings, mark=at position 0.6cm with {\arrow[black]{Triangle}}}, thick]  (v13) -- (v14);
\draw [postaction=decorate,decoration={markings, mark=at position 0.6cm with {\arrow[black]{Triangle}}}, thick]  (v22) -- (v12);
\draw [postaction=decorate,decoration={markings, mark=at position 0.6cm with {\arrow[black]{Triangle}}}, thick]  (v23) -- (v13);

\draw [postaction=decorate,decoration={markings, mark=at position 0.6cm with {\arrow[black]{Triangle}}}, thick]  (v14) -- (v24);
\draw [postaction=decorate,decoration={markings, mark=at position 0.6cm with {\arrow[black]{Triangle}}}, thick]  (v24) -- (v34);
\draw [postaction=decorate,decoration={markings, mark=at position 0.6cm with {\arrow[black]{Triangle}}}, thick]  (v34) -- (v44);

\draw [postaction=decorate,decoration={markings, mark=at position 0.6cm with {\arrow[black]{Triangle}}}, thick]  (v24) -- (v23);
\draw [postaction=decorate,decoration={markings, mark=at position 0.6cm with {\arrow[black]{Triangle}}}, thick]  (v34) -- (v33);
\draw [postaction=decorate,decoration={markings, mark=at position 0.6cm with {\arrow[black]{Triangle}}}, thick]  (v44) -- (v43);
\draw [postaction=decorate,decoration={markings, mark=at position 0.6cm with {\arrow[black]{Triangle}}}, thick]  (v43) -- (v42);

\draw [postaction=decorate,decoration={markings, mark=at position 0.6cm with {\arrow[black]{Triangle}}}, thick]  (v33) -- (v43);
\draw [postaction=decorate,decoration={markings, mark=at position 0.6cm with {\arrow[black]{Triangle}}}, thick]  (v32) -- (v42);

\draw [dashed, thick]  (v23) -- (v33);
\draw [dashed, thick]  (v23) -- (v22);
\draw [dashed, thick]  (v22) -- (v32);
\draw [dashed, thick]  (v32) -- (v33);
\end{tikzpicture}}
\caption{Example of $\mathcal{G}^{*}$ in $\mathcal{Z}$ for $L=4$.}
\label{fig:Zfig}
\end{figure}
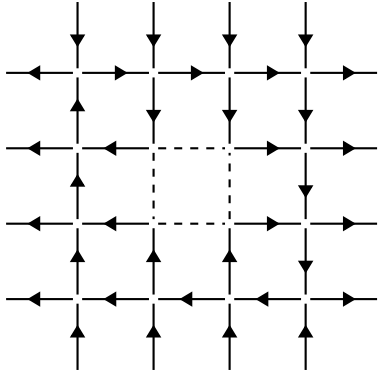

Although this work is mainly concerned with the partition function $\mathcal{Z}$, our analysis will reveal that this partition function is intimately related to another two partition functions also obtained as specializations of \eqref{PFT} with $K=L+1$. In this way, we write $\mathcal{F} (u_1, u_2 , \dots, u_{L-1} \mid v_1, v_2)$ 
for \eqref{PFT} with $K=L+1$, $\lambda_1 = v_1$, $\lambda_2 = v_2$, $\lambda_{i+2} = u_i$ ($i= 1,2, \dots, L-1$) and boundary vectors
\begin{align} \label{bvF}
\vec{\alpha}_0 &= (1, 1, \dots, 1, 1) & \vec{\beta}_0 & = (1, 1, \dots , 1) \nonumber \\
\vec{\alpha}_{L+1} &= (2, 2, 3 , 3 ,\dots, 3) & \vec{\beta}_{K+1} & = (3, 3, \dots , 3) \; .
\end{align}
Similarly, we also define $\bar{\mathcal{F}} (v_1, v_2 \mid u_1, u_2 , \dots, u_{L-1})$ as the specialization of \eqref{PFT} with $K=L+1$, $\lambda_i = u_i$ ($i=1, 2, \dots , L-1$), $\lambda_L = v_1$ and $\lambda_{L+1} = v_2$. As for the boundary vectors; $\vec{\alpha}_0$, $\vec{\beta}_0$ and $\vec{\beta}_{K+1}$ are identical to the ones in \eqref{bvF} while $\vec{\alpha}_{L+1} = (3,3, \dots, 3 ,2 ,2)$. A sample of graphs $\mathcal{G}^{*}$ giving rise to $\mathcal{F}$ and $\bar{\mathcal{F}}$ are then depicted in \Figref{fig:Ffig}.

\begin{figure} \centering
\scalebox{0.9}{
\begin{tikzpicture}
\foreach \x in {0,...,5}{
    \foreach \y in {0,...,6}{
        \path (\x,\y) node[circle, fill=white, scale=0.3](v\x\y) {};
}}  

\draw [postaction=decorate,decoration={markings, mark=at position 0.6cm with {\arrow[black]{Triangle}}}, thick]  (v11) --(v01);
\draw [postaction=decorate,decoration={markings, mark=at position 0.6cm with {\arrow[black]{Triangle}}}, thick]  (v12) -- (v02);
\draw [postaction=decorate,decoration={markings, mark=at position 0.6cm with {\arrow[black]{Triangle}}}, thick]  (v13) -- (v03);
\draw [postaction=decorate,decoration={markings, mark=at position 0.6cm with {\arrow[black]{Triangle}}}, thick]  (v14) -- (v04);
\draw [postaction=decorate,decoration={markings, mark=at position 0.6cm with {\arrow[black]{Triangle}}}, thick]  (v15) -- (v05);

\draw [postaction=decorate,decoration={markings, mark=at position 0.6cm with {\arrow[black]{Triangle}}}, thick]  (v10) -- (v11);
\draw [postaction=decorate,decoration={markings, mark=at position 0.6cm with {\arrow[black]{Triangle}}}, thick]  (v20) -- (v21);
\draw [postaction=decorate,decoration={markings, mark=at position 0.6cm with {\arrow[black]{Triangle}}}, thick]  (v30) -- (v31);
\draw [postaction=decorate,decoration={markings, mark=at position 0.6cm with {\arrow[black]{Triangle}}}, thick]  (v40) -- (v41);

\draw [dashed, thick]  (v41) -- (v51);
\draw [dashed, thick]  (v42) -- (v52);
\draw [postaction=decorate,decoration={markings, mark=at position 0.6cm with {\arrow[black]{Triangle}}}, thick]  (v43) -- (v53);
\draw [postaction=decorate,decoration={markings, mark=at position 0.6cm with {\arrow[black]{Triangle}}}, thick]  (v44) -- (v54);
\draw [postaction=decorate,decoration={markings, mark=at position 0.6cm with {\arrow[black]{Triangle}}}, thick]  (v45) -- (v55);

\draw [postaction=decorate,decoration={markings, mark=at position 0.6cm with {\arrow[black]{Triangle}}}, thick]  (v16) -- (v15);
\draw [postaction=decorate,decoration={markings, mark=at position 0.6cm with {\arrow[black]{Triangle}}}, thick]  (v26) -- (v25);
\draw [postaction=decorate,decoration={markings, mark=at position 0.6cm with {\arrow[black]{Triangle}}}, thick]  (v36) -- (v35);
\draw [postaction=decorate,decoration={markings, mark=at position 0.6cm with {\arrow[black]{Triangle}}}, thick]  (v46) -- (v45);

\draw [postaction=decorate,decoration={markings, mark=at position 0.6cm with {\arrow[black]{Triangle}}}, thick]  (v11) -- (v12);
\draw [postaction=decorate,decoration={markings, mark=at position 0.6cm with {\arrow[black]{Triangle}}}, thick]  (v12) -- (v13);
\draw [postaction=decorate,decoration={markings, mark=at position 0.6cm with {\arrow[black]{Triangle}}}, thick]  (v13) -- (v14);
\draw [postaction=decorate,decoration={markings, mark=at position 0.6cm with {\arrow[black]{Triangle}}}, thick]  (v14) -- (v15);

\draw [postaction=decorate,decoration={markings, mark=at position 0.6cm with {\arrow[black]{Triangle}}}, thick]  (v21) -- (v22);
\draw [postaction=decorate,decoration={markings, mark=at position 0.6cm with {\arrow[black]{Triangle}}}, thick]  (v22) -- (v23);
\draw [dashed, thick]  (v23) -- (v24);
\draw [postaction=decorate,decoration={markings, mark=at position 0.6cm with {\arrow[black]{Triangle}}}, thick]  (v25) -- (v24);

\draw [dashed, thick]  (v31) -- (v32);
\draw [postaction=decorate,decoration={markings, mark=at position 0.6cm with {\arrow[black]{Triangle}}}, thick]  (v33) -- (v32);
\draw [dashed, thick]  (v33) -- (v34);
\draw [postaction=decorate,decoration={markings, mark=at position 0.6cm with {\arrow[black]{Triangle}}}, thick]  (v35) -- (v34);

\draw [postaction=decorate,decoration={markings, mark=at position 0.6cm with {\arrow[black]{Triangle}}}, thick]  (v41) -- (v42);
\draw [postaction=decorate,decoration={markings, mark=at position 0.6cm with {\arrow[black]{Triangle}}}, thick]  (v42) -- (v43);
\draw [postaction=decorate,decoration={markings, mark=at position 0.6cm with {\arrow[black]{Triangle}}}, thick]  (v44) -- (v43);
\draw [postaction=decorate,decoration={markings, mark=at position 0.6cm with {\arrow[black]{Triangle}}}, thick]  (v45) -- (v44);

\draw [postaction=decorate,decoration={markings, mark=at position 0.6cm with {\arrow[black]{Triangle}}}, thick]  (v21) -- (v11);
\draw [postaction=decorate,decoration={markings, mark=at position 0.6cm with {\arrow[black]{Triangle}}}, thick]  (v31) -- (v21);
\draw [dashed, thick]  (v31) -- (v41);

\draw [postaction=decorate,decoration={markings, mark=at position 0.6cm with {\arrow[black]{Triangle}}}, thick]  (v22) -- (v12);
\draw [postaction=decorate,decoration={markings, mark=at position 0.6cm with {\arrow[black]{Triangle}}}, thick]  (v32) -- (v22);
\draw [dashed, thick]  (v32) -- (v42);

\draw [postaction=decorate,decoration={markings, mark=at position 0.6cm with {\arrow[black]{Triangle}}}, thick]  (v23) -- (v13);
\draw [postaction=decorate,decoration={markings, mark=at position 0.6cm with {\arrow[black]{Triangle}}}, thick]  (v43) -- (v33);
\draw [dashed, thick]  (v33) -- (v23);

\draw [postaction=decorate,decoration={markings, mark=at position 0.6cm with {\arrow[black]{Triangle}}}, thick]  (v24) -- (v14);
\draw [postaction=decorate,decoration={markings, mark=at position 0.6cm with {\arrow[black]{Triangle}}}, thick]  (v34) -- (v44);
\draw [dashed, thick]  (v24) -- (v34);

\draw [postaction=decorate,decoration={markings, mark=at position 0.6cm with {\arrow[black]{Triangle}}}, thick]  (v15) -- (v25);
\draw [postaction=decorate,decoration={markings, mark=at position 0.6cm with {\arrow[black]{Triangle}}}, thick]  (v25) -- (v35);
\draw [postaction=decorate,decoration={markings, mark=at position 0.6cm with {\arrow[black]{Triangle}}}, thick]  (v35) -- (v45);

\begin{scope}[xshift=8cm]
\foreach \x in {0,...,5}{
    \foreach \y in {0,...,6}{
        \path (\x,\y) node[circle, fill=white, scale=0.3](v\x\y) {};
}}  

\draw [postaction=decorate,decoration={markings, mark=at position 0.6cm with {\arrow[black]{Triangle}}}, thick]  (v11) --(v01);
\draw [postaction=decorate,decoration={markings, mark=at position 0.6cm with {\arrow[black]{Triangle}}}, thick]  (v12) -- (v02);
\draw [postaction=decorate,decoration={markings, mark=at position 0.6cm with {\arrow[black]{Triangle}}}, thick]  (v13) -- (v03);
\draw [postaction=decorate,decoration={markings, mark=at position 0.6cm with {\arrow[black]{Triangle}}}, thick]  (v14) -- (v04);
\draw [postaction=decorate,decoration={markings, mark=at position 0.6cm with {\arrow[black]{Triangle}}}, thick]  (v15) -- (v05);

\draw [postaction=decorate,decoration={markings, mark=at position 0.6cm with {\arrow[black]{Triangle}}}, thick]  (v10) -- (v11);
\draw [postaction=decorate,decoration={markings, mark=at position 0.6cm with {\arrow[black]{Triangle}}}, thick]  (v20) -- (v21);
\draw [postaction=decorate,decoration={markings, mark=at position 0.6cm with {\arrow[black]{Triangle}}}, thick]  (v30) -- (v31);
\draw [postaction=decorate,decoration={markings, mark=at position 0.6cm with {\arrow[black]{Triangle}}}, thick]  (v40) -- (v41);

\draw [dashed, thick]  (v45) -- (v55);
\draw [dashed, thick]  (v44) -- (v54);
\draw [postaction=decorate,decoration={markings, mark=at position 0.6cm with {\arrow[black]{Triangle}}}, thick]  (v41) -- (v51);
\draw [postaction=decorate,decoration={markings, mark=at position 0.6cm with {\arrow[black]{Triangle}}}, thick]  (v42) -- (v52);
\draw [postaction=decorate,decoration={markings, mark=at position 0.6cm with {\arrow[black]{Triangle}}}, thick]  (v43) -- (v53);

\draw [postaction=decorate,decoration={markings, mark=at position 0.6cm with {\arrow[black]{Triangle}}}, thick]  (v16) -- (v15);
\draw [postaction=decorate,decoration={markings, mark=at position 0.6cm with {\arrow[black]{Triangle}}}, thick]  (v26) -- (v25);
\draw [postaction=decorate,decoration={markings, mark=at position 0.6cm with {\arrow[black]{Triangle}}}, thick]  (v36) -- (v35);
\draw [postaction=decorate,decoration={markings, mark=at position 0.6cm with {\arrow[black]{Triangle}}}, thick]  (v46) -- (v45);

\draw [postaction=decorate,decoration={markings, mark=at position 0.6cm with {\arrow[black]{Triangle}}}, thick]  (v11) -- (v12);
\draw [postaction=decorate,decoration={markings, mark=at position 0.6cm with {\arrow[black]{Triangle}}}, thick]  (v13) -- (v12);
\draw [postaction=decorate,decoration={markings, mark=at position 0.6cm with {\arrow[black]{Triangle}}}, thick]  (v14) -- (v13);
\draw [postaction=decorate,decoration={markings, mark=at position 0.6cm with {\arrow[black]{Triangle}}}, thick]  (v15) -- (v14);

\draw [postaction=decorate,decoration={markings, mark=at position 0.6cm with {\arrow[black]{Triangle}}}, thick]  (v21) -- (v22);
\draw [postaction=decorate,decoration={markings, mark=at position 0.6cm with {\arrow[black]{Triangle}}}, thick]  (v22) -- (v23);
\draw [postaction=decorate,decoration={markings, mark=at position 0.6cm with {\arrow[black]{Triangle}}}, thick]  (v25) -- (v24);
\draw [dashed, thick]  (v23) -- (v24);

\draw [postaction=decorate,decoration={markings, mark=at position 0.6cm with {\arrow[black]{Triangle}}}, thick]  (v32) -- (v31);
\draw [postaction=decorate,decoration={markings, mark=at position 0.6cm with {\arrow[black]{Triangle}}}, thick]  (v33) -- (v32);
\draw [dashed, thick]  (v33) -- (v34);
\draw [dashed, thick]  (v34) -- (v35);

\draw [postaction=decorate,decoration={markings, mark=at position 0.6cm with {\arrow[black]{Triangle}}}, thick]  (v41) -- (v42);
\draw [postaction=decorate,decoration={markings, mark=at position 0.6cm with {\arrow[black]{Triangle}}}, thick]  (v42) -- (v43);
\draw [postaction=decorate,decoration={markings, mark=at position 0.6cm with {\arrow[black]{Triangle}}}, thick]  (v44) -- (v43);
\draw [postaction=decorate,decoration={markings, mark=at position 0.6cm with {\arrow[black]{Triangle}}}, thick]  (v45) -- (v44);

\draw [postaction=decorate,decoration={markings, mark=at position 0.6cm with {\arrow[black]{Triangle}}}, thick]  (v21) -- (v11);
\draw [postaction=decorate,decoration={markings, mark=at position 0.6cm with {\arrow[black]{Triangle}}}, thick]  (v31) -- (v21);
\draw [postaction=decorate,decoration={markings, mark=at position 0.6cm with {\arrow[black]{Triangle}}}, thick]  (v31) -- (v41);

\draw [postaction=decorate,decoration={markings, mark=at position 0.6cm with {\arrow[black]{Triangle}}}, thick]  (v12) -- (v22);
\draw [postaction=decorate,decoration={markings, mark=at position 0.6cm with {\arrow[black]{Triangle}}}, thick]  (v22) -- (v32);
\draw [postaction=decorate,decoration={markings, mark=at position 0.6cm with {\arrow[black]{Triangle}}}, thick]  (v32) -- (v42);

\draw [postaction=decorate,decoration={markings, mark=at position 0.6cm with {\arrow[black]{Triangle}}}, thick]  (v23) -- (v13);
\draw [dashed, thick]  (v23) -- (v33);
\draw [postaction=decorate,decoration={markings, mark=at position 0.6cm with {\arrow[black]{Triangle}}}, thick]  (v43) -- (v33);

\draw [postaction=decorate,decoration={markings, mark=at position 0.6cm with {\arrow[black]{Triangle}}}, thick]  (v24) -- (v14);
\draw [dashed, thick]  (v24) -- (v34);
\draw [dashed, thick]  (v34) -- (v44);

\draw [postaction=decorate,decoration={markings, mark=at position 0.6cm with {\arrow[black]{Triangle}}}, thick]  (v25) -- (v15);
\draw [postaction=decorate,decoration={markings, mark=at position 0.6cm with {\arrow[black]{Triangle}}}, thick]  (v35) -- (v25);
\draw [dashed, thick]  (v35) -- (v45);

\end{scope}
\end{tikzpicture}}
\caption{Graphs $\mathcal{G}^{*}$ for $L=4$ associated to $\mathcal{F}$ (left) and $\bar{\mathcal{F}}$ (right). }
\label{fig:Ffig}
\end{figure}
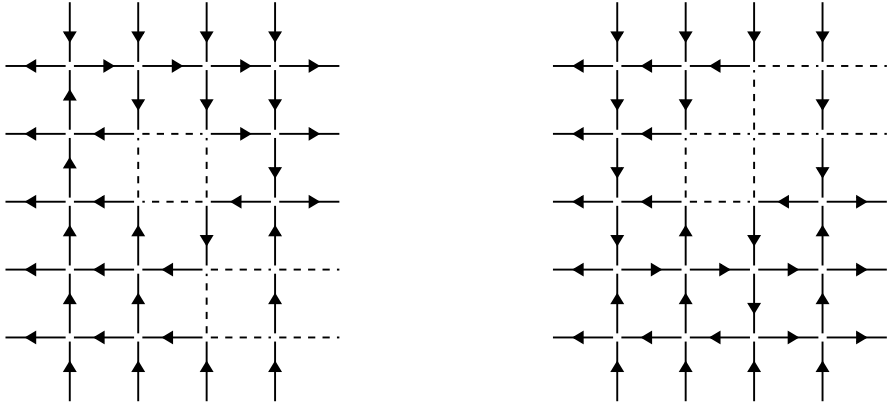

For the sake of clarity, it is also useful to having $\mathcal{Z}$, $\mathcal{F}$ and $\bar{\mathcal{F}}$ expressed directly in terms of entries of the monodromy matrix \eqref{abcd}. As for that, we introduce the simplified conventions 
\[
\mathcal{A}_1 (\lambda) \eqqcolon \mathcal{A} (\lambda) \; , \quad  \mathcal{B}_1 (\lambda) \eqqcolon \mathcal{B} (\lambda) \quad \text{and} \quad \mathcal{B}_2 (\lambda) \eqqcolon \mathcal{E} (\lambda) \; ;
 \]
as well as vectors $\ket{0} \coloneqq e_1^{\otimes L}$ and $\ket{\bar{0}} \coloneqq e_3^{\otimes L}$ in $\mathrm{End}(\V_{\mathcal{Q}})$. 
The aforementioned partition functions are then given by the following expected values:
\< \label{ZFF}
\mathcal{Z}(\lambda_1, \lambda_2, \dots, \lambda_L) &=& \bra{\bar{0}} \mathcal{E}(\lambda_{L}) \mathcal{E}(\lambda_{L-1}) \dots \mathcal{E}(\lambda_{1}) \ket{0} \nonumber \\
\mathcal{F}(u_1, u_2, \dots, u_{L-1} \mid v_1, v_2) &=& \bra{\bar{0}} \mathcal{E}(u_{L-1}) \mathcal{E}(u_{L-2}) \dots \mathcal{E}(u_{1}) \mathcal{B}(v_2) \mathcal{B}(v_1) \ket{0} \nonumber \\
\bar{\mathcal{F}}( v_1, v_2 \mid u_1, u_2, \dots, u_{L-1}) &=& \bra{\bar{0}} \mathcal{B}(v_2) \mathcal{B}(v_1) \mathcal{E}(u_{L-1}) \mathcal{E}(u_{L-2}) \dots \mathcal{E}(u_{1}) \ket{0} \; . \nonumber \\
\>
It is also important to remark here that having $\mathcal{Z}$, $\mathcal{F}$ and $\bar{\mathcal{F}}$ expressed as \eqref{ZFF} will play a major role in our forthcoming analysis.

\subsection{Symmetries} \label{sec:SYM}

In the \Appref{app:SUB} we have collected commutation relations satisfied by the operators $\mathcal{A}$, $\mathcal{B}$ and $\mathcal{E}$ built from the 
$\mathcal{R}$-matrix \eqref{bw} for the IK and FZ models. Among such commutation relations we have
\[ \label{EE}
\mathcal{E} (\lambda_i) \mathcal{E} (\lambda_j) = \mathcal{E} (\lambda_j) \mathcal{E} (\lambda_i)
\]
which has immediate consequences for $\mathcal{Z}$, $\mathcal{F}$ and $\bar{\mathcal{F}}$.
In order to examine such consequences, let us write $\mathfrak{S}_n$ for the symmetric group of degree $n$ on $\{\lambda_1, \lambda_2 , \dots , \lambda_n  \}$. Also, let $\pi_{i,j} \in \mathfrak{S}_n$ be a $2$-cycle acting as permutation of variables $\lambda_i$ and $\lambda_j$. 
Therefore, due to the commutation relation \eqref{EE}, we immediately obtain $\pi_{i,j} (\mathcal{Z}) = \mathcal{Z}$ which allows us to infer
$\mathcal{Z}(\lambda_1, \lambda_2, \dots , \lambda_L) \in \C[\lambda_1^{\pm 1} , \lambda_2^{\pm 1}, \dots , \lambda_L^{\pm 1}]^{\mathfrak{S}_L}$. 
In other words, the partition function $\mathcal{Z}$  is a symmetric function on all arguments $\lambda_j$.

On the other hand, according to formulae \eqref{ZFF}, the partition functions $\mathcal{F}$ and $\bar{\mathcal{F}}$ also involve the operator $\mathcal{B}$  whose commutation relations with $\mathcal{E}$ are sufficiently more involving. Hence, $\mathcal{F}$ and $\bar{\mathcal{F}}$ are not symmetric with respect to all of their arguments. However, they are clearly partially symmetric and we can also infer
\<
\mathcal{F}(\lambda_1, \lambda_2, \dots, \lambda_{L-1} \mid v_1, v_2), \bar{\mathcal{F}}( v_1, v_2 \mid \lambda_1, \lambda_2, \dots, \lambda_{L-1}) \in 
\C[\lambda_1^{\pm 1} , \lambda_2^{\pm 1}, \dots , \lambda_{L-1}^{\pm 1}]^{\mathfrak{S}_{L-1}} [v_1^{\pm 1}, v_2^{\pm 1}] \; . \nonumber \\
\>

\subsection{Polynomial structure} \label{sec:POL}

In the previous subsection we have analyzed the behavior of the functions $\mathcal{Z}$, $\mathcal{F}$ and $\bar{\mathcal{F}}$ with respect to the action of the symmetric group $\mathfrak{S}_n$. In this way, we were able to infer the kind of function space the functions of interest belongs to. Here we intend to further that analysis by examining in more details the dependence of $\mathcal{Z}$, $\mathcal{F}$ and $\bar{\mathcal{F}}$ on the spectral parameters. For that, it is convenient to introduce 
variables $x \coloneqq e^{2 \lambda}$, $x_i \coloneqq e^{2 \lambda_i}$, $y_1 \coloneqq e^{2 v_1}$ and $y_2 \coloneqq e^{2 v_2}$.

Now, turning our attention to the statistical weights \eqref{IKFZa}-\eqref{IKFZb} associated to the IK and FZ models, we can readily see they are polynomials in $x$ of degree two; except for $d_{1,2}$, $d_{1,3}$, $d_{2,3}$ and $c$. In their turn, the latter are polynomials in $x$ of degree one. Therefore, we can conclude the functions $\mathcal{Z}$, $\mathcal{F}$ and $\bar{\mathcal{F}}$ are polynomials in the appropriate variables and in what follows we intend to determine their polynomial degree.

In order to proceed, it is then useful to identify the graphs $g_{i,j}$ with their respective statistical weight \eqref{IKFZa}-\eqref{IKFZb}. For that we write
$w(g_{i,j}) \in \{ a, b, c, \bar{c}, d_{i,j}  \}$ and make this identification explicit in Figures \ref{fig:A}-\ref{fig:D}.

\begin{figure} \centering
\scalebox{0.85}{
\begin{tikzpicture}
\begin{scope}[xshift=2.5cm]
\path (0,0) node[circle, fill=black, scale=0.3](C) {}
      (-1,0) node[circle, fill=black, scale=0.3](W) {}
      (1,0) node[circle, fill=black, scale=0.3](E) {}
      (0,-1) node[circle, fill=black, scale=0.3](S) {}
      (0,1) node[circle, fill=black, scale=0.3](N) {};
\draw [postaction=decorate,decoration={markings, mark=at position 0.5cm with {\arrow[black]{Triangle}}}, thick]  (W) -- (C);
\draw [postaction=decorate,decoration={markings, mark=at position 0.5cm with {\arrow[black]{Triangle}}}, thick]  (C) -- (E);
\draw [postaction=decorate,decoration={markings, mark=at position 0.5cm with {\arrow[black]{Triangle}}}, thick]  (N) -- (C);
\draw [postaction=decorate,decoration={markings, mark=at position 0.5cm with {\arrow[black]{Triangle}}}, thick]  (C) -- (S);
\end{scope}

\path (0,0) node[circle, fill=black, scale=0.3](C) {}
      (-1,0) node[circle, fill=black, scale=0.3](W) {}
      (1,0) node[circle, fill=black, scale=0.3](E) {}
      (0,-1) node[circle, fill=black, scale=0.3](S) {}
      (0,1) node[circle, fill=black, scale=0.3](N) {};
\draw [postaction=decorate,decoration={markings, mark=at position 0.5cm with {\arrow[black]{Triangle}}}, thick]  (C) -- (W);
\draw [postaction=decorate,decoration={markings, mark=at position 0.5cm with {\arrow[black]{Triangle}}}, thick]  (E) -- (C);
\draw [postaction=decorate,decoration={markings, mark=at position 0.5cm with {\arrow[black]{Triangle}}}, thick]  (C) -- (N);
\draw [postaction=decorate,decoration={markings, mark=at position 0.5cm with {\arrow[black]{Triangle}}}, thick]  (S) -- (C);

\end{tikzpicture}}
\caption{Graphs $g_{i,j}$ with $w(g_{i,j}) = a$.}
\label{fig:A}
\end{figure}
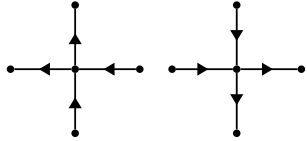

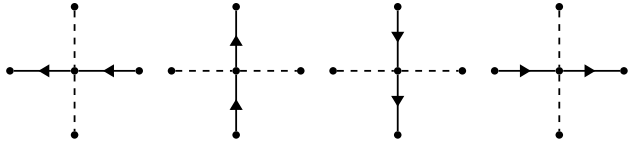
\begin{figure} \centering
\scalebox{0.85}{
\begin{tikzpicture}
\begin{scope}[xshift=7.5cm]
\path (0,0) node[circle, fill=black, scale=0.3](C) {}
      (-1,0) node[circle, fill=black, scale=0.3](W) {}
      (1,0) node[circle, fill=black, scale=0.3](E) {}
      (0,-1) node[circle, fill=black, scale=0.3](S) {}
      (0,1) node[circle, fill=black, scale=0.3](N) {};
\draw [postaction=decorate,decoration={markings, mark=at position 0.5cm with {\arrow[black]{Triangle}}}, thick]  (W) -- (C);
\draw [postaction=decorate,decoration={markings, mark=at position 0.5cm with {\arrow[black]{Triangle}}}, thick]  (C) -- (E);
\draw [dashed, thick]  (N) -- (C);
\draw [dashed, thick]  (C) -- (S);
\end{scope}

\path (0,0) node[circle, fill=black, scale=0.3](C) {}
      (-1,0) node[circle, fill=black, scale=0.3](W) {}
      (1,0) node[circle, fill=black, scale=0.3](E) {}
      (0,-1) node[circle, fill=black, scale=0.3](S) {}
      (0,1) node[circle, fill=black, scale=0.3](N) {};
\draw [postaction=decorate,decoration={markings, mark=at position 0.5cm with {\arrow[black]{Triangle}}}, thick]  (C) -- (W);
\draw [postaction=decorate,decoration={markings, mark=at position 0.5cm with {\arrow[black]{Triangle}}}, thick]  (E) -- (C);
\draw [dashed, thick]  (N) -- (C);
\draw [dashed, thick]  (C) -- (S);

\begin{scope}[xshift=2.5cm]
\path (0,0) node[circle, fill=black, scale=0.3](C) {}
      (-1,0) node[circle, fill=black, scale=0.3](W) {}
      (1,0) node[circle, fill=black, scale=0.3](E) {}
      (0,-1) node[circle, fill=black, scale=0.3](S) {}
      (0,1) node[circle, fill=black, scale=0.3](N) {};
\draw [postaction=decorate,decoration={markings, mark=at position 0.5cm with {\arrow[black]{Triangle}}}, thick]  (C) -- (N);
\draw [postaction=decorate,decoration={markings, mark=at position 0.5cm with {\arrow[black]{Triangle}}}, thick]  (S) -- (C);
\draw [dashed, thick]  (W) -- (C);
\draw [dashed, thick]  (C) -- (E);
\end{scope}

\begin{scope}[xshift=5cm]
\path (0,0) node[circle, fill=black, scale=0.3](C) {}
      (-1,0) node[circle, fill=black, scale=0.3](W) {}
      (1,0) node[circle, fill=black, scale=0.3](E) {}
      (0,-1) node[circle, fill=black, scale=0.3](S) {}
      (0,1) node[circle, fill=black, scale=0.3](N) {};
\draw [postaction=decorate,decoration={markings, mark=at position 0.5cm with {\arrow[black]{Triangle}}}, thick]  (N) -- (C);
\draw [postaction=decorate,decoration={markings, mark=at position 0.5cm with {\arrow[black]{Triangle}}}, thick]  (C) -- (S);
\draw [dashed, thick]  (W) -- (C);
\draw [dashed, thick]  (C) -- (E);
\end{scope}

\end{tikzpicture}}
\caption{Graphs $g_{i,j}$ with $w(g_{i,j}) = b$.}
\label{fig:B}
\end{figure}

\begin{figure} \centering
\scalebox{0.85}{
\begin{tikzpicture}

\path (0,0) node[circle, fill=black, scale=0.3](C) {}
      (-1,0) node[circle, fill=black, scale=0.3](W) {}
      (1,0) node[circle, fill=black, scale=0.3](E) {}
      (0,-1) node[circle, fill=black, scale=0.3](S) {}
      (0,1) node[circle, fill=black, scale=0.3](N) {};
\draw [postaction=decorate,decoration={markings, mark=at position 0.5cm with {\arrow[black]{Triangle}}}, thick]  (N) -- (C);
\draw [postaction=decorate,decoration={markings, mark=at position 0.5cm with {\arrow[black]{Triangle}}}, thick]  (C) -- (E);
\draw [dashed, thick]  (W) -- (C);
\draw [dashed, thick]  (S) -- (C);

\begin{scope}[xshift=2.5cm]
\path (0,0) node[circle, fill=black, scale=0.3](C) {}
      (-1,0) node[circle, fill=black, scale=0.3](W) {}
      (1,0) node[circle, fill=black, scale=0.3](E) {}
      (0,-1) node[circle, fill=black, scale=0.3](S) {}
      (0,1) node[circle, fill=black, scale=0.3](N) {};
\draw [postaction=decorate,decoration={markings, mark=at position 0.5cm with {\arrow[black]{Triangle}}}, thick]  (C) -- (W);
\draw [postaction=decorate,decoration={markings, mark=at position 0.5cm with {\arrow[black]{Triangle}}}, thick]  (S) -- (C);
\draw [dashed, thick]  (N) -- (C);
\draw [dashed, thick]  (C) -- (E);
\end{scope}

\begin{scope}[xshift=7.5cm]
\path (0,0) node[circle, fill=black, scale=0.3](C) {}
      (-1,0) node[circle, fill=black, scale=0.3](W) {}
      (1,0) node[circle, fill=black, scale=0.3](E) {}
      (0,-1) node[circle, fill=black, scale=0.3](S) {}
      (0,1) node[circle, fill=black, scale=0.3](N) {};
\draw [postaction=decorate,decoration={markings, mark=at position 0.5cm with {\arrow[black]{Triangle}}}, thick]  (C) -- (N);
\draw [postaction=decorate,decoration={markings, mark=at position 0.5cm with {\arrow[black]{Triangle}}}, thick]  (E) -- (C);
\draw [dashed, thick]  (W) -- (C);
\draw [dashed, thick]  (S) -- (C);
\end{scope}

\begin{scope}[xshift=10cm]
\path (0,0) node[circle, fill=black, scale=0.3](C) {}
      (-1,0) node[circle, fill=black, scale=0.3](W) {}
      (1,0) node[circle, fill=black, scale=0.3](E) {}
      (0,-1) node[circle, fill=black, scale=0.3](S) {}
      (0,1) node[circle, fill=black, scale=0.3](N) {};
\draw [postaction=decorate,decoration={markings, mark=at position 0.5cm with {\arrow[black]{Triangle}}}, thick]  (W) -- (C);
\draw [postaction=decorate,decoration={markings, mark=at position 0.5cm with {\arrow[black]{Triangle}}}, thick]  (C) -- (S);
\draw [dashed, thick]  (N) -- (C);
\draw [dashed, thick]  (C) -- (E);
\end{scope}

\end{tikzpicture}}
\caption{Graphs $g_{i,j}$ with $w(g_{i,j})=c$ (most left) and $w(g_{i,j})=\bar{c}$ (most right).}
\label{fig:C}
\end{figure}
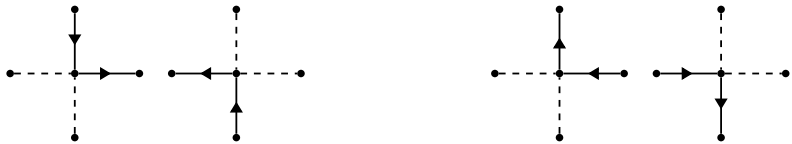

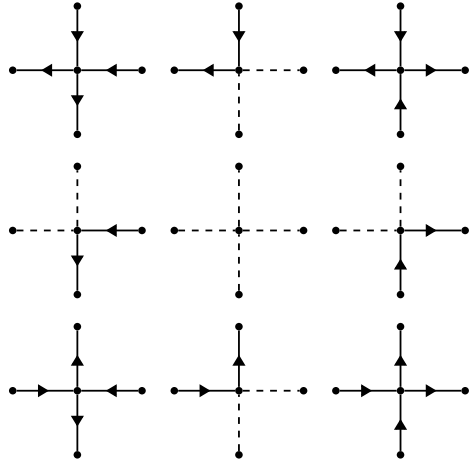
\begin{figure} \centering
\scalebox{0.85}{
\begin{tikzpicture}

\path (0,0) node[circle, fill=black, scale=0.3](C) {}
      (-1,0) node[circle, fill=black, scale=0.3](W) {}
      (1,0) node[circle, fill=black, scale=0.3](E) {}
      (0,-1) node[circle, fill=black, scale=0.3](S) {}
      (0,1) node[circle, fill=black, scale=0.3](N) {};
\draw [postaction=decorate,decoration={markings, mark=at position 0.5cm with {\arrow[black]{Triangle}}}, thick]  (C) -- (W);
\draw [postaction=decorate,decoration={markings, mark=at position 0.5cm with {\arrow[black]{Triangle}}}, thick]  (E) -- (C);
\draw [postaction=decorate,decoration={markings, mark=at position 0.5cm with {\arrow[black]{Triangle}}}, thick]  (N) -- (C);
\draw [postaction=decorate,decoration={markings, mark=at position 0.5cm with {\arrow[black]{Triangle}}}, thick]  (C) -- (S);

\begin{scope}[xshift=2.5cm]
\path (0,0) node[circle, fill=black, scale=0.3](C) {}
      (-1,0) node[circle, fill=black, scale=0.3](W) {}
      (1,0) node[circle, fill=black, scale=0.3](E) {}
      (0,-1) node[circle, fill=black, scale=0.3](S) {}
      (0,1) node[circle, fill=black, scale=0.3](N) {};
\draw [postaction=decorate,decoration={markings, mark=at position 0.5cm with {\arrow[black]{Triangle}}}, thick]  (C) -- (W);
\draw [postaction=decorate,decoration={markings, mark=at position 0.5cm with {\arrow[black]{Triangle}}}, thick]  (N) -- (C);
\draw [dashed, thick]  (C) -- (E);
\draw [dashed, thick]  (S) -- (C);
\end{scope}

\begin{scope}[xshift=5cm]
\path (0,0) node[circle, fill=black, scale=0.3](C) {}
      (-1,0) node[circle, fill=black, scale=0.3](W) {}
      (1,0) node[circle, fill=black, scale=0.3](E) {}
      (0,-1) node[circle, fill=black, scale=0.3](S) {}
      (0,1) node[circle, fill=black, scale=0.3](N) {};
\draw [postaction=decorate,decoration={markings, mark=at position 0.5cm with {\arrow[black]{Triangle}}}, thick]  (C) -- (W);
\draw [postaction=decorate,decoration={markings, mark=at position 0.5cm with {\arrow[black]{Triangle}}}, thick]  (C) -- (E);
\draw [postaction=decorate,decoration={markings, mark=at position 0.5cm with {\arrow[black]{Triangle}}}, thick]  (N) -- (C);
\draw [postaction=decorate,decoration={markings, mark=at position 0.5cm with {\arrow[black]{Triangle}}}, thick]  (S) -- (C);
\end{scope}

\begin{scope}[yshift=-2.5cm]
\path (0,0) node[circle, fill=black, scale=0.3](C) {}
      (-1,0) node[circle, fill=black, scale=0.3](W) {}
      (1,0) node[circle, fill=black, scale=0.3](E) {}
      (0,-1) node[circle, fill=black, scale=0.3](S) {}
      (0,1) node[circle, fill=black, scale=0.3](N) {};
\draw [postaction=decorate,decoration={markings, mark=at position 0.5cm with {\arrow[black]{Triangle}}}, thick]  (C) -- (S);
\draw [postaction=decorate,decoration={markings, mark=at position 0.5cm with {\arrow[black]{Triangle}}}, thick]  (E) -- (C);
\draw [dashed, thick]  (W) -- (C);
\draw [dashed, thick]  (C) -- (N);
\end{scope}

\begin{scope}[xshift=2.5cm, yshift=-2.5cm]
\path (0,0) node[circle, fill=black, scale=0.3](C) {}
      (-1,0) node[circle, fill=black, scale=0.3](W) {}
      (1,0) node[circle, fill=black, scale=0.3](E) {}
      (0,-1) node[circle, fill=black, scale=0.3](S) {}
      (0,1) node[circle, fill=black, scale=0.3](N) {};
\draw [dashed, thick]  (S) -- (C);
\draw [dashed, thick]  (C) -- (E);
\draw [dashed, thick]  (W) -- (C);
\draw [dashed, thick]  (C) -- (N);
\end{scope}

\begin{scope}[xshift=5cm, yshift=-2.5cm]
\path (0,0) node[circle, fill=black, scale=0.3](C) {}
      (-1,0) node[circle, fill=black, scale=0.3](W) {}
      (1,0) node[circle, fill=black, scale=0.3](E) {}
      (0,-1) node[circle, fill=black, scale=0.3](S) {}
      (0,1) node[circle, fill=black, scale=0.3](N) {};
\draw [postaction=decorate,decoration={markings, mark=at position 0.5cm with {\arrow[black]{Triangle}}}, thick]  (S) -- (C);
\draw [postaction=decorate,decoration={markings, mark=at position 0.5cm with {\arrow[black]{Triangle}}}, thick]  (C) -- (E);
\draw [dashed, thick]  (W) -- (C);
\draw [dashed, thick]  (C) -- (N);
\end{scope}

\begin{scope}[yshift=-5cm]
\path (0,0) node[circle, fill=black, scale=0.3](C) {}
      (-1,0) node[circle, fill=black, scale=0.3](W) {}
      (1,0) node[circle, fill=black, scale=0.3](E) {}
      (0,-1) node[circle, fill=black, scale=0.3](S) {}
      (0,1) node[circle, fill=black, scale=0.3](N) {};
\draw [postaction=decorate,decoration={markings, mark=at position 0.5cm with {\arrow[black]{Triangle}}}, thick]  (W) -- (C);
\draw [postaction=decorate,decoration={markings, mark=at position 0.5cm with {\arrow[black]{Triangle}}}, thick]  (E) -- (C);
\draw [postaction=decorate,decoration={markings, mark=at position 0.5cm with {\arrow[black]{Triangle}}}, thick]  (C) -- (N);
\draw [postaction=decorate,decoration={markings, mark=at position 0.5cm with {\arrow[black]{Triangle}}}, thick]  (C) -- (S);
\end{scope}

\begin{scope}[xshift=2.5cm, yshift=-5cm]
\path (0,0) node[circle, fill=black, scale=0.3](C) {}
      (-1,0) node[circle, fill=black, scale=0.3](W) {}
      (1,0) node[circle, fill=black, scale=0.3](E) {}
      (0,-1) node[circle, fill=black, scale=0.3](S) {}
      (0,1) node[circle, fill=black, scale=0.3](N) {};
\draw [postaction=decorate,decoration={markings, mark=at position 0.5cm with {\arrow[black]{Triangle}}}, thick]  (W) -- (C);
\draw [postaction=decorate,decoration={markings, mark=at position 0.5cm with {\arrow[black]{Triangle}}}, thick]  (C) -- (N);
\draw [dashed, thick]  (C) -- (E);
\draw [dashed, thick]  (S) -- (C);
\end{scope}

\begin{scope}[xshift=5cm, yshift=-5cm]
\path (0,0) node[circle, fill=black, scale=0.3](C) {}
      (-1,0) node[circle, fill=black, scale=0.3](W) {}
      (1,0) node[circle, fill=black, scale=0.3](E) {}
      (0,-1) node[circle, fill=black, scale=0.3](S) {}
      (0,1) node[circle, fill=black, scale=0.3](N) {};
\draw [postaction=decorate,decoration={markings, mark=at position 0.5cm with {\arrow[black]{Triangle}}}, thick]  (W) -- (C);
\draw [postaction=decorate,decoration={markings, mark=at position 0.5cm with {\arrow[black]{Triangle}}}, thick]  (C) -- (E);
\draw [postaction=decorate,decoration={markings, mark=at position 0.5cm with {\arrow[black]{Triangle}}}, thick]  (C) -- (N);
\draw [postaction=decorate,decoration={markings, mark=at position 0.5cm with {\arrow[black]{Triangle}}}, thick]  (S) -- (C);
\end{scope}

\end{tikzpicture}}
\caption{Graphs $g_{k,l}$ with $w(g_{k,l}) = d_{i,j}$ at the $i$-th row and $j$-th column.}
\label{fig:D}
\end{figure}

\begin{lem} \label{pol1}
The partition function $\mathcal{Z}(\lambda_1, \lambda_2 , \dots , \lambda_L)$ is a symmetric polynomial of degree $2L - 1$
in each variable $x_i = e^{2 \lambda_i}$ separately.
\end{lem}

\begin{proof}
The polynomial structure is a direct consequence of \eqref{bw}, \eqref{abcd}, \eqref{ZFF} and \eqref{IKFZa}-\eqref{IKFZb}; while the symmetry property with respect to the permutation of arguments has already been proved in \Secref{sec:SYM}. Also, due to \eqref{ZFF} one can see the whole dependence of $\mathcal{Z}$ on a given variable 
$\lambda_i$ is contained in a single operator $\mathcal{E}(\lambda_i)$.
Next, since $\mathcal{Z}$ is symmetric, it suffices to inspect its dependence on the variable $x_1$.
The latter then arises from the statistical weights associated to the concatenation of graphs $g_{1,j}$ for $1\leq j \leq L$, respecting the conservation of arrows discussed in Remarks \ref{AR} and \ref{AR1}.
In this way, one only needs to inspect the contribution originated from the sequence
\[ \label{seq}
w(g_{1,1}) \to w(g_{1,2}) \to \dots \to w(g_{1,L}) \; .
\]
Due to the domain-wall boundary conditions we have $w(g_{1,1}) \in \{ d_{1,3} , c, a  \}$ while 
$w(g_{1,L}) \in \{ d_{1,3} , d_{2,3},  d_{3,3}  \}$. On the same basis we find the restrictions $w(g_{1,j}) \in \{ d_{1,3} , d_{2,3}, c, a, b, d_{3,3}  \}$
for $2 \leq j \leq L-1$.
Next we introduce the short-hand notation $\{ \Lambda \}^n$ for the repeated sequence of $n$ terms
$\{ \Lambda \} \to \{ \Lambda \} \to \dots \to \{ \Lambda \}$ of any element $\Lambda$. The latter will be useful when describing the possible sequences \eqref{seq} arising
under domain-wall boundary conditions. Then, given the above described constraints, we have the following possible sequences:
\begin{enumerate}[label=(\roman*)]
\item $ \{ d_{1,3} \} \to \{ d_{3,3} \}^{L-1}$
\item $\{ c \} \to \{ b \}^n \to \{ d_{2,3} \} \to \{ d_{3,3} \}^{L-n-2} \hfill 0 \leq n \leq L-2$
\item $\{ a \}^{1+n} \to \{ d_{1,3} \} \to \{ d_{3,3} \}^{L-n-2} \hfill 0 \leq n \leq L-2$
\item $\{ a \}^{1+n} \to \{ c \} \to \{ b \}^m \to \{ d_{2,3} \} \to \{ d_{3,3} \}^{L-m-n-3} \hfill 0 \leq m \leq L-n-3 ; \; 0 \leq n \leq L-2$
\end{enumerate}

The sequences $(\mathrm{i})$ and $(\mathrm{iii})$ then give rise to polynomials in $x_1$ of degree $2L - 1$; while $(\mathrm{ii})$ and $(\mathrm{iv})$ contribute with polynomials of degree $2L -2$. Therefore, we can conclude $\mathcal{Z}$ is a polynomial in $x_1$ of degree $2L-1$.

\end{proof}

\begin{rema}
Alternatively, one could have similarly inspected the sequence $w(g_{L,1}) \to w(g_{L,2}) \to \dots \to w(g_{L,L})$ for the proof of Lemma \eqref{pol1}.
\end{rema}

\begin{lem} \label{pol2}
The functions $\mathcal{F}(\lambda_1, \lambda_2, \dots, \lambda_{L-1} \mid v_1, v_2)$ and $\bar{\mathcal{F}}( v_1, v_2 \mid \lambda_1, \lambda_2, \dots, \lambda_{L-1})$ are polynomials of degree $2L-1$ in each variable $x_i = e^{2 \lambda_i}$ separately; and also of degree $2L-1$ in each variable $y_i = e^{2 v_i}$. 
\end{lem}

\begin{proof}
The dependence of $\mathcal{F}$ and $\bar{\mathcal{F}}$ on the variable $x_i$ follows straightforwardly from the analysis performed in the proof of Lemma \ref{pol1}. Therefore, here we only need to examine the dependence on the variables $y_1$ and $y_2$.

We shall then start with the analysis of $\mathcal{F}$ and from \eqref{ZFF} we can see the whole dependence on $y_1$ is enclosed in the operator $\mathcal{B} (v_1)$. In this way, we only need to examine the contribution originated from the concatenation of graphs $g_{1,j}$ ($j=1,2,\dots, L$)
in order to determine the polynomial degree in the variable $y_1$.
More precisely, here we also need to inspect the sequence \eqref{seq} but now with possible statistical weights
\begin{align}
w(g_{1,1}) &\in \{ d_{1,3} , c , a \} & w(g_{1,L}) &\in \{ b , c \} \nonumber \\
w(g_{1,j}) &\in \{ d_{1,3} , d_{2,3} , d_{3,3} ,a , b, c \} & 2 \leq j & \leq L-1 \; ,
\end{align}
in order to comply with the required boundary conditions. In this way we find the allowed sequences
\begin{enumerate}[label=(\roman*)]
\item $ \{ c \} \to \{ b \}^{L-1}$
\item $\{ a \}^{1+n} \to \{ c \} \to \{ b \}^{L-n-2} \qquad\qquad 0 \leq n \leq L-2$ \; .
\end{enumerate}
Both sequences $(\mathrm{i})$ and $(\mathrm{ii})$ produce polynomials of degree $2L-1$ in the variable $y_1$. Therefore, the dependence on 
$y_1$ stated in Lemma \ref{pol2} is proved.

Next we move on to the dependence on $y_2$. The latter is similarly obtained from the inspection of sequences 
\[ \label{seq2}
w(g_{2,1}) \to w(g_{2,2}) \to \dots \to w(g_{2,L}) \; .
\]
The allowed sequences \eqref{seq2} will depend strongly on the possibilities $(\mathrm{i})$ and $(\mathrm{ii})$ for \eqref{seq} due to conservation of arrows. In this way, we shall split our analysis and first consider the scenario  $(\mathrm{i})$ for \eqref{seq}.
In that case we find the possible statistical weights
\begin{align}
w(g_{2,1}) &\in \{ b, d_{1,2}  \} & w(g_{2,L}) &\in \{ b , c \} \nonumber \\
w(g_{2,j}) &\in \{ d_{1,3} , d_{2,3} , d_{3,3} ,a , b, c \} & 2 \leq j & \leq L-1 \; ,
\end{align}
due to arrows conservation. The latter then yields the following possible sequences \eqref{seq2}:
\begin{enumerate}[label=(i.\alph*)]
\item $ \{ d_{1,2} \} \to \{ b \}^{L-1}$
\item $ \{ b \} \to \{ a \}^{r} \to \{ c \} \to \{ b \}^{L-r-2} \qquad\qquad 0 \leq r \leq L-2$ \; .
\end{enumerate}

The inspection of \eqref{seq2} under scenario $(\mathrm{ii})$ is more involving but still doable. In that case the statistical weights entering
\eqref{seq2} are restricted to
\begin{align} \label{bbb}
w(g_{2,1}) &\in \{ d_{1,3} , c , a \} & w(g_{2,L}) &\in \{ b , c \} \nonumber \\
w(g_{2,j}) &\in \{ d_{1,3} , d_{2,3} , d_{3,3} ,a , b, c \} & 2 \leq j  \leq L-1  ; \; & j\neq n+2 \nonumber \\
w(g_{2,n+2}) &\in \{ b, c, \bar{c}, d_{1,2} , d_{2,2} , d_{3,2}  \} \; . &  
\end{align}
Then, considering \eqref{bbb}, we have the following allowed sequences \eqref{seq2}:
\begin{enumerate}[label=(ii.\alph*)]
\item $ \{ d_{1,3} \} \to \{ d_{3,3} \}^{r} \to \{ d_{3,2} \} \to \{ b \}^{L-r-2} \hfill 0 \leq r \leq L-2$ 
\item $ \{ c \} \to \{ b \}^{r} \to \{ d_{2,2} \} \to \{ b \}^{L-r-2} \hfill 0 \leq r \leq L-2$ 
\item $ \{ c \} \to \{ b \}^r \to \{ d_{2,3} \} \to \{ d_{3,3} \}^s \to \{ d_{3,2} \} \to \{ b \}^{L-r-s-3} \hfill 0 \leq s \leq L-r-3 ; \; 0 \leq r \leq L-2$
\item $ \{ a \}^{1+r} \to \{ d_{1,3} \} \to \{ d_{3,3} \}^s \to \{ d_{3,2} \} \to \{ b \}^{L-r-s-3} \hfill 0 \leq s \leq L-r-3 ; \; 0 \leq r \leq L-2$
\item $ \{ a \}^{1+r} \to \{ c \} \to \{ b \}^{s} \to \{ d_{2,2} \} \to \{ b \}^{L-r-s-3} \hfill 0 \leq s \leq L-r-3 ; \; 0 \leq r \leq L-2$ 
\item $ \{ a \}^{1+r} \to \{ c \} \to \{ b \}^s \to \{ d_{2,3} \} \to \{ d_{3,3} \}^t \to \{ d_{3,2} \} \to \{ b \}^{L-r-s-t-4} \hfill 0 \leq t \leq L-r-s-4$ \begin{flushright} $0 \leq s \leq L-r-3 ; \; 0 \leq r \leq L-2$ \end{flushright}
\item $ \{ a \}^{1+r} \to \{ d_{1,2} \} \to \{ b \}^{L-r-2} \hfill 0 \leq r \leq L-2$ 
\item $ \{ a \}^{1+r} \to \{ b \} \to \{ c \} \to  \{ b \}^{L-r-3} \hfill 0 \leq r \leq L-3$ 
\item $ \{ a \}^{1+r} \to \{ b \} \to \{ a \}^s \to \{ c \} \to \{ b \}^{L-r-s-3} \hfill 0 \leq s \leq L-r-3 ; \; 0 \leq r \leq L-2$
\end{enumerate}
Except from $(\mathrm{ii.c})$ and $(\mathrm{ii.f})$, all the contributions arising from $(\mathrm{i.a})$ to $(\mathrm{ii.i})$ are polynomials in $y_2$
of degree $2L-1$. As for $(\mathrm{ii.c})$ and $(\mathrm{ii.f})$, they give rise to polynomials of degree $2L-2$. In this way, $\mathcal{F}(\lambda_1, \lambda_2, \dots, \lambda_{L-1} \mid v_1, v_2)$ is a polynomial of degree $2L-1$ in $y_2$.
This concludes the proof of Lemma \ref{pol2} for the function $\mathcal{F}$. 

As for the function $\bar{\mathcal{F}}$, one needs to inspect the sequences
\[ \label{seq3}
w(g_{L+1,1}) \to w(g_{L+1,2}) \to \dots \to w(g_{L+1,L})
\]
and 
\[ \label{seq4}
w(g_{L,1}) \to w(g_{L,2}) \to \dots \to w(g_{L,L})
\]
instead of \eqref{seq} and \eqref{seq2}, in order to analyze its dependence on $y_1$ and $y_2$. Due to our boundary conditions, the possible sequences
\eqref{seq3} and \eqref{seq4} can be obtained directly from \eqref{seq} and \eqref{seq2} relevant to $\mathcal{F}$ by mapping each graph $g_{i,j}$ in the sequence to its counterpart with vertical edges flipped around the central vertex. We then arrive at the same conclusions for $\bar{\mathcal{F}}$.

\end{proof}

\subsection{Simple zeroes of $\mathcal{F}$ and $\bar{\mathcal{F}}$} \label{sec:ZEROES}

From expressions \eqref{ZFF} one can promptly see that the sets of variables $\{ u_j \}$ and $\{ v_j \}$ entering the arguments of 
$\mathcal{F}$ and $\bar{\mathcal{F}}$ are not on equal footing. In particular, as previously discussed in \Secref{sec:SYM}, the functions 
$\mathcal{F}$ and $\bar{\mathcal{F}}$ are symmetric under the permutation of variables $u_j$; while a similar statement regarding the variables $v_j$ can not be made. Here, however, we intend to show the strategy used in the proof of Lemma \ref{pol2} can still yield us more information on the structure of the aforementioned functions.    

\begin{lem} \label{zeroF}
The partition function $\mathcal{F}(\lambda_1, \dots , \lambda_{L-1} \mid v_1, v_2)$ vanishes for the specializations $y_1 = e^{2 \mu_j} \zeta$ ($1 \leq j \leq L$).
\end{lem}

\begin{proof}
The dependence of $\mathcal{F}(\lambda_1, \dots , \lambda_{L-1} \mid v_1, v_2)$ on $y_1 = e^{2 v_1}$ is characterized by the allowed sequences $\mathrm{(i)}$ and $\mathrm{(ii)}$ described in the proof of Lemma \ref{pol2}. Those sequences contain only the statistical weights $b$, $c$ and $a$; which share the overall common factor $(x-\zeta)$. Therefore, the partition function $\mathcal{F}(\lambda_1, \dots , \lambda_L \mid v_1, v_2)$ vanishes for the specialization $y_1 = e^{2 \mu_j} \zeta$.
\end{proof}
\medskip
\begin{cor} \label{corFH}
The function $\mathcal{F}$ can be written as
\[ \label{FH}
\mathcal{F}(\lambda_1, \lambda_2,  \dots , \lambda_{L-1} \mid v_1, v_2) \eqqcolon \omega(y_1) \; \mathcal{H}(\lambda_1, \lambda_2,  \dots , \lambda_{L-1} \mid v_1, v_2) 
\]
with
\[ \label{omega}
\omega(y) \coloneqq \prod_{j=1}^L (y - e^{2 \mu_j} \zeta ) 
\]
and $\mathcal{H}(\lambda_1, \lambda_2,  \dots , \lambda_{L-1} \mid v_1, v_2)$ a polynomial of degree $L-1$ in $y_1$. The dependence on the other variables is still the same as of $\mathcal{F}$.
\end{cor}

\begin{proof}
Direct consequence of the polynomial structure described in Lemma \ref{pol2} and the simple zeroes of Lemma \ref{zeroF}.
\end{proof}

Next, we turn our attention to the function $\bar{\mathcal{F}}( v_1, v_2 \mid \lambda_1, \lambda_2, \dots, \lambda_{L-1})$ and
our goal is to obtain analogous versions of the Lemma \ref{zeroF} and Corollary \ref{corFH}. That can be obtained from the inspection of possible sequences \eqref{seq3} which, in their turn, can be directly read off from $\mathrm{(i)}$ and $\mathrm{(ii)}$ by flipping the vertical edges of the graphs $g_{i,j}$ entering those sequences around the central vertex. In this way, we obtain the following properties for the partition function $\bar{\mathcal{F}}$.

\begin{lem} \label{zeroFb}
The partition function $\bar{\mathcal{F}}( v_1, v_2 \mid \lambda_1, \lambda_2, \dots, \lambda_{L-1})$ vanishes when $y_2 = e^{2 \mu_j}$ for $1\leq j \leq L$.
\end{lem}

\begin{proof}
The dependence of $\bar{\mathcal{F}}$ on $y_2$ is characterized by the allowed sequences \eqref{seq3}; which can be obtained from 
$\mathrm{(i)}$ and $\mathrm{(ii)}$ through the aforementioned \emph{flipping} procedure. The latter is then mimicked by the 
maps $b \mapsto b$, $c \mapsto d_{2,3}$ and $a \mapsto d_{3,3}$, and we are left with the following sequences:
\begin{enumerate}[label=(\roman*)] \setcounter{enumi}{2}
\item $ \{ d_{2,3} \} \to \{ b \}^{L-1}$
\item $\{ d_{3,3} \}^{1+n} \to \{ d_{2,3} \} \to \{ b \}^{L-n-2} \qquad\qquad 0 \leq n \leq L-2$ \; .
\end{enumerate}
Similarly to the analysis performed for $\mathcal{F}$, one can now see $\mathrm{(iii)}$ and $\mathrm{(iv)}$ contain only the statistical weights $d_{2,3}$, $d_{3,3}$ and $b$. These weights, in their turn, share the overall common factor $(x-1)$ which implies 
$\bar{\mathcal{F}}( v_1, v_2 \mid \lambda_1, \lambda_2, \dots, \lambda_{L-1})$ vanishes when $y_2 = e^{2 \mu_j}$.
\end{proof}
\medskip
\begin{cor} \label{corFHb}
The function $\bar{\mathcal{F}}$ can be written as
\[ \label{FbHb}
\bar{\mathcal{F}}( v_1, v_2 \mid \lambda_1, \lambda_2, \dots, \lambda_{L-1}) \eqqcolon \bar{\omega}(y_2) \; \bar{\mathcal{H}}( v_1, v_2 \mid \lambda_1, \lambda_2, \dots, \lambda_{L-1}) 
\]
with 
\[ \label{omegaB}
\bar{\omega}(y) \coloneqq \prod_{j=1}^L (y - e^{2 \mu_j} ) 
\]
and $\bar{\mathcal{H}}( v_1, v_2 \mid \lambda_1, \lambda_2, \dots, \lambda_{L-1})$ a polynomial of degree $L-1$ in the variable $y_2$.
\end{cor}

\begin{proof}
Similarly to Corollary \eqref{corFH}, formulae \eqref{FbHb} and \eqref{omegaB} are direct consequences of Lemmas \ref{pol2} 
and \ref{zeroFb}.
\end{proof}

\subsection{Initial condition} \label{sec:Z0}

Up to this point we have collected definitions and properties associated to the partition functions $\mathcal{Z}$, $\mathcal{F}$
and $\bar{\mathcal{F}}$; and in the next section we intend to put forward a functional approach for studying the aforementioned
quantities. Our main goal here is to study the partition function $\mathcal{Z}$ but our framework will show such function is closely related to  $\mathcal{F}$ and $\bar{\mathcal{F}}$. In particular, using the AF method we will find a \emph{linear functional equation} characterizing $\mathcal{Z}$ and, as such, it will require an \emph{initial condition} in order to having the sought quantities completely fixed. In what follows we shall then demonstrate the existence of a special point where $\mathcal{Z}$ can be easily evaluated.

\begin{lem}[Initial condition] \label{init}
As for the specializations $\lambda_i = \mu_i$ for $1 \leq i \leq L$ we have
\[ \label{Z0}
\mathcal{Z} (\mu_1, \mu_2, \dots , \mu_L) = \prod_{i,j=1}^L a(\mu_i - \mu_j) \; .
\]
\end{lem}

\begin{proof}
We first notice the $\mathcal{R}$-matrix characterized by \eqref{rmat}, \eqref{bw} and \eqref{IKFZa}-\eqref{IKFZb} is regular in the sense that $\mathcal{R}_{ij}(\mu_j, \mu_j) = \mathcal{R}_{ij}(0) = a(0) \mathcal{P}_{ij}$ with $\mathcal{P}_{ij} \colon \V_i \otimes \V_j \to \V_j \otimes \V_i$ the permutation operator. In particular, as for $\mathcal{R}_{ij}(\mu_j, \mu_j)$, we only have contributions from the nine graphs $g_{i,j}$ depicted in Figures \ref{fig:A}, \ref{fig:C} and the anti-diagonal elements of \Figref{fig:D}. 
Hence, given the domain-wall boundary conditions, we have the single allowed sequence
\<
\{ a(\mu_{L-i+1} - \mu_1) \} &\to& \{ a(\mu_{L-i+1} - \mu_2) \} \to \dots \to \{ a(\mu_{L-i+1} - \mu_{L-i}) \} \to \{ d_{1,3}(\mu_{L-i+1} - \mu_{L-i+1}) \}  \nonumber \\ 
&\to& \{ a(\mu_{L-i+1} - \mu_{L-i+2}) \} \to \{ a(\mu_{L-i+1} - \mu_{L-i+3}) \}  \to \dots \to \{ a(\mu_{L-i+1} - \mu_L) \} \nonumber \\
\>
for the $i$-th row of our lattice under the specialization $\lambda_i = \mu_{L-i+1}$. For the sake of clarity, such configuration is diagrammatically represented in \Figref{fig:init} for $L=4$. Then, considering $d_{1,3} (0) = a(0)$ and that $\mathcal{Z}$ is a symmetric function, we immediately obtain formula \eqref{Z0}.
\end{proof}

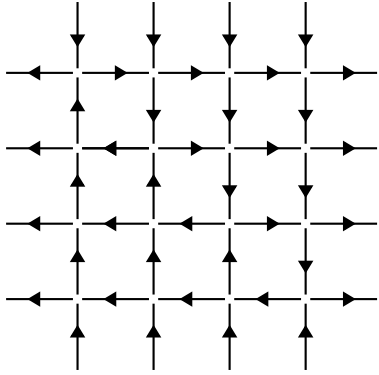
\begin{figure} \centering
\scalebox{1}{
\begin{tikzpicture}
\foreach \x in {0,...,5}{
    \foreach \y in {0,...,5}{
        \path (\x,\y) node[circle, fill=white, scale=0.3](v\x\y) {};
}}  

\draw [postaction=decorate,decoration={markings, mark=at position 0.6cm with {\arrow[black]{Triangle}}}, thick]  (v11) --(v01);
\draw [postaction=decorate,decoration={markings, mark=at position 0.6cm with {\arrow[black]{Triangle}}}, thick]  (v11) -- (v12);
\draw [postaction=decorate,decoration={markings, mark=at position 0.6cm with {\arrow[black]{Triangle}}}, thick]  (v10) -- (v11);
\draw [postaction=decorate,decoration={markings, mark=at position 0.6cm with {\arrow[black]{Triangle}}}, thick]  (v21) -- (v11);
 
\draw [postaction=decorate,decoration={markings, mark=at position 0.6cm with {\arrow[black]{Triangle}}}, thick]  (v21) -- (v22);
\draw [postaction=decorate,decoration={markings, mark=at position 0.6cm with {\arrow[black]{Triangle}}}, thick]  (v20) -- (v21);
\draw [postaction=decorate,decoration={markings, mark=at position 0.6cm with {\arrow[black]{Triangle}}}, thick]  (v31) -- (v21);

\draw [postaction=decorate,decoration={markings, mark=at position 0.6cm with {\arrow[black]{Triangle}}}, thick]  (v41) -- (v31);
\draw [postaction=decorate,decoration={markings, mark=at position 0.6cm with {\arrow[black]{Triangle}}}, thick]  (v31) -- (v32);
\draw [postaction=decorate,decoration={markings, mark=at position 0.6cm with {\arrow[black]{Triangle}}}, thick]  (v30) -- (v31);
 
\draw [postaction=decorate,decoration={markings, mark=at position 0.6cm with {\arrow[black]{Triangle}}}, thick]  (v41) -- (v51);
\draw [postaction=decorate,decoration={markings, mark=at position 0.6cm with {\arrow[black]{Triangle}}}, thick]  (v42) -- (v41);
\draw [postaction=decorate,decoration={markings, mark=at position 0.6cm with {\arrow[black]{Triangle}}}, thick]  (v40) -- (v41);

\draw [postaction=decorate,decoration={markings, mark=at position 0.6cm with {\arrow[black]{Triangle}}}, thick]  (v12) -- (v02);
\draw [postaction=decorate,decoration={markings, mark=at position 0.6cm with {\arrow[black]{Triangle}}}, thick]  (v12) -- (v13);
\draw [postaction=decorate,decoration={markings, mark=at position 0.6cm with {\arrow[black]{Triangle}}}, thick]  (v22) -- (v12);
 
\draw [postaction=decorate,decoration={markings, mark=at position 0.6cm with {\arrow[black]{Triangle}}}, thick]  (v32) -- (v22);
\draw [postaction=decorate,decoration={markings, mark=at position 0.6cm with {\arrow[black]{Triangle}}}, thick]  (v22) -- (v23);

\draw [postaction=decorate,decoration={markings, mark=at position 0.6cm with {\arrow[black]{Triangle}}}, thick]  (v33) -- (v32);
\draw [postaction=decorate,decoration={markings, mark=at position 0.6cm with {\arrow[black]{Triangle}}}, thick]  (v32) -- (v42);

\draw [postaction=decorate,decoration={markings, mark=at position 0.6cm with {\arrow[black]{Triangle}}}, thick]  (v43) -- (v42);
\draw [postaction=decorate,decoration={markings, mark=at position 0.6cm with {\arrow[black]{Triangle}}}, thick]  (v42) -- (v52);
 
\draw [postaction=decorate,decoration={markings, mark=at position 0.6cm with {\arrow[black]{Triangle}}}, thick]  (v13) -- (v03);
\draw [postaction=decorate,decoration={markings, mark=at position 0.6cm with {\arrow[black]{Triangle}}}, thick]  (v23) -- (v13);

\draw [postaction=decorate,decoration={markings, mark=at position 0.6cm with {\arrow[black]{Triangle}}}, thick]  (v13) -- (v14);
\draw [postaction=decorate,decoration={markings, mark=at position 0.6cm with {\arrow[black]{Triangle}}}, thick]  (v23) -- (v13);

\draw [postaction=decorate,decoration={markings, mark=at position 0.6cm with {\arrow[black]{Triangle}}}, thick]  (v24) -- (v23);
\draw [postaction=decorate,decoration={markings, mark=at position 0.6cm with {\arrow[black]{Triangle}}}, thick]  (v23) -- (v33);

\draw [postaction=decorate,decoration={markings, mark=at position 0.6cm with {\arrow[black]{Triangle}}}, thick]  (v34) -- (v33);
\draw [postaction=decorate,decoration={markings, mark=at position 0.6cm with {\arrow[black]{Triangle}}}, thick]  (v33) -- (v43);

\draw [postaction=decorate,decoration={markings, mark=at position 0.6cm with {\arrow[black]{Triangle}}}, thick]  (v44) -- (v43);
\draw [postaction=decorate,decoration={markings, mark=at position 0.6cm with {\arrow[black]{Triangle}}}, thick]  (v43) -- (v53);

\draw [postaction=decorate,decoration={markings, mark=at position 0.6cm with {\arrow[black]{Triangle}}}, thick]  (v14) -- (v04);
\draw [postaction=decorate,decoration={markings, mark=at position 0.6cm with {\arrow[black]{Triangle}}}, thick]  (v15) -- (v14);
\draw [postaction=decorate,decoration={markings, mark=at position 0.6cm with {\arrow[black]{Triangle}}}, thick]  (v14) -- (v24);

\draw [postaction=decorate,decoration={markings, mark=at position 0.6cm with {\arrow[black]{Triangle}}}, thick]  (v25) -- (v24);
\draw [postaction=decorate,decoration={markings, mark=at position 0.6cm with {\arrow[black]{Triangle}}}, thick]  (v24) -- (v34);

\draw [postaction=decorate,decoration={markings, mark=at position 0.6cm with {\arrow[black]{Triangle}}}, thick]  (v35) -- (v34);
\draw [postaction=decorate,decoration={markings, mark=at position 0.6cm with {\arrow[black]{Triangle}}}, thick]  (v34) -- (v44);

\draw [postaction=decorate,decoration={markings, mark=at position 0.6cm with {\arrow[black]{Triangle}}}, thick]  (v45) -- (v44);
\draw [postaction=decorate,decoration={markings, mark=at position 0.6cm with {\arrow[black]{Triangle}}}, thick]  (v44) -- (v54);
\end{tikzpicture}}
\caption{Graph $\mathcal{G}^{*}$ in $\mathcal{Z}$ for $L=4$ and $\lambda_i = \mu_i$.}
\label{fig:init}
\end{figure}

\section{Algebraic-Functional approach} \label{sec:AFM}

The previous section was devoted to the description of nineteen-vertex models with domain-wall boundary conditions. In particular, we have discussed properties of the partition functions $\mathcal{Z}$, $\mathcal{F}$ and $\bar{\mathcal{F}}$ defined by \eqref{PF} with 
boundary conditions \eqref{dw} and \eqref{bvF}. In addition to that, we have also rewritten the aforementioned partition functions in an operatorial manner in \eqref{ZFF}. 
The evaluation of such partition functions in closed form is certainly an important step towards the exact computation of physical properties of those systems; 
and finding constraints fully characterizing $\mathcal{Z}$, $\mathcal{F}$ and $\bar{\mathcal{F}}$ is our present goal.

In the subsections \ref{sec:SYM} through \ref{sec:Z0}  we have derived a series of properties satisfied by our partition functions but they are clearly not enough to characterize the aforementioned quantities. Therefore, in order to present enough constraints fixing our partition functions, here we intend to formulate the above problem in terms of functional equations along the lines of the  Algebraic-Functional (AF) method.
This framework has offered an alternative to Izergin and Korepin's method \cite{Izergin_1987, Korepin_1982} in the case of six-vertex models and generalizations; and it is based on the characterization of quantities of interest by means of functional equations originated from the Yang-Baxter algebra. 
Roughly speaking, the AF method is a framework aiming to convert algebraic relations into functional equations describing quantities of interest and, as for integrable vertex models, the so-called Yang-Baxter algebra is a suitable algebraic structure for that end.

\subsection{Yang-Baxter algebra} \label{sec:YBA}

Let us write $\mathcal{L}_i \in \mathrm{End}(\V_i)$ for a matrix with non-commutative entries fulfilling the relation
\[ \label{yba}
\mathcal{R}_{ij} (x-y) \; \mathcal{L}_i (x) \mathcal{L}_j (y) = \mathcal{L}_j (y) \mathcal{L}_i (x) \; \mathcal{R}_{ij} (x-y) \qquad \in \mathrm{End}(\V_i \otimes \V_j)
\]
with $\mathcal{R}_{ij}$ previously defined in \Secref{sec:19V}. We then refer to \eqref{yba} as Yang-Baxter algebra and use 
$\mathscr{A} (\mathcal{R})$ to denote it. The partition functions discussed in \Secref{sec:19V} will then be related to a particular representation of $\mathscr{A} (\mathcal{R})$ consisting of a pair $(\V_{\mathcal{Q}}, \mathcal{L})$, where the entries of
$\mathcal{L}$ are meromorphic functions on $\C$ with values in $\mathrm{End}(\V_{\mathcal{Q}})$.

\subsection{Modules over $\mathscr{A} (\mathcal{R})$} \label{sec:mod}
Let $\mathcal{L}_k \colon \C \to \mathrm{End}(\V_k \otimes \V_{\mathcal{Q}})$ for $k=i,j$ be meromorphic and recall
$\V_{\mathcal{Q}}$ introduced in \Secref{sec:ALG} is a diagonalizable module. Then, due to the Yang-Baxter equation \eqref{ybe}, the
pair $(\V_{\mathcal{Q}}, \mathcal{L})$ fulfills \eqref{yba} in $\mathrm{End}(\V_i \otimes \V_j \otimes \V_{\mathcal{Q}})$ with $\mathcal{L}$ identified with the monodromy matrix $\mathcal{T}$ defined in \eqref{mono}.

\subsection{Singular vectors} \label{sec:sing}

Next we shall describe a class of \emph{singular vectors} in the $\mathscr{A} (\mathcal{R})$-module $(\V_{\mathcal{Q}}, \mathcal{L})$.
More precisely, we refer to singular vectors as the non-zero elements $v_0 \in \V_{\mathcal{Q}}$ such that $\mathcal{C}_i (\lambda) v_0 = 0$ $(i=1,2,3)$ for 
all $\lambda \in \C$, with \eqref{abcd} taken into account.
In addition to that we assign the weight $\left(\Lambda_1 (\lambda), \Lambda_2 (\lambda) , \Lambda_3 (\lambda) \right)$ to an element $v\in \V_{\mathcal{Q}}$ satisfying $\mathcal{A}_i (\lambda) v = \Lambda_i (\lambda) v$ for $i=1,2,3$. The above definitions, together with the $\mathscr{A} (\mathcal{R})$-module discussed in \Secref{sec:mod}, pave the way to introduce a weight-module constituted of singular vectors $v_0 \in \V_{\mathcal{Q}}$ with weight
$\left(\Lambda_1 (\lambda), \Lambda_2 (\lambda) , \Lambda_3 (\lambda) \right)$. Hence, due to \eqref{mono} and \eqref{abcd}, one can show  the vector
$\ket{0}$ previously defined in \Secref{sec:DWBC} belongs to the above defined weight-module with weight
\[ \label{L123}
\Lambda_1 (\lambda) \coloneqq \prod_{j=1}^L a(\lambda - \mu_j) \; , \quad \Lambda_2 (\lambda) \coloneqq \prod_{j=1}^L b(\lambda - \mu_j) \quad \text{and} \quad
\Lambda_3 (\lambda) \coloneqq \prod_{j=1}^L d_{3,3} (\lambda - \mu_j) \; .
\]

Dual singular vectors and dual weight-modules are defined in a similar way. For instance, we call dual singular vectors the elements $v_0^{\dagger} \in \V_{\mathcal{Q}}^{\dagger}$ such that $v_0^{\dagger} \; \mathcal{C}_i (\lambda) = 0$ $(i=1,2,3)$ for all $\lambda \in \C$. Moreover, 
in order to characterize a dual weight-module, we assign the weight $\left(\bar{\Lambda}_1 (\lambda), \bar{\Lambda}_2 (\lambda) , \bar{\Lambda}_3 (\lambda) \right)$ to any element $v^{\dagger} \in \V_{\mathcal{Q}}^{\dagger}$ satisfying $v^{\dagger} \; \mathcal{A}_i (\lambda) = \bar{\Lambda}_i (\lambda) \; v^{\dagger}$ for $i=1,2,3$. In this way we define a dual weight-module consisting of dual singular vectors with the aforementioned weight. It is then clear the dual vector $\bra{\bar{0}}$ defined in \Secref{sec:DWBC} belongs to the dual weight-module with weight
\[ \label{bL123}
\bar{\Lambda}_1 (\lambda) \coloneqq \prod_{j=1}^L d_{1,1} (\lambda - \mu_j) \; , \quad \bar{\Lambda}_2 (\lambda) \coloneqq \prod_{j=1}^L b(\lambda - \mu_j) \quad \text{and} \quad
\bar{\Lambda}_3 (\lambda) \coloneqq \prod_{j=1}^L a_{3,3} (\lambda - \mu_j) \; .
\]

\subsection{Higher-order relations} \label{sec:HOR}

The most notable use of the algebra $\mathscr{A}(\mathcal{R})$ is within the context of the Algebraic Bethe Ansatz (ABA) method \cite{Sk_Faddeev_1979, Takh_Faddeev_1979} used for the diagonalization of transfer matrices associated to integrable vertex models with periodic boundary conditions.
The AF method can then be regarded as an alternative use of the algebra $\mathscr{A}(\mathcal{R})$ and it has found fruitful soil in models with domain-wall boundary conditions. The ABA in its turn has been formulated for a variety of models, including nineteen-vertex models \cite{Izergin_Korepin_1981},
and it is then natural to speculate if the AF method can also be applied in nineteen-vertex models with domain-wall boundaries. 
In order to investigate such possibility it is important to first examine the particular algebraic relations in $\mathscr{A}(\mathcal{R})$ associated to nineteen-vertex models. In that case $\mathscr{A}(\mathcal{R})$ is an algebra over $\C$ generated by elements $\mathcal{A}_i$, $\mathcal{B}_i$
and $\mathcal{C}_i$ defined in \eqref{abcd}. In particular, $\mathscr{A}(\mathcal{R})$ will be regarded as a matrix algebra with elements in 
$\C \llbracket x, x^{-1} \rrbracket \otimes \mathrm{End}(\V_{\mathcal{Q}})$. We shall also use $\mathscr{A}_2 (\mathcal{R})$ to denote the Yang-Baxter algebra $\mathscr{A}(\mathcal{R})$ in order to emphasize it is a quadratic algebra.
Next we introduce $\mathscr{M}_n \coloneqq \{ \mathcal{A}_i , \mathcal{B}_i , \mathcal{C}_i \mid i=1,2,3 \}(x_{n-1})$ such that one can define
$\mathscr{A}_n (\mathcal{R}) \simeq \mathscr{A}_{n-1} (\mathcal{R}) \otimes \mathscr{M}_n / \mathscr{A}_2 (\mathcal{R})$ for $n >2$ through the repeated use of $\mathscr{A}(\mathcal{R})$. We refer to $\mathscr{A}_n (\mathcal{R})$ as higher-order Yang-Baxter algebra and in what follows we shall look for relations in $\mathscr{A}_n (\mathcal{R})$ suitable for the implementation of the AF method.
For instance, the elimination of terms of the form $\mathcal{E} (\lambda) \mathcal{A}(\mu)$ using the first relation in \eqref{G3} and the second relation in \eqref{G2} gives us the following relation in $\mathscr{A}_2 (\mathcal{R})$,
\< \label{AE}
&& \left[\frac{d_{1,1}(\lambda_1 - \lambda_0)}{d_{1,2}(\lambda_1 - \lambda_0)} - \frac{d_{3,1}(\lambda_1 - \lambda_0)}{ d_{3,2}(\lambda_1 - \lambda_0)} \right] \mathcal{A}(\lambda_0) \mathcal{E}(\lambda_1) + \frac{a(\lambda_1 - \lambda_0)}{ d_{3,2}(\lambda_1 - \lambda_0)} \mathcal{A}(\lambda_1) \mathcal{E}(\lambda_0) = \nonumber \\
&& \left[ \frac{d_{2,1}(\lambda_1 - \lambda_0)}{ d_{3,2}(\lambda_1 - \lambda_0)} - \frac{d_{2,2}(\lambda_1 - \lambda_0)}{ d_{3,2}(\lambda_1 - \lambda_0)} \frac{d_{1,1}(\lambda_1 - \lambda_0)}{d_{1,2}(\lambda_1 - \lambda_0)}  \right] \mathcal{B}(\lambda_0) \mathcal{B}(\lambda_1) \nonumber \\
&& \qquad\qquad\qquad\qquad\qquad\qquad\qquad\qquad\qquad + \; \frac{a(\lambda_1 - \lambda_0)}{ d_{3,2}(\lambda_1 - \lambda_0)} \frac{d_{1,1}(\lambda_1 - \lambda_0)}{d_{1,2}(\lambda_1 - \lambda_0)}  \mathcal{B}(\lambda_1) \mathcal{B}(\lambda_0) \; . 
\>

\begin{rema} \label{ind}
Another relation exhibiting the same structure of \eqref{AE} but with apparently different coefficients can also be obtained by combining the first relation in \eqref{G3} and the third relation in \eqref{G1}. However, a closer look at the coefficients of this alternative relation shows it is not linearly independent from \eqref{AE}.
\end{rema}

Along the same lines employed in the derivation of \eqref{AE}, we also notice the elimination of $\mathcal{A} (\lambda) \mathcal{E}(\mu)$ in between the first relation of \eqref{G3}
and the third relation in \eqref{G1} leaves us with the relation

\< \label{EA}
&& \frac{a(\lambda_1 - \lambda_0)}{d_{3,3}(\lambda_1 - \lambda_0)} \mathcal{E}(\lambda_1) \mathcal{A}(\lambda_0) + \left[ \frac{d_{1,2}(\lambda_1 - \lambda_0)}{d_{3,2}(\lambda_1 - \lambda_0)} - \frac{d_{1,3}(\lambda_1 - \lambda_0)}{d_{3,3}(\lambda_1 - \lambda_0)} \right] \mathcal{E}(\lambda_0) \mathcal{A}(\lambda_1) =  \nonumber \\
&& \frac{a(\lambda_1 - \lambda_0)}{d_{3,2}(\lambda_1 - \lambda_0)} \mathcal{B}(\lambda_1) \mathcal{B}(\lambda_0) + \left[ \frac{d_{2,3}(\lambda_1 - \lambda_0)}{d_{3,3}(\lambda_1 - \lambda_0)} - \frac{d_{2,2}(\lambda_1 - \lambda_0)}{d_{3,2}(\lambda_1 - \lambda_0)} \right] \mathcal{B}(\lambda_0) \mathcal{B}(\lambda_1) \; .
\>

\begin{rema}
The elimination of $\mathcal{A} (\lambda) \mathcal{E}(\mu)$ using the first relation of \eqref{G3} and the second relation in \eqref{G2} also yields a relation with the same structure of \eqref{EA}. However, similarly to Remark \ref{ind}, the resulting relation is not linearly independent from \eqref{EA}.
\end{rema}

Both relations \eqref{AE} and \eqref{EA} live in $\mathscr{A}_2 (\mathcal{R})$ and, as it will become clear later on, we shall need relations in 
$\mathscr{A}_{L+1} (\mathcal{R})$ which can be exploited along the AF method. Such relations can then be obtained by letting 
\eqref{AE} to act on $\PROD{2}{j}{L} \mathcal{E}(\lambda_j)$ from the left; and by letting \eqref{EA} to act on the same product of operators from the right.
In this way, considering the third relation in \eqref{G3}, we are left with the following relations in $\mathscr{A}_{L+1} (\mathcal{R})$, 
\< \label{AEL}
&& \left[\frac{d_{1,1}(\lambda_1 - \lambda_0)}{d_{1,2}(\lambda_1 - \lambda_0)} - \frac{d_{3,1}(\lambda_1 - \lambda_0)}{ d_{3,2}(\lambda_1 - \lambda_0)} \right] \mathcal{A}(\lambda_0) \mathop{\overrightarrow\prod}\limits_{\substack{0 \le j \le L \\ j \neq 0}} \mathcal{E}(\lambda_j)
+ \frac{a(\lambda_1 - \lambda_0)}{ d_{3,2}(\lambda_1 - \lambda_0)} \mathcal{A}(\lambda_1) \mathop{\overrightarrow\prod}\limits_{\substack{0 \le j \le L \\ j \neq 1}} \mathcal{E}(\lambda_j) = \nonumber \\
&& \left[ \frac{d_{2,1}(\lambda_1 - \lambda_0)}{ d_{3,2}(\lambda_1 - \lambda_0)} - \frac{d_{2,2}(\lambda_1 - \lambda_0)}{ d_{3,2}(\lambda_1 - \lambda_0)} \frac{d_{1,1}(\lambda_1 - \lambda_0)}{d_{1,2}(\lambda_1 - \lambda_0)}  \right] \mathcal{B}(\lambda_0) \mathcal{B}(\lambda_1) \PROD{2}{j}{L} \mathcal{E}(\lambda_j) \nonumber \\
&& \qquad\qquad\qquad\qquad\qquad\qquad\qquad + \; \frac{a(\lambda_1 - \lambda_0)}{ d_{3,2}(\lambda_1 - \lambda_0)} \frac{d_{1,1}(\lambda_1 - \lambda_0)}{d_{1,2}(\lambda_1 - \lambda_0)}  \mathcal{B}(\lambda_1) \mathcal{B}(\lambda_0) \PROD{2}{j}{L} \mathcal{E}(\lambda_j) \nonumber \\
\>
\< \label{EAL}
&& \frac{a(\lambda_1 - \lambda_0)}{d_{3,3}(\lambda_1 - \lambda_0)} \mathop{\overrightarrow\prod}\limits_{\substack{0 \le j \le L \\ j \neq 0}} \mathcal{E}(\lambda_j) \; \mathcal{A}(\lambda_0) + \left[ \frac{d_{1,2}(\lambda_1 - \lambda_0)}{d_{3,2}(\lambda_1 - \lambda_0)} - \frac{d_{1,3}(\lambda_1 - \lambda_0)}{d_{3,3}(\lambda_1 - \lambda_0)} \right] \mathop{\overrightarrow\prod}\limits_{\substack{0 \le j \le L \\ j \neq 1}} \mathcal{E}(\lambda_j) \; \mathcal{A}(\lambda_1) =  \nonumber \\
&& \left[ \frac{d_{2,3}(\lambda_1 - \lambda_0)}{d_{3,3}(\lambda_1 - \lambda_0)} - \frac{d_{2,2}(\lambda_1 - \lambda_0)}{d_{3,2}(\lambda_1 - \lambda_0)} \right] \PROD{2}{j}{L} \mathcal{E}(\lambda_j) \; \mathcal{B}(\lambda_0) \mathcal{B}(\lambda_1) \nonumber \\
&& \qquad\qquad\qquad\qquad\qquad\qquad\qquad\qquad\qquad\quad + \; \frac{a(\lambda_1 - \lambda_0)}{d_{3,2}(\lambda_1 - \lambda_0)} \PROD{2}{j}{L} \mathcal{E}(\lambda_j) \; \mathcal{B}(\lambda_1) \mathcal{B}(\lambda_0) \; . 
\>
Now one can readily recognize the product of operators appearing in the RHS of \eqref{AEL} and \eqref{EAL} as the same operators characterizing the partition functions $\mathcal{F}$ and $\bar{\mathcal{F}}$ according to \eqref{ZFF}. As it will become clear later on, we then proceed by looking for relations in $\mathscr{A}_{L+1} (\mathcal{R})$ allowing us to express $\mathcal{B}(\lambda) \mathcal{B}(\mu) \PROD{2}{j}{L} \mathcal{E}(\lambda_j)$ in terms of 
$\PROD{2}{j}{L} \mathcal{E}(\lambda_j) \; \mathcal{B}(\bar{\lambda}) \mathcal{B}(\bar{\mu})$ and vice-versa. Such task can be accomplished through the use of suitable commutation relations in $\mathscr{S}_{\mathcal{A}, \mathcal{B} , \mathcal{E}}$. For instance, we shall use the third relation in \eqref{G2} and the last two relations in \eqref{G3}. Such commutation rules read
\< \label{GEB}
\mathcal{B}(\lambda_0) \mathcal{E}(\lambda_1) &=& \frac{a(\lambda_1-\lambda_0)}{b(\lambda_1-\lambda_0)} \mathcal{E}(\lambda_1)  \mathcal{B}(\lambda_0) -\frac{c(\lambda_1-\lambda_0)}{b(\lambda_1-\lambda_0)} \mathcal{E}(\lambda_0)  \mathcal{B}(\lambda_1) \nonumber \\
\mathcal{E}(\lambda_0) \mathcal{B}(\lambda_1) &=& \frac{a(\lambda_1-\lambda_0)}{b(\lambda_1-\lambda_0)} \mathcal{B}(\lambda_1)  \mathcal{E}(\lambda_0) -\frac{\bar{c}(\lambda_1-\lambda_0)}{b(\lambda_1-\lambda_0)} \mathcal{B}(\lambda_0)  \mathcal{E}(\lambda_1) \nonumber \\
\mathcal{E}(\lambda_0) \mathcal{E}(\lambda_1) &=& \mathcal{E}(\lambda_1) \mathcal{E}(\lambda_0)  \; ,
\>
and they provide neat exchange relations between the operators $\mathcal{B}$ and $\mathcal{E}$. Moreover, one can notice \eqref{GEB} are essentially the same commutation relations found in the six-vertex model \cite{Korepin_book}. Therefore, we can readily use the known results for the six-vertex model to find the following relations
in $\mathscr{A}_{n+2} (\mathcal{R})$,
\< \label{26}
\mathcal{B}(\lambda_{n+1}) \mathcal{B}(\lambda_0) \PROD{1}{j}{n} \mathcal{E}(\lambda_j) &=& \sum_{j=0}^n \sum_{\substack{k=0 \\ k \neq j}}^{n+1} \mathcal{M}_j^{(n)} \; \mathcal{N}_{j,k}^{(n)}  \mathop{\overrightarrow\prod}\limits_{\substack{0 \le l \le n+1 \\ l \neq j, k}} \mathcal{E}(\lambda_l) \; \mathcal{B}(\lambda_{k}) \mathcal{B}(\lambda_{j})  \label{FFb} \\
\PROD{1}{j}{n} \mathcal{E}(\lambda_j) \; \mathcal{B}(\lambda_0) \mathcal{B}(\lambda_{n+1}) &=& \sum_{j=0}^n \sum_{\substack{k=0 \\ k \neq j}}^{n+1} \bar{\mathcal{M}}_j^{(n)}  \; \bar{\mathcal{N}}_{j,k}^{(n)}  \; \mathcal{B}(\lambda_{j}) \mathcal{B}(\lambda_{k})  \mathop{\overrightarrow\prod}\limits_{\substack{0 \le l \le n+1 \\ l \neq j, k}} \mathcal{E}(\lambda_l)  \; . \label{FbF}
\>
The coefficients in \eqref{FFb} are in their turn given by
\<
\mathcal{M}_j^{(n)}  &\coloneqq& \begin{cases} 
\displaystyle \prod_{l=1}^n \frac{a(\lambda_l - \lambda_0)}{b(\lambda_l - \lambda_0)} \qquad\qquad\qquad\qquad\quad\;\;\; j=0 \\ 
\displaystyle - \frac{c(\lambda_j - \lambda_0)}{b(\lambda_j - \lambda_0)} \prod_{\substack{l=1 \\ l \neq j}}^n \frac{a(\lambda_l - \lambda_j)}{b(\lambda_l - \lambda_j)} \qquad\qquad\;\; 1 \leq j \leq n \end{cases} \nonumber \\
\mathcal{N}_{j,k}^{(n)}  &\coloneqq& \begin{cases} 
\displaystyle \prod_{\substack{l=0 \\ l \neq j}}^n \frac{a(\lambda_l - \lambda_{n+1})}{b(\lambda_l - \lambda_{n+1})} \qquad\qquad\qquad\qquad\quad\!\! k=n+1 \\ 
\displaystyle - \frac{c(\lambda_k - \lambda_{n+1})}{b(\lambda_k - \lambda_{n+1})} \prod_{\substack{l=0 \\ l \neq j, k}}^n \frac{a(\lambda_l - \lambda_k)}{b(\lambda_l - \lambda_k)} \qquad\quad 0 \leq k \leq n ; \; k \neq j \end{cases}
\>
while the ones in \eqref{FbF} read
\<
\bar{\mathcal{M}}_j^{(n)}  &\coloneqq& \begin{cases} 
\displaystyle \prod_{l=1}^n \frac{a(\lambda_0 - \lambda_l)}{b(\lambda_0 - \lambda_l)} \qquad\qquad\qquad\qquad\quad\;\;\; j=0 \\ 
\displaystyle - \frac{\bar{c}(\lambda_0 - \lambda_j)}{b(\lambda_0 - \lambda_j)} \prod_{\substack{l=1 \\ l \neq j}}^n \frac{a(\lambda_j - \lambda_l)}{b(\lambda_j - \lambda_l)} \qquad\qquad\;\; 1 \leq j \leq n \end{cases} \nonumber \\
\bar{\mathcal{N}}_{j,k}^{(n)}  &\coloneqq& \begin{cases} 
\displaystyle \prod_{\substack{l=0 \\ l \neq j}}^n \frac{a(\lambda_{n+1} - \lambda_l)}{b(\lambda_{n+1} - \lambda_l)} \qquad\qquad\qquad\qquad\quad\!\! k=n+1 \\ 
\displaystyle - \frac{\bar{c}(\lambda_{n+1} - \lambda_k)}{b(\lambda_{n+1} - \lambda_k)} \prod_{\substack{l=0 \\ l \neq j, k}}^n \frac{a(\lambda_k - \lambda_l)}{b(\lambda_k - \lambda_l)} \qquad\quad 0 \leq k \leq n ; \; k \neq j \end{cases} \; . 
\>
As previously remarked, our approach will require relations of type \eqref{FFb} and \eqref{FbF} in  $\mathscr{A}_{L+1} (\mathcal{R})$. The latter can then be obtained from \eqref{FFb} and \eqref{FbF} by setting $n=L-1$.

\subsection{The functional $\Phi$} \label{sec:PHI}

After having established suitable higher-order algebraic relations in $\mathscr{A}_{n+1} (\mathcal{R})$, the next step within the AF approach is to find a linear functional $\Phi \colon \mathscr{A}_{n+1} (\mathcal{R}) \to \C [ \lambda_0^{\pm 1}, \lambda_1^{\pm 1}, \dots , \lambda_n^{\pm 1}]$ allowing us to write functional equations for quantities of interest. In particular, we would like the functional $\Phi$ to satisfy the property
\[ 
\Phi(J_n) = \omega_J (\lambda_0 , \lambda_1 , \dots , \lambda_{n-1}) \; \Phi(J_{n-1}) 
\]
for certain elements $J_n \subseteq \mathscr{A}_{n} (\mathcal{R})$ and a fixed meromorphic functions $\omega_J$.
As for the characterization of the partition functions $\mathcal{Z}$, $\mathcal{F}$ and $\bar{\mathcal{F}}$; we shall employ the higher-order relations \eqref{AEL}, \eqref{EAL}, \eqref{FFb} and \eqref{FbF} in $\mathscr{A}_{L+1} (\mathcal{R})$. Moreover, a closer inspection of such relations suggests considering the following realization of the functional $\Phi$, namely
\[ \label{funk}
\Phi (J_{L+1}) = \bra{\bar{0}} J_{L+1} \ket{0} ,
\]
with vectors $\bra{\bar{0}}$ and $\ket{0}$ previously defined in \Secref{sec:DWBC}. In the next section we shall then precise the functional equations obtained from the application of \eqref{funk} on \eqref{AEL}, \eqref{EAL}, \eqref{FFb} and \eqref{FbF}.

\section{Functional equations} \label{sec:PROP}

This section is concerned with the explicit construction and analysis of functional relations satisfied by the partition functions $\mathcal{Z}$, $\mathcal{F}$ and $\bar{\mathcal{F}}$; using the AF method described in the previous section. However, it is fair to say in the previous section we have only collected the ingredients required for the derivation of the anticipated functional equations. Here we intend to bring that procedure to conclusion by combining all those ingredients in a suitable way. For that it is convenient to introduce the following extra conventions. 

Let us write $\gen{X} \coloneqq \{ \lambda_1, \lambda_2 , \dots , \lambda_L \}$ for fixed $L \in \Z_{\leq 1}$ and additionally introduce the short-hand notation 
\[
\gen{X}_{\alpha_1, \alpha_2 , \dots , \alpha_l}^{\beta_1, \beta_2 , \dots , \beta_m} \coloneqq \gen{X} \cup \{ \lambda_{\beta_1} , \lambda_{\beta_2} , \dots , \lambda_{\beta_m} \} \backslash \{ \lambda_{\alpha_1} , \lambda_{\alpha_2} , \dots , \lambda_{\alpha_l} \} \; .
\]
Moreover, we shall also use 
\begin{align} \label{LOM}
\Lambda (\lambda) &\coloneqq    \prod_{j=1}^L a(\lambda - \mu_j)  & \bar{\Lambda} (\lambda) &\coloneqq  \prod_{j=1}^L d_{1,1} (\lambda - \mu_j) \nonumber \\
\omega(\lambda) &\coloneqq  \prod_{j=1}^L (e^{2 \lambda} - e^{2 \mu_j} \zeta ) & \bar{\omega}(\lambda) & \coloneqq \prod_{j=1}^L (e^{2 \lambda} - e^{2 \mu_j} ) \; .
\end{align}
In this way, we can construct the following functional relations. 

\begin{lem} \label{ZH}
The functions $\mathcal{Z}$ and $\mathcal{H}$ are related through the equation
\[ \label{zh}
\Omega_0 \; \mathcal{Z}(\gen{X}) + \Omega_1 \; \mathcal{Z}(\gen{X}_1^0) = \Upsilon_0 \; \mathcal{H} (\gen{X}_1 \mid \lambda_0 , \lambda_1 )      
+ \Upsilon_1 \; \mathcal{H} (\gen{X}_1 \mid \lambda_1 , \lambda_0 )
\]
with coefficients
\begin{align} \label{OM}
\Omega_0 &\coloneqq \frac{a(\lambda_1 - \lambda_0)}{d_{3,3}(\lambda_1 - \lambda_0)} \Lambda (\lambda_0) &\Omega_1 &\coloneqq \left[ \frac{d_{1,2}(\lambda_1 - \lambda_0)}{d_{3,2}(\lambda_1 - \lambda_0)} - \frac{d_{1,3}(\lambda_1 - \lambda_0)}{d_{3,3}(\lambda_1 - \lambda_0)} \right]  \Lambda (\lambda_1)  \nonumber \\
\Upsilon_0 &\coloneqq \frac{a(\lambda_1 - \lambda_0)}{d_{3,2}(\lambda_1 - \lambda_0)} \omega(\lambda_0) &\Upsilon_1 &\coloneqq \left[ \frac{d_{2,3}(\lambda_1 - \lambda_0)}{d_{3,3}(\lambda_1 - \lambda_0)} - \frac{d_{2,2}(\lambda_1 - \lambda_0)}{d_{3,2}(\lambda_1 - \lambda_0)} \right] \omega(\lambda_1) \; . \nonumber \\
\end{align}
\end{lem}
\begin{proof}
The proof is straightforward and equation \eqref{zh} follows from the application of the functional $\Phi$ defined in \eqref{funk} on the higher-order relation \eqref{EAL}. Also, in order to obtain \eqref{zh}, one also needs to recall $\ket{0}$ and $\ket{\bar{0}}$ are singular vectors, with properties described in \Secref{sec:sing}, and use formulae \eqref{ZFF} and \eqref{FH}.
\end{proof}

\begin{lem} \label{ZHb}
Similarly to Lemma \ref{ZH}, there also exists a functional relation between $\mathcal{Z}$ and $\bar{\mathcal{H}}$, namely
\[ \label{zhb} 
\bar{\Omega}_0 \; \mathcal{Z}(\gen{X}) + \bar{\Omega}_1 \; \mathcal{Z}(\gen{X}_1^0) = \bar{\Upsilon}_0 \; \bar{\mathcal{H}} (\lambda_1 , \lambda_0 \mid \gen{X}_1)+ \bar{\Upsilon}_1 \; \bar{\mathcal{H}} (\lambda_0 , \lambda_1 \mid \gen{X}_1)  \; ,
\]
with coefficients reading
\< \label{bOM}
\bar{\Omega}_0 &\coloneqq& \left[\frac{d_{1,1}(\lambda_1 - \lambda_0)}{d_{1,2}(\lambda_1 - \lambda_0)} - \frac{d_{3,1}(\lambda_1 - \lambda_0)}{ d_{3,2}(\lambda_1 - \lambda_0)} \right] \bar{\Lambda} (\lambda_0) \qquad 
\bar{\Upsilon}_1 \coloneqq \frac{a(\lambda_1 - \lambda_0)}{ d_{3,2}(\lambda_1 - \lambda_0)} \frac{d_{1,1}(\lambda_1 - \lambda_0)}{d_{1,2}(\lambda_1 - \lambda_0)} \bar{\omega}(\lambda_1) \nonumber \\
\bar{\Upsilon}_0 &\coloneqq&  \left[ \frac{d_{2,1}(\lambda_1 - \lambda_0)}{ d_{3,2}(\lambda_1 - \lambda_0)} - \frac{d_{2,2}(\lambda_1 - \lambda_0)}{ d_{3,2}(\lambda_1 - \lambda_0)} \frac{d_{1,1}(\lambda_1 - \lambda_0)}{d_{1,2}(\lambda_1 - \lambda_0)}  \right] \bar{\omega}(\lambda_0) \qquad
\bar{\Omega}_1 \coloneqq  \frac{a(\lambda_1 - \lambda_0)}{ d_{3,2}(\lambda_1 - \lambda_0)} \bar{\Lambda} (\lambda_1) \; . \nonumber \\
\>
\end{lem}

\begin{proof}
Along the same lines used in the proof of Lemma \ref{ZH}, we simply apply the functional $\Phi$ on the higher-order relation \eqref{AEL}, keeping in mind formulae \eqref{ZFF} and \eqref{FbHb}.
\end{proof}

Some comments are in order at this point. For instance, Lemma \ref{ZH} establishes a relation between a set of functions $\mathcal{Z}$ and a set of functions $\mathcal{H}$. In this way, once the function $\mathcal{H}$ is known, we would have a functional equation involving solely the partition function $\mathcal{Z}$. This would be the optimal situation resembling the equations found in the six-vertex model through the AF method. Also, we stress here the same remarks apply when considering Lemma \ref{ZHb} and the function $\bar{\mathcal{H}}$.
However, we are in the situation that neither $\mathcal{H}$ nor $\bar{\mathcal{H}}$ are known a priori, and this fact makes our analysis significantly more involving.
Therefore, in order to circumvent the aforementioned difficulty, we shall consider the following strategy. Taking into account the functional relations described in Lemmas \ref{ZH} and \ref{ZHb}, one can notice we would have an effective relation between functions $\mathcal{Z}$ (with different spectral parameters) in case we were able to relate the functions $\mathcal{H}$ and $\bar{\mathcal{H}}$ appearing respectively in Lemmas \ref{ZH} and \ref{ZHb}. In this way, one could regard $\mathcal{H}$ and $\bar{\mathcal{H}}$ as \emph{auxiliary functions} and an schematic representation of this strategy can be found in \Figref{fig:schema}. Fortunately, the sought relation between $\mathcal{H}$ and $\bar{\mathcal{H}}$ can also be obtained using the AF method and we shall refer to the resulting system of equations as $\gen{ZH}$-system.

\begin{lem} \label{HHb}
The following relations hold 
\< \label{hhb}
\omega(\lambda_L) \; \mathcal{H}(\gen{X}_L \mid \lambda_L , \lambda_0 ) &=& \sum_{j=0}^{L-1} \sum_{\substack{k=0 \\ k \neq j}}^{L} \bar{\mathcal{M}}_j^{(L-1)} \; \bar{\mathcal{N}}_{j,k}^{(L-1)} \; \bar{\omega}(\lambda_j) \; \bar{\mathcal{H}} ( \lambda_k , \lambda_j \mid \gen{X}_{j,k}^{0} ) \nonumber \\
\bar{\omega}(\lambda_L) \; \bar{\mathcal{H}}(\lambda_0 , \lambda_L \mid \gen{X}_L) &=& \sum_{j=0}^{L-1} \sum_{\substack{k=0 \\ k \neq j}}^{L} \mathcal{M}_j^{(L-1)} \; \mathcal{N}_{j,k}^{(L-1)} \; \omega(\lambda_j) \; \mathcal{H} (\gen{X}_{j,k}^{0} \mid \lambda_j , \lambda_k) \; . \nonumber \\
\>
\end{lem}

\begin{proof}
We apply the functional $\Phi$ defined in \eqref{funk} on the higher-order relations \eqref{26} with $n=L-1$. Then we are able to recognize the functions $\mathcal{H}$ and $\bar{\mathcal{H}}$ with the help of \eqref{ZFF}, \eqref{FH} and \eqref{FbHb}.
\end{proof}

\begin{rema}
The inspection of equations \eqref{hhb} for small values of $L$ shows the first equation is immediately satisfied upon the substitution of the second equation and vice-versa.
Therefore, equations \eqref{hhb} are not linearly independent and one can consider only one of them.
\end{rema}

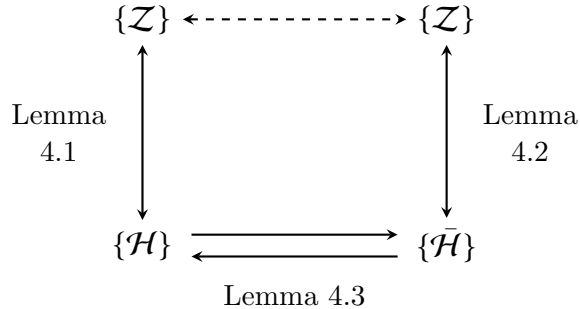
\begin{figure} \centering
\scalebox{1}{
\begin{tikzpicture}[>=stealth]
\path (0,0) node[rectangle,align=center] (p1) {{\normalsize $\{ \mathcal{H}  \}$}}
      (0,3) node[rectangle,align=center] (p2) {{\normalsize $\{ \mathcal{Z}  \}$}}
      (4,0) node[rectangle,align=center] (p3) {{\normalsize $\{ \bar{\mathcal{H}}  \}$}}
      (4,3) node[rectangle,align=center] (p4) {{\normalsize $\{ \mathcal{Z}  \}$}};

\path (0.5,0) node[rectangle,align=center] (pp1) {}
      (3.5,0) node[rectangle,align=center] (pp3) {}
      (2,-0.65) node[rectangle,align=center] (qh) {{\small Lemma \ref{HHb}}}
      (-1.1,1.5) node[rectangle,align=center] (qhz) {{\small Lemma} \\ {\small \ref{ZH}}}
      (5.1,1.5) node[rectangle,align=center] (qhbz) {{\small Lemma} \\ {\small \ref{ZHb}}};

\draw [->, thick] (pp1.north east) -- (pp3.north west);
\draw [->, thick]  (pp3.south west) -- (pp1.south east);
\draw [<->, thick]  (p1.north) -- (p2.south);
\draw [<->, thick]  (p3.north) -- (p4.south);
\draw [<->, thick, dashed]  (p2.east) -- (p4.west);
\end{tikzpicture}}
\caption{Schematic representation of the $\gen{ZH}$-system.}
\label{fig:schema}
\end{figure}

\subsection{The $\gen{ZH}$-system} \label{sec:ZH}

Lemmas \ref{ZH}, \ref{ZHb} and \ref{HHb} describe a system of functional equations relating the partition functions with domain-wall boundaries discussed in \Secref{sec:DWBC}. Here we shall refer to the above described system of equations as $\gen{ZH}$-system and, more precisely, it comprises equations \eqref{zh}, \eqref{zhb} and \eqref{hhb}. Our partition functions $\mathcal{Z}$, $\mathcal{F}$ and $\bar{\mathcal{F}}$ satisfy the $\gen{ZH}$-system by construction; however, it is not a priori clear  if the $\gen{ZH}$-system fixes uniquely such functions. Here we claim this is indeed the case; and to support our claim we proceed with a more detailed analysis of the $\gen{ZH}$-system and present explicit solutions for small lattices.  

We start our analysis by recalling the coefficients given in \eqref{OM} and \eqref{bOM} depend solely on the spectral parameters $\lambda_0$ and $\lambda_1$. Hence, in order to emphasize such dependence, we also write $\Omega_i = \Omega_i (\lambda_0, \lambda_1)$, $\bar{\Omega}_i = \bar{\Omega}_i (\lambda_0, \lambda_1)$, $\Upsilon_i = \Upsilon_i (\lambda_0, \lambda_1)$ and  $\bar{\Upsilon}_i = \bar{\Upsilon}_i (\lambda_0, \lambda_1)$. Next, we use Cramer's method to solve the system of equations formed by \eqref{zh} and \eqref{zhb} for $\mathcal{Z}(\gen{X})$ and $\mathcal{Z}(\gen{X}_1^0)$. By doing so we obtain two expressions for the function $\mathcal{Z}$; one depending on the set of variables $\gen{X}$ and another one for $\mathcal{Z}$ depending on $\gen{X}_1^0$. The obtained expressions need to be consistent and denote the same function upon an appropriate renaming of variables.
In this way, we are left with a functional relation between $\mathcal{H}$ and $\bar{\mathcal{H}}$.
More precisely, the resolution of \eqref{zh} and \eqref{zhb} yields the expressions
\< \label{ZX1}
\mathcal{Z} (\gen{X}) &=& \frac{\Upsilon_0 (\lambda_0, \lambda_1) \; \bar{\Omega}_1 (\lambda_0, \lambda_1) }{\mathcal{W}(\lambda_0 , \lambda_1)} \mathcal{H}(\gen{X}_1 \mid \lambda_0 , \lambda_1) + \frac{\Upsilon_1 (\lambda_0, \lambda_1) \; \bar{\Omega}_1 (\lambda_0, \lambda_1) }{\mathcal{W}(\lambda_0 , \lambda_1)} \mathcal{H}(\gen{X}_1 \mid \lambda_1 , \lambda_0) \nonumber \\
&& - \; \frac{\bar{\Upsilon}_0 (\lambda_0, \lambda_1) \; \Omega_1 (\lambda_0, \lambda_1) }{\mathcal{W}(\lambda_0 , \lambda_1)} \bar{\mathcal{H}}(\lambda_1 , \lambda_0 \mid \gen{X}_1) - \frac{\bar{\Upsilon}_1 (\lambda_0, \lambda_1) \; \Omega_1 (\lambda_0, \lambda_1) }{\mathcal{W}(\lambda_0 , \lambda_1)} \bar{\mathcal{H}}(\lambda_0 , \lambda_1 \mid \gen{X}_1)
\>
and 
\< \label{ZX2}
\mathcal{Z} (\gen{X}) &=& \frac{\Omega_0 (\lambda_1, \lambda_{\bar{0}}) \; \bar{\Upsilon}_1 (\lambda_1, \lambda_{\bar{0}}) }{\mathcal{W}(\lambda_1 , \lambda_{\bar{0}})} \bar{\mathcal{H}}(\lambda_1 , \lambda_{\bar{0}} \mid \gen{X}_1) + \frac{\Omega_0 (\lambda_1, \lambda_{\bar{0}}) \; \bar{\Upsilon}_0 (\lambda_1, \lambda_{\bar{0}}) }{\mathcal{W}(\lambda_1 , \lambda_{\bar{0}})} \bar{\mathcal{H}}(\lambda_{\bar{0}}, \lambda_1  \mid \gen{X}_1) \nonumber \\
&& - \; \frac{\bar{\Omega}_0 (\lambda_1, \lambda_{\bar{0}}) \; \Upsilon_1 (\lambda_1, \lambda_{\bar{0}}) }{\mathcal{W}(\lambda_1 , \lambda_{\bar{0}})} \mathcal{H}(\gen{X}_1 \mid \lambda_{\bar{0}} , \lambda_1 ) - \frac{\bar{\Omega}_0 (\lambda_1, \lambda_{\bar{0}}) \; \Upsilon_0 (\lambda_1, \lambda_{\bar{0}}) }{\mathcal{W}(\lambda_1 , \lambda_{\bar{0}})} \mathcal{H}(\gen{X}_1 \mid \lambda_1, \lambda_{\bar{0}} ) 
\>
with
\< \label{WW}
\mathcal{W}(\lambda_0 , \lambda_1) &\coloneqq& \gen{det} \begin{pmatrix} \Omega_0 (\lambda_0, \lambda_1) & \Omega_1 (\lambda_0, \lambda_1) \\ \bar{\Omega}_0 (\lambda_0, \lambda_1) & \bar{\Omega}_1 (\lambda_0, \lambda_1) \end{pmatrix} \nonumber \\
&=&  \frac{[ 1 - q^2 e^{2(\lambda_0 - \lambda_1)} ]^2 [ 1 - \zeta e^{2(\lambda_0 - \lambda_1)} ]^2 }{[ e^{2(\lambda_0 - \lambda_1)} - 1  ]^2 [ q^2 - \zeta e^{2(\lambda_0 - \lambda_1)} ]} \frac{\Lambda(\lambda_0) \bar{\Lambda}(\lambda_1)}{q^{\frac{1}{2}} (q^2 -1)} \nonumber \\
&& - \;  \frac{[ 1 - \zeta e^{2(\lambda_0 - \lambda_1)} ]^2 [ q^4 - \zeta^2 e^{2(\lambda_0 - \lambda_1)} ]^2 }{ \zeta^2 [ e^{2(\lambda_0 - \lambda_1)} - 1  ]^2 [ q^2 - \zeta e^{2(\lambda_0 - \lambda_1)} ]} \frac{\Lambda(\lambda_1) \bar{\Lambda}(\lambda_0)}{q^{\frac{5}{2}} (q^2 -1)} \; .
\>
As for the aforementioned relation between $\mathcal{H}$ and $\bar{\mathcal{H}}$, it is readily obtained through the identity $\mathcal{Z} (\gen{X}) = \mathcal{Z} (\gen{X})$ using \eqref{ZX1} and \eqref{ZX2}. In fact, different relations between $\mathcal{H}$ and $\bar{\mathcal{H}}$ could also be obtained from \eqref{ZX1} and \eqref{ZX2}. For instance, one finds a differential equation through the obvious identity $\partial \mathcal{Z} (\gen{X}) / \partial \lambda_0 = 0$ with $\mathcal{Z} (\gen{X})$ given by
\eqref{ZX1}; and also $\partial \mathcal{Z} (\gen{X}) / \partial \lambda_{\bar{0}} = 0$ using \eqref{ZX2}. Moreover, from \eqref{ZX1} and/or \eqref{ZX2} one can clearly see that the partition function $\mathcal{Z}$ is fixed once we determine the functions $\mathcal{H}$ and $\bar{\mathcal{H}}$. In what follows we shall then discuss the resolution of the $\gen{ZH}$-system for lattice lengths $L=1,2,3$.

\subsection{Case $L=1$} \label{sec:L1}

This is the simplest instance of the $\gen{ZH}$-system and its analysis requires only trivial considerations given the results already obtained in the previous sections. For instance, Eq. \eqref{hhb} for $L=1$ gives
\[ \label{HH1}
\omega(\lambda_1) \; \mathcal{H}(\emptyset \mid \lambda_1, \lambda_0) = \bar{\omega}(\lambda_0) \; \bar{\mathcal{H}} (\lambda_1, \lambda_0 \mid \emptyset) \; ,
\]
which corresponds to the condition $\mathcal{F} = \bar{\mathcal{F}}$ expected from formulae \eqref{ZFF}.
Now, taking into account the polynomial structure described in Corollaries \ref{corFH} and \ref{corFHb}, the identity \eqref{HH1} allows us to conclude
\[ \label{HHH}
\mathcal{H}(\emptyset \mid \lambda_1, \lambda_0) = \kappa \; \bar{\omega}(\lambda_0) \qquad \text{and} \qquad \bar{\mathcal{H}} (\lambda_1, \lambda_0 \mid \emptyset) = \kappa \; \omega(\lambda_1)
\]
with $\kappa \in \C$ a constant yet to be determined. We can then simply substitute formulae \eqref{HHH} in \eqref{ZX1} or \eqref{ZX2} to find the partition function $\mathcal{Z}$ up to the overall multiplicative factor $\kappa$. The latter is then fixed by the initial condition \eqref{Z0} and we end up with
$\mathcal{Z}(\lambda_1) = d_{1,3}(\lambda_1 - \mu_1)$ as expected.

\subsection{Case $L=2$} \label{sec:L2}

Our goal in solving explicitly the $\gen{ZH}$-system for small values of the lattice length $L$ is to provide evidences that our system of equations indeed constrain the partition functions $\mathcal{Z}$, $\mathcal{H}$ and $\bar{\mathcal{H}}$ up to an overall multiplicative constant. In particular, we are interested in showing the existence of unique polynomial solutions with structure described in Lemmas \ref{pol1} and \ref{pol2}.
Therefore, taking into account the aforementioned Lemmas as well as Corollaries \ref{corFH} and \ref{corFHb}, we can write for $L=2$
\< \label{PH}
\mathcal{H}(\lambda_0 \mid \lambda_1, \lambda_2) &=& \sum_{i = 0}^3 \sum_{j = 0}^1 \sum_{k = 0}^3 \phi_{i,j,k} \; e^{2 i \lambda_0 + 2 j \lambda_1 + 2 k \lambda_2} \nonumber \\
\bar{\mathcal{H}}(\lambda_0 , \lambda_1 \mid \lambda_2) &=& \sum_{i = 0}^3 \sum_{j = 0}^1 \sum_{k = 0}^3 \bar{\phi}_{i,j,k} \; e^{2 i \lambda_0 + 2 j \lambda_1 + 2 k \lambda_2} 
\>
with coefficients $\phi_{i,j,k}$ and $\bar{\phi}_{i,j,k}$ still undetermined. Although the partition function $\mathcal{Z}$ is also a polynomial according to Lemma \ref{pol1}, one can directly read it off from formulae \eqref{ZX1} and \eqref{ZX2}. In this way, we can restrict our attention to the functions $\mathcal{H}$ and $\bar{\mathcal{H}}$ in order to solve the $\gen{ZH}$-system. 

Now turning our attention to the $\gen{ZH}$-system, it is then convenient to eliminate the function $\mathcal{Z}$ from 
our problem. The latter task can then be simply accomplished through the identification of \eqref{ZX1} and \eqref{ZX2}.  Therefore, we are left with a reduced system of equations reading
\< \label{ZH2}
&& \frac{\Upsilon_0 (\lambda_0, \lambda_1) \; \bar{\Omega}_1 (\lambda_0, \lambda_1) }{\mathcal{W}(\lambda_0 , \lambda_1)} \mathcal{H}(\lambda_2 \mid \lambda_0 , \lambda_1) + \frac{\Upsilon_1 (\lambda_0, \lambda_1) \; \bar{\Omega}_1 (\lambda_0, \lambda_1) }{\mathcal{W}(\lambda_0 , \lambda_1)} \mathcal{H}(\lambda_2 \mid \lambda_1 , \lambda_0)  \nonumber \\
&& \quad - \; \frac{\bar{\Upsilon}_0 (\lambda_0, \lambda_1) \; \Omega_1 (\lambda_0, \lambda_1) }{\mathcal{W}(\lambda_0 , \lambda_1)} \bar{\mathcal{H}}(\lambda_1 , \lambda_0 \mid \lambda_2) - \frac{\bar{\Upsilon}_1 (\lambda_0, \lambda_1) \; \Omega_1 (\lambda_0, \lambda_1) }{\mathcal{W}(\lambda_0 , \lambda_1)} \bar{\mathcal{H}}(\lambda_0 , \lambda_1 \mid \lambda_2) = \nonumber \\
&& \frac{\Omega_0 (\lambda_1, \lambda_{\bar{0}}) \; \bar{\Upsilon}_1 (\lambda_1, \lambda_{\bar{0}}) }{\mathcal{W}(\lambda_1 , \lambda_{\bar{0}})} \bar{\mathcal{H}}(\lambda_1 , \lambda_{\bar{0}} \mid \lambda_2) + \frac{\Omega_0 (\lambda_1, \lambda_{\bar{0}}) \; \bar{\Upsilon}_0 (\lambda_1, \lambda_{\bar{0}}) }{\mathcal{W}(\lambda_1 , \lambda_{\bar{0}})} \bar{\mathcal{H}}(\lambda_{\bar{0}}, \lambda_1  \mid \lambda_2) \nonumber \\
&& \quad - \;  \frac{\bar{\Omega}_0 (\lambda_1, \lambda_{\bar{0}}) \; \Upsilon_1 (\lambda_1, \lambda_{\bar{0}}) }{\mathcal{W}(\lambda_1 , \lambda_{\bar{0}})} \mathcal{H}(\lambda_2 \mid \lambda_{\bar{0}} , \lambda_1 ) - \frac{\bar{\Omega}_0 (\lambda_1, \lambda_{\bar{0}}) \; \Upsilon_0 (\lambda_1, \lambda_{\bar{0}}) }{\mathcal{W}(\lambda_1 , \lambda_{\bar{0}})} \mathcal{H}(\lambda_2 \mid \lambda_1, \lambda_{\bar{0}} ) \nonumber \\ \nonumber \\
&& \bar{\omega}(\lambda_2) \; \bar{\mathcal{H}}(\lambda_0 , \lambda_2 \mid \lambda_1) = \mathcal{M}_0^{(1)} \omega(\lambda_0) \left[ \mathcal{N}_{0,1}^{(1)} \;  \mathcal{H} (\lambda_2 \mid \lambda_0 , \lambda_1) + \mathcal{N}_{0,2}^{(1)} \; \mathcal{H} (\lambda_1 \mid \lambda_0 , \lambda_2) \right] \nonumber \\
&& \qquad\qquad\qquad\qquad\qquad + \; \mathcal{M}_1^{(1)} \omega(\lambda_1) \left[ \mathcal{N}_{1,0}^{(1)} \; \mathcal{H} (\lambda_2 \mid \lambda_1 , \lambda_0) + \mathcal{N}_{1,2}^{(1)} \; \mathcal{H} (\lambda_0 \mid \lambda_1 , \lambda_2) \right] 
\>
for $L=2$. The latter then involves only the functions $\mathcal{H}$ and $\bar{\mathcal{H}}$.

The system of equations \eqref{ZX2} is linear and the use of expressions \eqref{PH} will consequently yield  linear algebraic equations for the coefficients $\phi_{i,j,k}$ and $\bar{\phi}_{i,j,k}$. In this way, the existence of unique trigonometric polynomials $\mathcal{H}$ and $\bar{\mathcal{H}}$ 
(up to an overall constant) solving \eqref{ZH2} will depend on having enough independent equations constraining our coefficients.  
Moreover, it is important to emphasize that \eqref{ZH2} is a system of equations on the variables $\lambda_0$, $\lambda_1$ and $\lambda_2$; although the coefficients in \eqref{ZH2} also depend on the inhomogeneity parameters $\mu_1$ and $\mu_2$. In this way, by setting $\mu_i$ to particular values we could only decrease the number of linearly independent equations while keeping the same number of coefficients $\phi_{i,j,k}$ and $\bar{\phi}_{i,j,k}$. Therefore, if we already find enough constraints for particular values of the anisotropy parameter, then it is certainly enough for generic values of the latter since we would have at least the same number of equations. 
We then proceed by fixing $\mu_i = 0$ for simplicity reasons and present the coefficients $\phi_{i,j,k}$ and $\bar{\phi}_{i,j,k}$ obtained from the resolution of \eqref{ZX2} in \Appref{app:COEF}.

\subsection{Case $L=3$} \label{sec:L3}

As for the case $L=3$ we proceed along the same lines employed in \Secref{sec:L2} for $L=2$. We then start by considering a reduced version of the 
$\gen{ZH}$-system obtained through the elimination of the partition function $\mathcal{Z}$. More precisely, we consider the equation resulting from the identification of \eqref{ZX1} and \eqref{ZX2}; in addition to the second relation in \eqref{hhb} for $L=3$. The latter equations read
\< \label{ZH3}
&& \frac{\Upsilon_0 (\lambda_0, \lambda_1) \; \bar{\Omega}_1 (\lambda_0, \lambda_1) }{\mathcal{W}(\lambda_0 , \lambda_1)} \mathcal{H}(\lambda_2, \lambda_3 \mid \lambda_0 , \lambda_1) + \frac{\Upsilon_1 (\lambda_0, \lambda_1) \; \bar{\Omega}_1 (\lambda_0, \lambda_1) }{\mathcal{W}(\lambda_0 , \lambda_1)} \mathcal{H}(\lambda_2 , \lambda_3 \mid \lambda_1 , \lambda_0)  \nonumber \\
&& \quad - \; \frac{\bar{\Upsilon}_0 (\lambda_0, \lambda_1) \; \Omega_1 (\lambda_0, \lambda_1) }{\mathcal{W}(\lambda_0 , \lambda_1)} \bar{\mathcal{H}}(\lambda_1 , \lambda_0 \mid \lambda_2 , \lambda_3) - \frac{\bar{\Upsilon}_1 (\lambda_0, \lambda_1) \; \Omega_1 (\lambda_0, \lambda_1) }{\mathcal{W}(\lambda_0 , \lambda_1)} \bar{\mathcal{H}}(\lambda_0 , \lambda_1 \mid \lambda_2 , \lambda_3) = \nonumber \\
&& \frac{\Omega_0 (\lambda_1, \lambda_{\bar{0}}) \; \bar{\Upsilon}_1 (\lambda_1, \lambda_{\bar{0}}) }{\mathcal{W}(\lambda_1 , \lambda_{\bar{0}})} \bar{\mathcal{H}}(\lambda_1 , \lambda_{\bar{0}} \mid \lambda_2 , \lambda_3) + \frac{\Omega_0 (\lambda_1, \lambda_{\bar{0}}) \; \bar{\Upsilon}_0 (\lambda_1, \lambda_{\bar{0}}) }{\mathcal{W}(\lambda_1 , \lambda_{\bar{0}})} \bar{\mathcal{H}}(\lambda_{\bar{0}}, \lambda_1  \mid \lambda_2 , \lambda_3) \nonumber \\
&& \quad - \;  \frac{\bar{\Omega}_0 (\lambda_1, \lambda_{\bar{0}}) \; \Upsilon_1 (\lambda_1, \lambda_{\bar{0}}) }{\mathcal{W}(\lambda_1 , \lambda_{\bar{0}})} \mathcal{H}(\lambda_2 , \lambda_3 \mid \lambda_{\bar{0}} , \lambda_1 ) - \frac{\bar{\Omega}_0 (\lambda_1, \lambda_{\bar{0}}) \; \Upsilon_0 (\lambda_1, \lambda_{\bar{0}}) }{\mathcal{W}(\lambda_1 , \lambda_{\bar{0}})} \mathcal{H}(\lambda_2 , \lambda_3 \mid \lambda_1, \lambda_{\bar{0}} ) \nonumber \\ \nonumber \\
&& \bar{\omega}(\lambda_3) \; \bar{\mathcal{H}}(\lambda_0 , \lambda_3 \mid \lambda_1, \lambda_2) = \nonumber \\ 
&& \mathcal{M}_0^{(2)} \omega(\lambda_0) \left[ \mathcal{N}_{0,1}^{(2)} \;  \mathcal{H} (\lambda_2, \lambda_3 \mid \lambda_0 , \lambda_1) + \mathcal{N}_{0,2}^{(2)} \; \mathcal{H} (\lambda_1, \lambda_3 \mid \lambda_0 , \lambda_2) + \mathcal{N}_{0,3}^{(2)} \; \mathcal{H} (\lambda_1, \lambda_2 \mid \lambda_0 , \lambda_3) \right] \nonumber \\
&& + \; \mathcal{M}_1^{(2)} \omega(\lambda_1) \left[ \mathcal{N}_{1,0}^{(2)} \;  \mathcal{H} (\lambda_2, \lambda_3 \mid \lambda_1 , \lambda_0) + \mathcal{N}_{1,2}^{(2)} \; \mathcal{H} (\lambda_0, \lambda_3 \mid \lambda_1 , \lambda_2) + \mathcal{N}_{1,3}^{(2)} \; \mathcal{H} (\lambda_0, \lambda_2 \mid \lambda_1 , \lambda_3) \right] \nonumber \\
&& + \; \mathcal{M}_2^{(2)} \omega(\lambda_2) \left[ \mathcal{N}_{2,0}^{(2)} \;  \mathcal{H} (\lambda_1, \lambda_3 \mid \lambda_2 , \lambda_0) + \mathcal{N}_{2,1}^{(2)} \; \mathcal{H} (\lambda_0, \lambda_3 \mid \lambda_2 , \lambda_1) + \mathcal{N}_{2,3}^{(2)} \; \mathcal{H} (\lambda_0, \lambda_1 \mid \lambda_2 , \lambda_3) \right] \nonumber \\
\>
and it is worth remarking the function $\mathcal{H}$ is symmetric on the first two arguments; while $\bar{\mathcal{H}}$ is symmetric on the last two arguments. 

Next we write 
\< \label{PH3}
\mathcal{H}(\lambda_0 , \lambda_1 \mid \lambda_2, \lambda_3) &=& \sum_{i = 0}^5 \sum_{j = 0}^5 \sum_{k = 0}^2 \sum_{l = 0}^5 \phi_{i,j,k,l} \; e^{2 i \lambda_0 + 2 j \lambda_1 + 2 k \lambda_2 + 2 l \lambda_3} \nonumber \\
\bar{\mathcal{H}}(\lambda_0 , \lambda_1 \mid \lambda_2 , \lambda_3) &=& \sum_{i = 0}^5 \sum_{j = 0}^2 \sum_{k = 0}^5 \sum_{l = 0}^5 \bar{\phi}_{i,j,k,l} \; e^{2 i \lambda_0 + 2 j \lambda_1 + 2 k \lambda_2 + 2 l \lambda_3} 
\>
in accordance with the polynomial structure discussed in \Secref{sec:POL}. Moreover, for our purposes here we can also set $\mu_j = 0$ using the same arguments discussed in \Secref{sec:L2} for the case $L=2$. In this way, the substitution of \eqref{PH3} in the system of functional equations \eqref{ZH3}
yields a system of linear algebraic equations for the coefficients $\phi_{i,j,k,l}$ and $\bar{\phi}_{i,j,k,l}$. The resolution of the latter shows all coefficients $\phi_{i,j,k,l}$ and $\bar{\phi}_{i,j,k,l}$ are fixed except for one. Due to the large amount of coefficients present in this case we have preferred not to display the solutions as we did for the case $L=2$. However, our results for the case $L=3$ also corroborates our claim that the $\gen{ZH}$-system is sufficient  to characterize the partition function $\mathcal{Z}$ up to an overall multiplicative factor.

\section{Concluding remarks} \label{sec:CONCL}

The main result of this paper is the system of functional equations formed by \eqref{zh}, \eqref{zhb} and \eqref{hhb} describing partition functions of two integrable nineteen-vertex models, namely the IK and FZ models, with three different types of domain-wall boundary conditions. We refer to such system of equations as $\gen{ZH}$-system and it provides a relation between the partition functions $\mathcal{Z}$, $\mathcal{F}$ and $\bar{\mathcal{F}}$ defined in \eqref{ZFF} in algebraic manner. 
A schematic representation of the $\gen{ZH}$-system can also be found in \Figref{fig:schema}. 

The Yang-Baxter algebra attached to nineteen-vertex models is the origin of the $\gen{ZH}$-system; and the derivation of the latter system of equations follows the AF method previously devised for six-vertex models. Although the idea employed here is essentially the same as the one used in the case of the six-vertex model; the derivation of the $\gen{ZH}$-system still encompasses some additional mechanisms.
For instance, one can notice the similar role played by the functions $\Lambda$ ($\bar{\Lambda}$) and $\omega$ ($\bar{\omega}$) defined in
\eqref{LOM} in the coefficients of \eqref{zh} and \eqref{zhb}. In the case of the six-vertex model we only have the presence of terms with the same origin as $\Lambda$ and $\bar{\Lambda}$; which are direct consequences of the weight-modules discussed in \Secref{sec:sing}. The terms $\omega$ and $\bar{\omega}$, in their turn, arises from the existence of simple zeroes as shown by Lemmas \ref{zeroF} and \ref{zeroFb}.

At first look the structure of the $\gen{ZH}$-system seems to be completely different from the functional equations obtained for the six-vertex model. However, there are still important similarities worth remarking. For instance, the relations between the functions
$\mathcal{H}$ and $\bar{\mathcal{H}}$ stated in Lemma \ref{HHb} can be regarded as a \emph{doubled} and \emph{inhomogeneous} version of the six-vertex model's equations. This is essentially due to the commutation relations \eqref{GEB}; which can be recognized as the same
relations appearing in the six-vertex model.

Although the partition functions $\mathcal{Z}$, $\mathcal{F}$ and $\bar{\mathcal{F}}$ satisfy the $\gen{ZH}$-system by construction, it is not \emph{a priori} clear if our system of equations is indeed capable of fixing the aforementioned quantities uniquely. In fact, we have not presented a rigorous proof of the latter property in this work and leave it as a conjecture supported by complementary results. For instance, in \Secref{sec:L1} through \Secref{sec:L3} we have discussed the explicit resolution of the $\gen{ZH}$-system for 
lattice lengths $L=1,2,3$; taking into account the polynomial structure expected from our partition functions.
This explicit analysis for small lattice lengths shows the $\gen{ZH}$-system is indeed capable of fixing our partition functions up to an overall multiplicative constant. The latter can then be fixed by evaluating any of our three partition functions at particular values of the spectral parameters. In \Secref{sec:Z0} we have then found the specialization $\mathcal{Z}(\mu_1, \mu_2, \dots, \mu_L)$ can be easily obtained and use \eqref{Z0} to fix the overall constant. The latter is not a fundamental quantity from the Statistical Mechanics perspective; but it becomes relevant, for instance, for possible applications in Enumerative Combinatorics similar to the counting of Alternating-Sign-Matrices \cite{Kuperberg_1996}.

Moreover, the explicit inspection of the $\gen{ZH}$-system for $L=2$ presented in \Secref{sec:L2} and \Appref{app:COEF} shows an interesting difference between the IK and FZ models considered in this work. For instance, one can notice several coefficients of the functions $\mathcal{H}$ and 
$\bar{\mathcal{H}}$ vanish for the FZ model, while they are all different from zero for the IK model. Such functions are building blocks of the partition
function $\mathcal{Z}$, as it can be seen from formulae \eqref{ZX1} and \eqref{ZX2}, and this feature might justify the possibility of expressing 
$\mathcal{Z}$ for the FZ model as a determinant according to the work \cite{Caradoc_2006}.
As far as the general solution of the $\gen{ZH}$-system is concerned, this problem has eluded us so far but we hope the methods put forward in \cite{Galleas_2012, Galleas_2013} and \cite{Galleas_2016a, Galleas_2016b} to shed some light into possible multiple contour integral or determinantal solutions.

\section{Acknowledgements} \label{sec:ACK}
The authors thank N. Beisert for discussions and comments.

\bibliographystyle{alpha}
\bibliography{references}

\appendix
%
%
%
%

\section{The sub-algebra $\mathscr{S}_{\mathcal{A}, \mathcal{B} , \mathcal{E}}$} \label{app:SUB}
 
 As for nineteen-vertex models described by the $\mathcal{R}$-matrix \eqref{bw}, the associated Yang-Baxter algebra comprises eighty-one commutation relations involving the entries of the monodromy matrix \eqref{abcd}. Here, however, we are interested only in the sub-algebra spanned by the generators
 $\mathcal{A}$, $\mathcal{B}$ and $\mathcal{E}$. We refer to that sub-algebra as $\mathscr{S}_{\mathcal{A}, \mathcal{B} , \mathcal{E}}$ and it consists
 of the following commutation rules:
\< \label{G1}
\mathcal{A}(\lambda_1) \mathcal{A}(\lambda_2) &=& \mathcal{A}(\lambda_2) \mathcal{A}(\lambda_1) \nonumber \\
\mathcal{A}(\lambda_1) \mathcal{B}(\lambda_2) &=& \frac{a(\lambda_2-\lambda_1)}{b(\lambda_2-\lambda_1)} \mathcal{B}(\lambda_2)  \mathcal{A}(\lambda_1) - \frac{c(\lambda_2-\lambda_1)}{b(\lambda_2-\lambda_1)} \mathcal{B}(\lambda_1)  \mathcal{A}(\lambda_2) \nonumber \\
\mathcal{A}(\lambda_1) \mathcal{E}(\lambda_2) &=& \frac{a(\lambda_2-\lambda_1)}{d_{3,3}(\lambda_2-\lambda_1)} \mathcal{E}(\lambda_2) \mathcal{A}(\lambda_1) - \frac{d_{1,3}(\lambda_2-\lambda_1)}{d_{3,3}(\lambda_2-\lambda_1)} \mathcal{E}(\lambda_1) \mathcal{A}(\lambda_2) \nonumber \\
&& \qquad \qquad \qquad\qquad \qquad \qquad \qquad - \; \frac{d_{2,3}(\lambda_2-\lambda_1)}{d_{3,3}(\lambda_2-\lambda_1)} \mathcal{B}(\lambda_1) \mathcal{B}(\lambda_2) 
\>
 
\< \label{G2}
\mathcal{B}(\lambda_1) \mathcal{A}(\lambda_2) &=& \frac{a(\lambda_2-\lambda_1)}{b(\lambda_2-\lambda_1)} \mathcal{A}(\lambda_2)  \mathcal{B}(\lambda_1) -\frac{\bar{c}(\lambda_2-\lambda_1)}{b(\lambda_2-\lambda_1)} \mathcal{A}(\lambda_1)  \mathcal{B}(\lambda_2) \nonumber \\
\mathcal{B}(\lambda_1) \mathcal{B}(\lambda_2) &=& \frac{a(\lambda_2-\lambda_1)}{d_{2,1}(\lambda_2-\lambda_1)} \mathcal{A}(\lambda_2)  \mathcal{E}(\lambda_1) - \frac{d_{3,1}(\lambda_2-\lambda_1)}{d_{2,1}(\lambda_2-\lambda_1)} \mathcal{A}(\lambda_1)  \mathcal{E}(\lambda_2) \nonumber \\
&& \qquad \qquad \qquad\qquad \qquad \qquad \qquad - \; \frac{d_{1,1}(\lambda_2-\lambda_1)}{d_{2,1}(\lambda_2-\lambda_1)} \mathcal{E}(\lambda_1) \mathcal{A}(\lambda_2) \nonumber \\
\mathcal{B}(\lambda_1) \mathcal{E}(\lambda_2) &=& \frac{a(\lambda_2-\lambda_1)}{b(\lambda_2-\lambda_1)} \mathcal{E}(\lambda_2)  \mathcal{B}(\lambda_1) -\frac{c(\lambda_2-\lambda_1)}{b(\lambda_2-\lambda_1)} \mathcal{E}(\lambda_1)  \mathcal{B}(\lambda_2) \nonumber \\
\>

\< \label{G3}
\mathcal{E}(\lambda_1) \mathcal{A}(\lambda_2) &=& \frac{a(\lambda_2-\lambda_1)}{d_{1,2}(\lambda_2-\lambda_1)} \mathcal{B}(\lambda_2) \mathcal{B}(\lambda_1) - \frac{d_{2,2}(\lambda_2-\lambda_1)}{d_{1,2}(\lambda_2-\lambda_1)} \mathcal{B}(\lambda_1) \mathcal{B}(\lambda_2) \nonumber \\
&& \qquad \qquad \qquad\qquad \qquad \qquad \qquad - \; \frac{d_{3,2}(\lambda_2-\lambda_1)}{d_{1,2}(\lambda_2-\lambda_1)} \mathcal{A}(\lambda_1) \mathcal{E}(\lambda_2) \nonumber \\
\mathcal{E}(\lambda_1) \mathcal{B}(\lambda_2) &=& \frac{a(\lambda_2-\lambda_1)}{b(\lambda_2-\lambda_1)} \mathcal{B}(\lambda_2)  \mathcal{E}(\lambda_1) -\frac{\bar{c}(\lambda_2-\lambda_1)}{b(\lambda_2-\lambda_1)} \mathcal{B}(\lambda_1)  \mathcal{E}(\lambda_2) \nonumber \\
\mathcal{E}(\lambda_1) \mathcal{E}(\lambda_2) &=& \mathcal{E}(\lambda_2) \mathcal{E}(\lambda_1)  \nonumber \\
\>

 \newpage

\section{Coefficients $\phi_{i,j,k}$ and $\bar{\phi}_{i,j,k}$} \label{app:COEF}

In this appendix we present the coefficients for the functions $\mathcal{H}$ and $\bar{\mathcal{H}}$ according to formulae \eqref{PH} obtained through the resolution of the $\gen{ZH}$-system for the case $L=2$ . In particular, we explicit the coefficients $\phi_{i,j,k}$ and $\bar{\phi}_{i,j,k}$ for the IK model in Tables \ref{T1} and \ref{T2} respectively. Tables \ref{T3} and \ref{T4} then contains respectively the results for $\phi_{i,j,k}$ and $\bar{\phi}_{i,j,k}$ associated to the FZ model.

\begin{table}[h]
\caption{Coefficients of the function $\mathcal{H}$ for the IK model with $L=2$.}
\label{T1}
\begin{center}
\begin{tabular}{cccccccc}
\multicolumn{4}{c}{\begin{tabular}{|c|c|c|c|}
\hline
i & j & k & $\phi_{i,j,k}/\phi_{0,0,0}$ \\ 
\hline & & & \\[-1em]
0 & 0 & 0 & 1 \\
0 & 0 & 1 & $-\frac{2 \left(2 q^4-2q-1\right)}{q^2 \left(q^2+1\right)\left(q^2+q+1\right)}$ \\
0 & 0 & 2 & $\frac{q^7+q^5-8 q^4-3q^3+q^2+q+1}{q^5 \left(q^2+1\right) \left(q^2+q+1\right)}$ \\
0 & 0 & 3 & $\frac{2 \left(q^4-q-1\right)}{q^5 \left(q^2+1\right) \left(q^2+q+1\right)}$ \\
1 & 0 & 0 & $-\frac{2 \left(q^4-2 q^3-2 q+1\right)}{(q-1) q\left(q^2+1\right) \left(q^2+q+1\right)}$ \\
1 & 0 & 1 & $-\frac{q^{10}-4 q^8+5q^7-q^6+17 q^5-13 q^4-3 q^3+q+1}{q^5 \left(q^5+q^3-q^2-1\right)}$ \\
1 & 0 & 2 & $-\frac{2 \left(q^7-8 q^6+7 q^5+q^4+5 q^3-2 q^2-3 q+1\right)}{q^6 \left(q^5+q^3-q^2-1\right)}$ \\
1 & 0 & 3 & $-\frac{q^8-4 q^7+2q^6-8 q^5+2 q^4+4 q^3-1}{q^8 \left(q^5+q^3-q^2-1\right)}$ \\
2 & 0 & 0 & $-\frac{q^4+q^3+6 q^2+q+1}{q^3 \left(q^2+1\right) \left(q^2+q+1\right)}$ \\
2 & 0 & 1 & $\frac{2 \left(6q^4-q^3-3 q-1\right)}{q^5 \left(q^2+1\right) \left(q^2+q+1\right)}$ \\
2 & 0 & 2 & $-\frac{5 q^7-4 q^6+q^5-12 q^4-3 q^3+q^2+q+1}{q^8 \left(q^2+1\right) \left(q^2+q+1\right)}$ \\
2 & 0 & 3 & $-\frac{2 \left(q^6+2 q^4-q-1\right)}{q^8 \left(q^2+1\right) \left(q^2+q+1\right)}$ \\
3 & 0 & 0 & $\frac{2}{q^3-q^6}$ \\
3 & 0 & 1 & $\frac{4 q^3-q^2-1}{q^6 \left(q^3-1\right)}$ \\
3 & 0 & 2 & $-\frac{2 \left(q^3-q^2-1\right)}{q^6 \left(q^3-1\right)}$ \\
3 & 0 & 3 & $-\frac{q^2+1}{q^6 \left(q^3-1\right)}$ \\
\hline
\end{tabular}} &
\multicolumn{4}{c}{\begin{tabular}{|c|c|c|c|}
\hline
i & j & k & $\phi_{i,j,k}/\phi_{0,0,0}$ \\
\hline & & & \\[-1em]
0 & 1 & 0 & $-\frac{2}{q^2+1}$ \\
0 & 1 & 1 & $\frac{q^7+q^6+5 q^5-3q^4-4 q^3-3 q^2+1}{q^5 \left(q^2+1\right) \left(q^2+q+1\right)}$ \\
0 & 1 & 2 & $-\frac{2 \left(q^5-q^4-q^3-4 q^2+2\right)}{q^5 \left(q^2+1\right) \left(q^2+q+1\right)}$ \\
0 & 1 & 3 & $-\frac{3 q^4-q^3-2q^2-q-1}{q^7 \left(q^2+1\right) \left(q^2+q+1\right)}$ \\
1 & 1 & 0 & $\frac{q^6+q^4-8 q^3+q^2+1}{q^3 \left(q^5+q^3-q^2-1\right)}$ \\
1 & 1 & 1 & $-\frac{2 \left(2 q^6-7 q^5+q^4-4 q^3+8 q^2+q-3\right)}{q^5 \left(q^5+q^3-q^2-1\right)}$ \\
1 & 1 & 2 & $\frac{q^{11}-5 q^{10}+2 q^9-12 q^8+22 q^7+q^6-3 q^4-5 q^3+4 q^2-1}{q^{10} \left(q^5+q^3-q^2-1\right)}$ \\
1 & 1 & 3 & $\frac{2 \left(q^6-4 q^5+q^4-2 q^3+q^2+2 q-1\right)}{q^8 \left(q^5+q^3-q^2-1\right)}$ \\
2 & 1 & 0 & $\frac{2 \left(2 q^2+q+2\right)}{q^3 \left(q^2+1\right) \left(q^2+q+1\right)}$ \\
2 & 1 & 1 & $-\frac{5 q^7+q^6+9 q^5-7 q^4-4 q^3-3 q^2+1}{q^8 \left(q^2+1\right) \left(q^2+q+1\right)}$ \\
2 & 1 & 2 & $\frac{2 \left(q^7+2 q^5-5 q^4-q^3-4 q^2+2\right)}{q^8 \left(q^2+1\right) \left(q^2+q+1\right)}$ \\
2 & 1 & 3 & $\frac{4 q^4-q^2-q-1}{q^{10} \left(q^2+q+1\right)}$ \\
3 & 1 & 0 & $\frac{q^2+1}{q^5 \left(q^3-1\right)}$ \\
3 & 1 & 1 & $-\frac{2 \left(q^3+q-1\right)}{q^6 \left(q^3-1\right)}$ \\
3 & 1 & 2 & $\frac{q^3+q-4}{q^6 \left(q^3-1\right)}$ \\
3 & 1 & 3 & $\frac{2}{q^6 \left(q^3-1\right)}$ \\
\hline
\end{tabular}} 
\end{tabular}
\end{center}
\end{table}

\newpage
\begin{table}[h]
\caption{Coefficients of the function $\bar{\mathcal{H}}$ for the IK model with $L=2$.}
\label{T2}
\begin{center}
\begin{tabular}{cccccccc}
\multicolumn{4}{c}{\begin{tabular}{|c|c|c|c|}
\hline & & & \\[-1em]
i & j & k & $\bar{\phi}_{i,j,k}/\phi_{0,0,0}$ \\ 
\hline & & & \\[-1em]
0 & 0 & 0 & $q^4$ \\
0 & 0 & 1 & $-\frac{2 q^3 \left(q^4-2 q^3-2 q+1\right)}{(q-1) \left(q^2+1\right) \left(q^2+q+1\right)}$ \\
0 & 0 & 2 & $-\frac{q \left(q^4+q^3+6 q^2+q+1\right)}{\left(q^2+1\right) \left(q^2+q+1\right)}$ \\
0 & 0 & 3 & $-\frac{2 q}{(q-1) \left(q^2+q+1\right)}$ \\
1 & 0 & 0 & $-\frac{2 q^3 \left(q^4+2 q^3-2\right)}{\left(q^2+1\right) \left(q^2+q+1\right)}$ \\
1 & 0 & 1 & $\frac{q^{10}+q^9-3 q^7-13 q^6+17 q^5-q^4+5 q^3-4 q^2+1}{(q-1) q^2 \left(q^2+1\right) \left(q^2+q+1\right)}$ \\
1 & 0 & 2 & $\frac{2 \left(q^4+3 q^3+q-6\right)}{\left(q^2+1\right) \left(q^2+q+1\right)}$ \\
1 & 0 & 3 & $\frac{q^3+q-4}{(q-1) q^2 \left(q^2+q+1\right)}$ \\
2 & 0 & 0 & $\frac{q^7+q^6+q^5-3 q^4-8 q^3+q^2+1}{\left(q^2+1\right) \left(q^2+q+1\right)}$ \\
2 & 0 & 1 & $-\frac{2 \left(q^7-3 q^6-2 q^5+5 q^4+q^3+7 q^2-8 q+1\right)}{(q-1) q\left(q^2+1\right) \left(q^2+q+1\right)}$ \\
2 & 0 & 2 & $-\frac{q^7+q^6+q^5-3 q^4-12 q^3+q^2-4 q+5}{q^3 \left(q^2+1\right) \left(q^2+q+1\right)}$ \\
2 & 0 & 3 & $\frac{2 \left(q^3+q-1\right)}{(q-1) q^5 \left(q^2+q+1\right)}$ \\
3 & 0 & 0 & $\frac{2 \left(q^4+q^3-1\right)}{\left(q^2+1\right) \left(q^2+q+1\right)}$ \\
3 & 0 & 1 & $-\frac{q^8-4 q^5-2 q^4+8 q^3-2 q^2+4 q-1}{(q-1) q^3 \left(q^2+1\right) \left(q^2+q+1\right)}$ \\
3 & 0 & 2 & $-\frac{2\left(q^6+q^5-2 q^2-1\right)}{q^5 \left(q^2+1\right) \left(q^2+q+1\right)}$ \\
3 & 0 & 3 & $\frac{q^2+1}{(q-1) q^7 \left(q^2+q+1\right)}$ \\
\hline
\end{tabular}} &
\multicolumn{4}{c}{\begin{tabular}{|c|c|c|c|}
\hline & & & \\[-1em]
i & j & k & $\bar{\phi}_{i,j,k}/\phi_{0,0,0}$ \\
\hline & & & \\[-1em]
 0 & 1 & 0 & $\frac{2 q^3}{q^2+1}$ \\
 0 & 1 & 1 & $-\frac{q^6+q^4-8 q^3+q^2+1}{(q-1) \left(q^2+1\right) \left(q^2+q+1\right)}$ \\
 0 & 1 & 2 & $-\frac{2 \left(2 q^2+q+2\right)}{\left(q^2+1\right) \left(q^2+q+1\right)}$ \\
 0 & 1 & 3 & $-\frac{q^2+1}{(q-1) q^2 \left(q^2+q+1\right)}$ \\
 1 & 1 & 0 & $\frac{q^7-3 q^5-4 q^4-3 q^3+5 q^2+q+1}{\left(q^2+1\right) \left(q^2+q+1\right)}$ \\
 1 & 1 & 1 & $\frac{2 \left(3 q^6-q^5-8 q^4+4 q^3-q^2+7 q-2\right)}{(q-1) q \left(q^2+1\right) \left(q^2+q+1\right)}$ \\
 1 & 1 & 2 & $-\frac{q^7-3 q^5-4 q^4-7 q^3+9 q^2+q+5}{q^3 \left(q^2+1\right) \left(q^2+q+1\right)}$ \\
 1 & 1 & 3 & $\frac{2 \left(q^3-q^2-1\right)}{(q-1) q^5 \left(q^2+q+1\right)}$ \\
 2 & 1 & 0 & $\frac{2 \left(2 q^5-4 q^3-q^2-q+1\right)}{q \left(q^2+1\right) \left(q^2+q+1\right)}$ \\
 2 & 1 & 1 & $\frac{q^{11}-4 q^9+5 q^8+3 q^7-q^5-22 q^4+12 q^3-2 q^2+5 q-1}{(q-1) q^4 \left(q^2+1\right) \left(q^2+q+1\right)}$ \\
 2 & 1 & 2 & $-\frac{2 \left(2 q^7-4 q^5-q^4-5 q^3+2 q^2+1\right)}{q^6 \left(q^2+1\right)\left(q^2+q+1\right)}$ \\
 2 & 1 & 3 & $\frac{4 q^3-q^2-1}{(q-1) q^8 \left(q^2+q+1\right)}$ \\
 3 & 1 & 0 & $\frac{q^4+q^3+2 q^2+q-3}{q \left(q^2+1\right) \left(q^2+q+1\right)}$ \\
 3 & 1 & 1 & $-\frac{2 \left(q^6-2 q^5-q^4+2 q^3-q^2+4 q-1\right)}{(q-1) q^4 \left(q^2+1\right) \left(q^2+q+1\right)}$ \\
 3 & 1 & 2 & $-\frac{q^4+q^3+q^2-4}{q^6 \left(q^2+q+1\right)}$ \\
 3 & 1 & 3 & $\frac{2}{(q-1) q^8 \left(q^2+q+1\right)}$ \\
 \hline
\end{tabular}} 
\end{tabular}

\end{center}
\end{table}

\newpage
\begin{table}[h]
\caption{Coefficients of the function $\mathcal{H}$ for the FZ model with $L=2$.}
\label{T3}
\begin{center}
\begin{tabular}{cccccccc}
\multicolumn{4}{c}{\begin{tabular}{|c|c|c|c|}
\hline
i & j & k & $\phi_{i,j,k}/\phi_{0,0,0}$ \\ 
\hline & & & \\[-1em]
0 & 0 & 0 & $1$ \\
0 & 0 & 1 & $-\frac{2 (q+2)}{q^2+1}$ \\
0 & 0 & 2 & $\frac{(2 q+1) \left(q^2+4 q+1\right)}{q \left(q^2+1\right) \left(q^2+q+1\right)}$ \\
0 & 0 & 3 & $-\frac{2}{q \left(q^2+1\right)}$ \\
1 & 0 & 0 & $-\frac{2 (q+1)}{q^2+1}$ \\
1 & 0 & 1 & $-\frac{(q+1) \left(q^4-q^3-8 q^2-9 q-1\right)}{q \left(q^2+1\right) \left(q^2+q+1\right)}$ \\
1 & 0 & 2 & $\frac{2 (q+1) \left(q^3-2 q^2-5 q-3\right)}{q \left(q^2+1\right) \left(q^2+q+1\right)}$ \\
1 & 0 & 3 & $\frac{(q+1) \left(q^2+4 q+1\right)}{q^2 \left(q^2+1\right) \left(q^2+q+1\right)}$ \\
2 & 0 & 0 & $\frac{q^2+4 q+1}{\left(q^2+1\right) \left(q^2+q+1\right)}$ \\
2 & 0 & 1 & $\frac{2 \left(q^4-q^3-5 q^2-3 q-1\right)}{q \left(q^2+1\right) \left(q^2+q+1\right)}$ \\
2 & 0 & 2 & $-\frac{q^5-4 q^3-9 q^2-5 q-1}{q^2 \left(q^2+1\right) \left(q^2+q+1\right)}$ \\
2 & 0 & 3 & $-\frac{2}{q^2 \left(q^2+1\right)}$ \\
3 & 0 & 0 & $0$ \\
3 & 0 & 1 & $0$ \\
3 & 0 & 2 & $0$ \\
3 & 0 & 3 & $0$ \\
\hline
\end{tabular}} &
\multicolumn{4}{c}{\begin{tabular}{|c|c|c|c|}
\hline
i & j & k & $\phi_{i,j,k}/\phi_{0,0,0}$ \\
\hline & & & \\[-1em]
0 & 1 & 0 & $-\frac{2}{q^2+1}$ \\
0 & 1 & 1 & $\frac{q^5+5 q^4+9 q^3+4 q^2-1}{q^3 \left(q^2+1\right) \left(q^2+q+1\right)}$ \\
0 & 1 & 2 & $-\frac{2 \left(q^4+3 q^3+5 q^2+q-1\right)}{q^3 \left(q^2+1\right) \left(q^2+q+1\right)}$ \\
0 & 1 & 3 & $\frac{q^2+4 q+1}{q^2 \left(q^2+1\right) \left(q^2+q+1\right)}$ \\
1 & 1 & 0 & $\frac{(q+1) \left(q^2+4 q+1\right)}{q \left(q^2+1\right) \left(q^2+q+1\right)}$ \\
1 & 1 & 1 & $-\frac{2 (q+1) \left(3 q^3+5 q^2+2 q-1\right)}{q^3 \left(q^2+1\right) \left(q^2+q+1\right)}$ \\
1 & 1 & 2 & $\frac{(q+1) \left(q^4+9 q^3+8 q^2+q-1\right)}{q^4 \left(q^2+1\right) \left(q^2+q+1\right)}$ \\
1 & 1 & 3 & $-\frac{2 (q+1)}{q^3 \left(q^2+1\right)}$ \\
2 & 1 & 0 & $-\frac{2}{q \left(q^2+1\right)}$ \\
2 & 1 & 1 & $\frac{(q+2) \left(q^2+4 q+1\right)}{q^2 \left(q^2+1\right)\left(q^2+q+1\right)}$ \\
2 & 1 & 2 & $-\frac{2 (2 q+1)}{q^3 \left(q^2+1\right)}$ \\
2 & 1 & 3 & $\frac{1}{q^4}$ \\
3 & 1 & 0 & $0$ \\
3 & 1 & 1 & $0$ \\
3 & 1 & 2 & $0$ \\
3 & 1 & 3 & $0$ \\
\hline
\end{tabular}} 
\end{tabular}
\end{center}
\end{table}

\newpage
\begin{table}[h]
\caption{Coefficients of the function $\bar{\mathcal{H}}$ for the FZ model with $L=2$.}
\label{T4}
\begin{center}
\begin{tabular}{cccccccc}
\multicolumn{4}{c}{\begin{tabular}{|c|c|c|c|}
\hline & & & \\[-1em]
i & j & k & $\bar{\phi}_{i,j,k}/\phi_{0,0,0}$ \\ 
\hline & & & \\[-1em]
0 & 0 & 0 & $1$ \\
0 & 0 & 1 & $-\frac{2 (q+1)}{q^2+1}$ \\
0 & 0 & 2 & $\frac{q^2+4 q+1}{q^4+q^3+2 q^2+q+1}$ \\
0 & 0 & 3 & $0$ \\
1 & 0 & 0 & $-\frac{4 q+2}{q^2+1}$ \\
1 & 0 & 1 & $\frac{(q+1) \left(q^4+9 q^3+8 q^2+q-1\right)}{q^2 \left(q^2+1\right) \left(q^2+q+1\right)}$ \\
1 & 0 & 2 & $-\frac{2 \left(q^4+3 q^3+5 q^2+q-1\right)}{q^2 \left(q^2+1\right) \left(q^2+q+1\right)}$ \\
1 & 0 & 3 & $0$ \\
2 & 0 & 0 & $\frac{(q+2) \left(q^2+4 q+1\right)}{\left(q^2+1\right) \left(q^2+q+1\right)}$ \\
2 & 0 & 1 & $-\frac{2 (q+1) \left(3 q^3+5 q^2+2 q-1\right)}{q^2 \left(q^2+1\right) \left(q^2+q+1\right)}$ \\
2 & 0 & 2 & $\frac{q^5+5 q^4+9 q^3+4 q^2-1}{q^3 \left(q^2+1\right) \left(q^2+q+1\right)}$ \\
2 & 0 & 3 & $0$ \\
3 & 0 & 0 & $-\frac{2}{q^2+1}$ \\
3 & 0 & 1 & $\frac{q^3+5 q^2+5 q+1}{q^5+q^4+2 q^3+q^2+q}$ \\
3 & 0 & 2 & $-\frac{2}{q^3+q}$ \\
3 & 0 & 3 & $0$ \\
\hline
\end{tabular}} &
\multicolumn{4}{c}{\begin{tabular}{|c|c|c|c|}
\hline & & & \\[-1em]
i & j & k & $\bar{\phi}_{i,j,k}/\phi_{0,0,0}$ \\
\hline & & & \\[-1em]
0 & 1 & 0 & $-\frac{2 q}{q^2+1}$ \\
0 & 1 & 1 & $\frac{(q+1) \left(q^2+4 q+1\right)}{\left(q^2+1\right) \left(q^2+q+1\right)}$ \\
0 & 1 & 2 & $-\frac{2}{q^2+1}$ \\
0 & 1 & 3 & $0$ \\
1 & 1 & 0 & $\frac{-q^5+4 q^3+9 q^2+5 q+1}{q^4+q^3+2 q^2+q+1}$ \\
1 & 1 & 1 & $\frac{2 (q+1) \left(q^3-2 q^2-5 q-3\right)}{\left(q^2+1\right) \left(q^2+q+1\right)}$ \\
1 & 1 & 2 & $\frac{2 q^3+9 q^2+6 q+1}{q^5+q^4+2 q^3+q^2+q}$ \\
1 & 1 & 3 & $0$ \\
2 & 1 & 0 & $\frac{2 \left(q^4-q^3-5 q^2-3 q-1\right)}{\left(q^2+1\right) \left(q^2+q+1\right)}$ \\
2 & 1 & 1 & $\frac{-q^5+9 q^3+17 q^2+10 q+1}{q^5+q^4+2 q^3+q^2+q}$ \\
2 & 1 & 2 & $-\frac{2 (q+2)}{q^3+q}$ \\
2 & 1 & 3 & $0$ \\
3 & 1 & 0 & $\frac{q^2+4 q+1}{q^4+q^3+2 q^2+q+1}$ \\
3 & 1 & 1 & $-\frac{2 (q+1)}{q^3+q}$ \\
3 & 1 & 2 & $\frac{1}{q^2}$ \\
3 & 1 & 3 & $0$ \\
\hline
\end{tabular}} 
\end{tabular}

\end{center}
\end{table}

\end{document}